\documentclass[10pt,english]{article}
\usepackage[T1]{fontenc}
\usepackage[latin9]{inputenc}
\usepackage[a4paper]{geometry}
\usepackage{color}
\usepackage{babel}
\usepackage{enumitem}
\usepackage{float}
\usepackage{mathtools}
\usepackage{amsmath}
\usepackage{amsthm}
\usepackage{amssymb}
\usepackage{graphicx}
\usepackage[square,numbers]{natbib}
\usepackage[unicode=true,
 bookmarks=true,bookmarksnumbered=false,bookmarksopen=false,
 breaklinks=false,pdfborder={0 0 0},pdfborderstyle={},backref=false,colorlinks=true]
 {hyperref}

 \usepackage{marginnote}
 \usepackage{mparhack}

\makeatletter

\floatstyle{ruled}
\newfloat{algorithm}{tbp}{loa}
\providecommand{\algorithmname}{Algorithm}
\floatname{algorithm}{\protect\algorithmname}
\usepackage{algorithm,algorithmic}
\theoremstyle{plain}
\newtheorem{thm}{\protect\theoremname}
\theoremstyle{plain}
\newtheorem{lem}{\protect\lemmaname}
\theoremstyle{plain}
\newtheorem{prop}{\protect\propositionname}
\theoremstyle{definition}

\theoremstyle{plain}
\newtheorem{assump}{}

\def\Y{\mathbf{Y}}
\def\E{\mathbf{E}}
\def\W{\mathbf{W}}
\def\I{\mathbf{I}}
\def\Phib{\boldsymbol{\Phi}}
\def\Zc{\mathcal{Z}}

\usepackage{wrapfig}

\makeatother

\usepackage{listings}
\providecommand{\definitionname}{Definition}
\providecommand{\lemmaname}{Lemma}
\providecommand{\propositionname}{Proposition}
\providecommand{\theoremname}{Theorem}

\usepackage{authblk}

\begin{document}
\title{Statistical exploration of the Manifold Hypothesis}

\author[a]{Nick Whiteley\thanks{{\it Corresponding author:} Prof. Nick Whiteley, School of Mathematics, Woodland Rd, Bristol BS8 1UG, United Kingdom; email: {\tt nick.whiteley@bristol.ac.uk}}} 
\author[a]{Annie Gray}
\author[b]{Patrick Rubin-Delanchy}
\affil[a]{School of Mathematics, University of Bristol, United Kingdom}
\affil[b]{School of Mathematics, University of Edinburgh, United Kingdom}
\maketitle

\begin{abstract}
The Manifold Hypothesis is a widely accepted tenet of Machine Learning which asserts that nominally high-dimensional data are in fact concentrated near a low-dimensional manifold, embedded in high-dimensional space. This phenomenon is observed empirically in many real world situations, has led to development of a wide range of statistical methods in the last few decades, and has been suggested as a key factor in the success of modern AI technologies. We show that rich and sometimes intricate manifold structure in data can emerge from a generic and remarkably simple statistical model --- the Latent Metric Model --- via elementary concepts such as latent variables, correlation and stationarity. This establishes a general statistical explanation for why the Manifold Hypothesis seems to hold in so many situations. Informed by the Latent Metric Model we derive procedures to discover and interpret the geometry of high-dimensional data, and explore hypotheses about the data generating mechanism. These procedures operate under minimal assumptions and make use of well known graph-analytic algorithms.  
\end{abstract}

% \blfootnote{\noindent {\it Corresponding author:} Prof. Nick Whiteley. Email:{\tt nick.whiteley@bristol.ac.uk}. 
% Address: School of Mathematics, Woodland Rd, Bristol BS8 1UG}
\section{Introduction}\label{sec:Introduction}
The manifold hypothesis is a widely accepted tenet of Machine Learning which posits that \citep{cayton2005algorithms}:

``\emph{...the dimensionality of many data sets is
only artificially high; though each data point consists of perhaps
thousands of features, it may be described as a function of only a
few underlying parameters. That is, the data points are actually samples
from a low-dimensional manifold that is embedded in a high-dimensional
space}''.

This phenomenon has impacted a wide range of methods and algorithms. Presence of manifold structure is the premise of manifold estimation and testing \citep{fefferman2016testing,genovese2012minimax,genovese2012manifold}, nonlinear dimension reduction techniques \cite{roweis2000nonlinear,tenenbaum2000global,hinton2002stochastic,belkin2003laplacian,weinberger2004learning,van2008visualizing,mcinnes2018umap}, intrinsic dimension estimation \citep{kegl2002intrinsic,levina2004maximum,hein2005intrinsic,carter2009local}, and regression and classification techniques specially adapted to settings in which covariates are valued on manifolds \cite{bickel2007local,aswani2011regression,cheng2013local,yang2016bayesian,lin2019extrinsic,niu2019intrinsic}. Assumptions that data are concentrated near low-dimensional topological or
geometric structures underpin clustering techniques and topological data analysis \citep{edelsbrunner2008persistent,niyogi2008finding,carlsson2009topology,balakrishnan2012minimax,chazal2013persistence,chazal2021introduction}. Some nonparametric techniques, such as nearest neighbour or tree-based regression methods, function without manifold structure necessarily being present, but  benefit significantly when it is there, since their convergence rates depends on intrinsic rather than ambient dimension of covariates \citep{kpotufe2011k,kpotufe2012tree}. It has been proved that deep neural networks exhibit a similar property  \cite{nakada2020adaptive}.  More broadly, the presence of manifold structure has been suggested as a key factor in the success of deep learning methods \citep{brahma2015deep}. Assumptions that data lie on a low-dimensional manifold embedded in high-dimensional space are central to very recent practical and theoretical developments in generative modelling in Artificial Intelligence, especially diffusion models \citep{song2019generative,song2020score,ho2020denoising,de2022convergence,de2022riemannian,stanczuk2022your,chung2022improving,NEURIPS2022_e8fb575e,he2023manifold,diffusion-fields}.

% In collective terms, such assumptions are sometimes loosely called
% ``the manifold hypothesis'' \citep{bengio2013representation}, the
% spirit of which is captured in the following quote from \citet{cayton2005algorithms}:

% ``\emph{...the idea that the dimensionality of many data sets is
% only artificially high; though each data point consists of perhaps
% thousands of features, it may be described as a function of only a
% few underlying parameters. That is, the data points are actually samples
% from a low-dimensional manifold that is embedded in a high-dimensional
% space}''. 

\begin{figure}
\begin{centering}
\includegraphics[width=0.52\columnwidth]{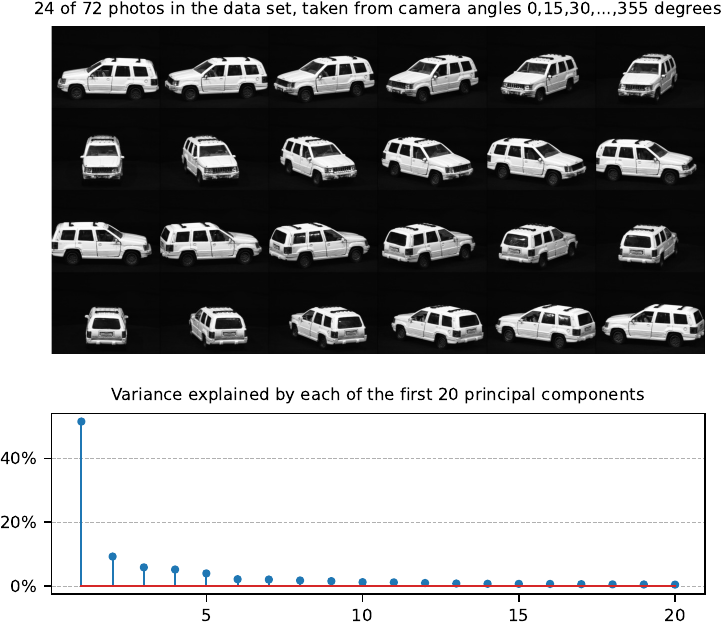}\hfill
\includegraphics[width=0.45\columnwidth]{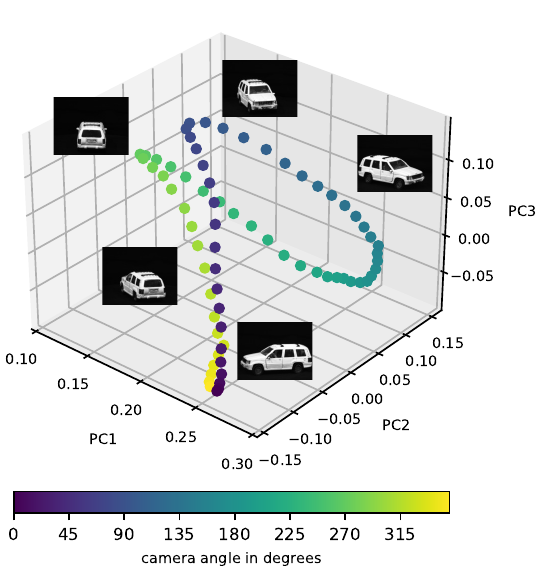}
\end{centering}
\caption{A collection of images reduced in dimension using PCA.\label{fig:car:intro}}

\end{figure}

Why might manifold structure be present in data? In some situations, such as image analysis, an intuitive albeit heuristic
explanation can be given in terms of the physical mechanism which generated the data (see e.g., \citet{pless2009survey}
for a review of manifold estimation in this context). Figure
\ref{fig:car:intro} shows 24 grayscale images of a car, a subset
of $n=75$ images from \citep{geusebroek2005amsterdam}, taken from angles
$0,5,10,\ldots,355$ degrees around the circumference of a circle.
Each image is of resolution $384\times288$ pixels and so can be
represented as a vector of length $p=110592$. However, at least intuitively,
we can account for the variation across the collection of images using
far fewer dimensions, in terms of the position of the camera in the
three-dimensional space of the world around us. Figure \ref{fig:car:intro}
shows the result of using principal component analysis (PCA) to reduce dimension, upon which we make the following observations.

The first 20 principal components account for 91.5\% of the total
variance, suggesting that the data are concentrated somewhere in a
low-dimensional linear subspace of $\mathbb{R}^{110592}$. The first three dimensions  --- the coordinates
of the data with respect to the eigenvectors associated with the three
largest eigenvalues --- exhibit points around a loop which is somewhat irregular
in shape but resembles the circle of camera positions, subject to
deformation by bending and twisting. The points appear roughly equally spaced around the loop, like the camera positions which are equally spaced at intervals of $5$ degrees around a circle.

% resembling
% the geometry of the camera positions which are equally spaced at intervals
% of $5$ degrees around a circle.
%\begin{enumerate}
%\item 
%\item 
%\item 
%\end{enumerate}

Evidently reducing the dimension of these image data by PCA allows us to access
some of the geometric structure of the data generating mechanism, but questions
remain. We have chosen to plot the first three dimensions for ease
of visualisation, is this a ``good'' choice? What might the other dimensions convey? What explains the precise
shape of the loop and the spacing of the points along it, relative to the underlying circle of camera positions? 

In other situations, embedded topological and geometric structure may appear
in different forms and have different interpretations.  Figure \ref{fig:planaria_intro}
shows two approaches to visualising  expression levels of $p=5821$ genes measured across $n=5000$ individual cells from adult planarians, a type of flatworm. In the field of single-cell transcriptomics --- as set out in the 2018 Science paper \citep{plass2018cell} --- such data offer the possibility of discovering the cell lineage tree of an entire animal: the aim is to find out if the data reflect the tree-structured process by which stem cells differentiate into a variety of distinct cell types.  The data are preprocessed in the same way as the original paper \citep{plass2018cell}, using the Python package Scanpy \citep{wolf2018scanpy}.

% , which implements a standard pipeline involving filtering, normalisation, log-transform, scaling, following the methods of  \citep{plass2018cell}. 

The left plot in figure \ref{fig:planaria_intro} shows the result of
dimension reduction from $5821$ to $2$ using PCA. The
right plot shows the result of first reducing from $5821$
to $14$ dimensions using PCA, followed by reduction to $2$ dimensions
using $t$-SNE \citep{van2008visualizing}, a very popular nonlinear dimension
reduction method which finds a lower dimensional representation of a
 data set by minimising a particular measure of distortion
of pairwise distances. We used the default $t$-SNE parameter settings
in \verb|scikit-learn| \cite{scikit-learn}. In both plots, the points are coloured by cell type, but neither PCA nor $t$-SNE have access to this information. Similarly to figure \ref{fig:car:intro}, it is evident from figure \ref{fig:planaria_intro} that performing some form of dimension reduction allows us to access structure underlying the data, albeit in the form of discrete cell types rather than the geometry of camera positions. In figure \ref{fig:car:intro}, using only PCA
to reduce dimension was enough to make this structure visible. However, in figure \ref{fig:planaria_intro}, using only PCA and reducing to $2$ dimensions, distinct cell types are not clearly separated, whereas PCA down to $14$ dimensions followed
by $t$-SNE seems to be more effective. The \textit{t}-SNE  visualisation  hints at the presence of tree structure underlying the data, with some areas having branch-like arms originating at the central point cloud, but other lineages lack clarity or seem to be disconnected. Could we combine methods differently to obtain a clearer picture?

\begin{figure}[ht!]
\includegraphics[width=1\columnwidth]{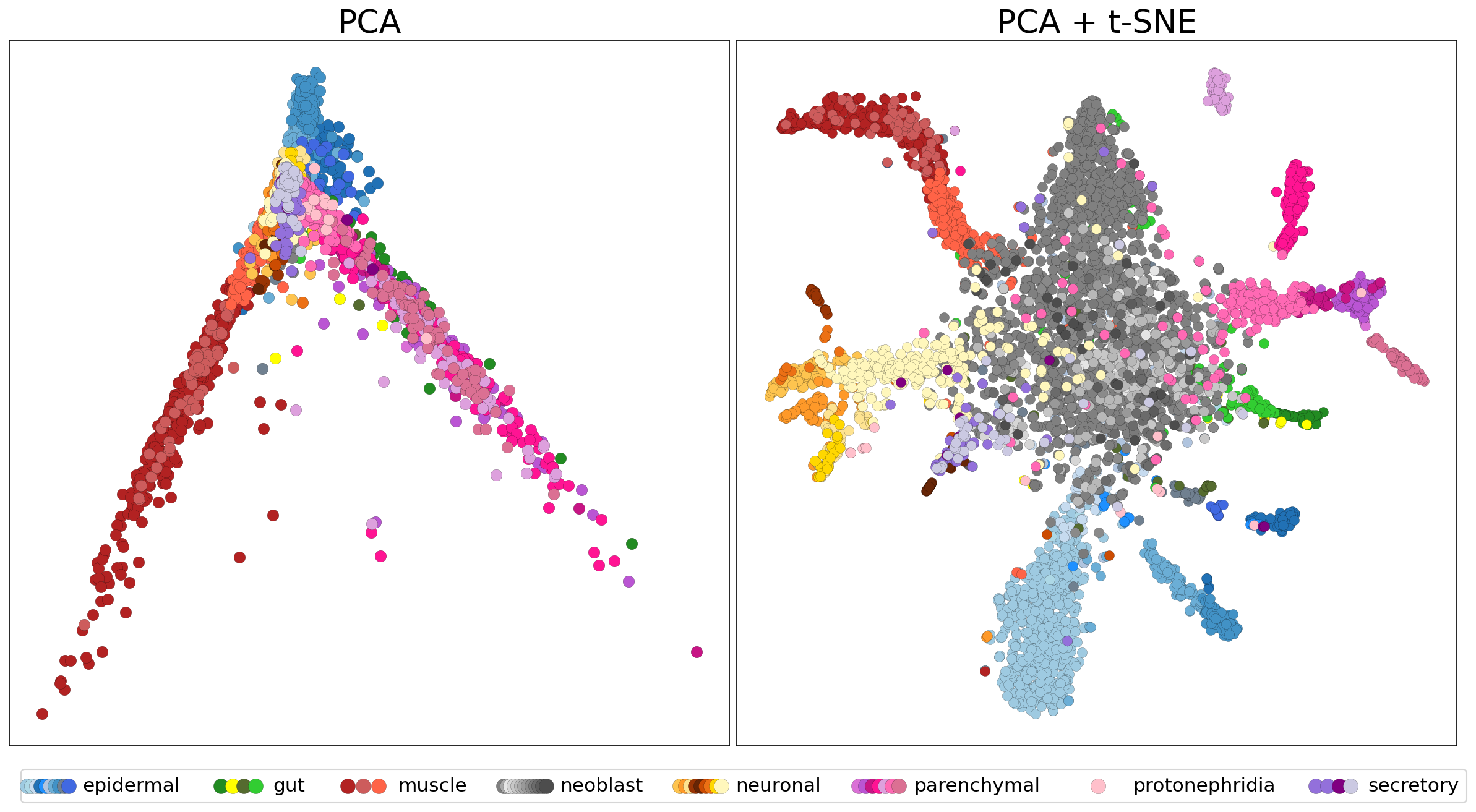}\caption{Planaria example. Left: first $2$ dimensions of the PCA embedding.
Right: representation of the data in $2$ dimensions obtained by first
reducing to 14 dimensions using PCA, then applying $t$-SNE.\label{fig:planaria_intro} }
\end{figure}

These examples illustrate just some of the ways in which underlying structure can manifest itself in embedded topological and geometric patterns in data. Many other examples can be found: in genomics, where genotyping DNA sites has revealed striking geographic patterns \citep{novembre2008genes,lao2008correlation, diaz2019umap}; neuroscience, where simultaneous recordings from Grid cells have been shown to exhibit toroidal structure seemingly independent of behavioural tasks \citep{gardner2022toroidal}; as well as manifold structure in data from wireless sensor networks \citep{patwari2004manifold}, visual speech recognition \citep{bregler1995nonlinear},
drug discovery \citep{reutlinger2012nonlinear}, RNA sequencing \citep{moon2018manifold},
and human motion synthesis \citep{lee2006human}.

In this work we put forward a perspective that embedded topological and geometric structure in data can be explained as a general statistical phenomenon, without reference to physical properties or other domain-specific details of the data generating mechanism. 

%\section{Overview and contributions}
\paragraph{Main contributions.} Our first main contribution is to propose a simple and generic statistical model which produces hidden, low-dimensional manifold structure in high-dimensional data, thus providing \emph{a statistical justification for the manifold hypothesis}.

Our second main contribution is to describe how this hidden manifold relates to a true latent domain defined by the model, explaining, for example, why the points in the right panel of figure~\ref{fig:car:intro} are not in a perfect circle, as the camera positions are, but still form a loop. More precisely, we give mild conditions under which the relationship between the manifold and the latent domain is a homeomorphism (a topological equivalence), and stronger conditions under which it becomes an isometry (a metric equivalence).

Our third main contribution is to show that our model and theory enable a combination of \emph{simple or well-known techniques} to be given a new, model-based interpretation and put to use in exploring hypotheses and uncovering information about the latent domain and broader data generating mechanism. Given data vectors $\Y_1,\ldots,\Y_n\in\mathbb{R}^p$, we rationalise the following workflow:
\begin{enumerate}%[leftmargin=0.4cm]
\item Dimension selection, using $\Y_1,\ldots,\Y_n$ to choose number of dimensions $\hat{r}$.
\item Linear dimension reduction of $\Y_1,\ldots,\Y_n$ by PCA, resulting in an $\hat{r}$-dimensional embedding, $\zeta_1,\ldots,\zeta_n$.
\item Spherical projection of the embedding,  setting $\zeta_i^{\mathrm{sp}}\coloneqq \zeta_i/\|\zeta_i\|$, $i=1,\ldots,n$
\item Nearest neighbour graph construction from $\zeta_1^{\mathrm{sp}},\ldots,\zeta_n^{\mathrm{sp}}$.
\item Analysis and visualisation of the nearest neighbour graph, e.g., shortest paths, minimum spanning tree, topology. %, in order to discover topological and geometric information, and explore hypotheses about the data generating mechanism.
\end{enumerate}
For step 1., we introduce a new Wasserstein distance-based dimension selection method. 

The remainder of this article is structured as follows. In Section~\ref{sec:Model} we introduce the Latent Metric Model, and the associated manifold $\mathcal{M}$, which arises as a consequence of correlation over a latent domain $\Zc$. In section~\ref{sec:connecting} we describe how this manifold structure hides in the data and how the manifold relates to $\Zc$. We establish a representation formula  (proposition~\ref{prop:Phi_W_expansion}) uncovering the perhaps surprising fact that, under the Latent Metric Model, data are noisy, random projections of points in $\mathcal{M}$.  Standard statistical concepts, such as stationarity, give rise to striking geometric relationships between $\mathcal{M}$ and $\Zc$, such as isometry. In section~\ref{sec:methodology} we develop theory and methodology supporting the workflow above, elucidating the benefits of applying PCA (theorem~\ref{thm:consistency_summary}), proposing a new dimension selection method, and more. In Section~\ref{sec:examples} we demonstrate the workflow on real data, revisiting the image and transcriptomics data from section \ref{sec:Introduction}, as well as a temperature time series example. The key new feature of these analyses is that we can explore manifold hypotheses grounded in a statistical model.  In section \ref{sec:discussion} we draw together conclusions and discuss connections to the literature, including geometric representation of high-dimensional data, PCA in high dimensions, Gaussian process latent variable models, nonlinear dimension reduction, and exploratory data analysis.

\section{The Latent Metric Model}\label{sec:Model}

The Latent Metric Model (LMM) is constructed from three independent sources of randomness.

\paragraph*{Latent Variables.}
$Z_{1},\ldots,Z_{n}$ are independent and identically distributed random elements of a metric space
$(\mathcal{Z},d_{\mathcal{Z}})$, that is $\mathcal{Z}$ is a set, and $d_{\mathcal{Z}}(\cdot,\cdot)$ is a distance function  on $\mathcal{Z}$. It is assumed that the metric space $(\mathcal{Z},d_{\mathcal{Z}})$ is compact, and $Z_{1},\ldots,Z_{n}$ are distributed
according to a Borel probability measure $\mu$ supported on $\mathcal{Z}$.

\paragraph*{Random Functions.}
$X_{1},\ldots,X_{p}$ are random $\mathbb{R}$-valued functions, each with domain $\mathcal{Z}$. That is, for each $z\in\mathcal{Z}$
and $j=1,\ldots,p$, $X_{j}(z)$ is an $\mathbb{R}$-valued random
variable. It is not assumed that $X_{1},\ldots,X_{p}$ are identically distributed, but it is assumed that $\mathbb{E}[|X_{j}(z)|^{2}]<\infty$, for all $j=1,\ldots,p$ and $z\in\Zc$.

\paragraph*{Noise.}
$\mathbf{E}\in\mathbb{R}^{n\times p}$ is a matrix of random variables
whose elements are each zero-mean and unit-variance. The columns of
$\mathbf{E}$ are assumed independent and elements in distinct rows
of $\mathbf{E}$ are assumed pairwise uncorrelated.

 \vspace{\baselineskip}
\noindent The data matrix $\mathbf{Y}\in\mathbb{R}^{n\times p}$ is defined by:
\begin{equation}
\mathbf{Y}_{ij}\coloneqq X_{j}(Z_{i})+\sigma\mathbf{E}_{ij}\label{eq:model}
\end{equation}
for some $\sigma\geq0$. It will sometimes be convenient to think of data vectors $\Y_1,\ldots,\Y_n\in\mathbb{R}^p$ such that $[\Y_1|\cdots|\Y_n]^\top \equiv \Y$, so $\mathbf{Y}_{ij}$  is the $j$th element of $\Y_i$. Similarly we shall write noise vectors $[\E_1|\cdots|\E_n]^\top\equiv\E$.

We call:
\begin{equation}
f(z,z^{\prime})\coloneqq\frac{1}{p}\sum_{j=1}^{p}\mathbb{E}[X_{j}(z)X_{j}(z^{\prime})]\label{eq:f_defn}
\end{equation}
the \emph{mean correlation kernel} associated with the LMM. The following assumption is taken to hold throughout the paper without further mention.
\begin{assump}\label{ass:cont_covar}
For each $j=1,\ldots,p$, $\mathbb{E}[X_{j}(z)X_{j}(z^{\prime})]$
is a continuous function of $(z,z^{\prime})\in\mathcal{Z}\times\mathcal{Z}$.
\end{assump}
Assumption \ref{ass:cont_covar} implies $f(z,z^\prime)$ is continuous in $z,z^\prime$, and by a generalisation of Mercer's theorem \citep[Thm 4.49]{mercer1909functions,steinwart2008support} given in section \ref{sec:Supporting-results-for-Model},
when \ref{ass:cont_covar} holds there exists a countable collection
of non-negative real numbers $(\lambda_{k}^{f})_{k\geq1}$, $\lambda_{1}^{f}\geq\lambda_{2}^{f}\geq\cdots$,
and a sequence of functions $(u_{k}^{f})_{k\geq1}$ which are orthonormal
in $L_{2}(\mu)$ such that 
\begin{equation}
f(z,z^{\prime})=\sum_{k=1}^{\infty}\lambda_{k}^{f}u_{k}^{f}(z)u_{k}^{f}(z^{\prime})=\left\langle \phi(z),\phi(z^{\prime})\right\rangle_{\ell_2} ,\label{eq:mercers_theorem_intro}
\end{equation} 
where the series converges absolutely and uniformly in $z,z^\prime$. The inner product in \eqref{eq:mercers_theorem_intro} is $\langle x,x^\prime\rangle_{\ell_2} \coloneqq \sum_{k=1}^\infty x_k x_k^\prime$, between infinitely long vectors $x=[x_1\,x_2\,\cdots]^\top$ belonging to $\ell_2\coloneqq\{x\in\mathbb{R}^{\mathbb{N}}: \|x\|_{\ell_2} <\infty\}$, where $\|x\|_{\ell_2}\coloneqq\left( \sum_{k=1}^\infty |x_k|^2\right)^{1/2} = \langle x,x \rangle_{\ell_2}^{1/2}$. The function  $\phi:\Zc\to\ell_2$ is the ``feature map'':
\begin{equation}
\phi(z)\coloneqq\left[(\lambda_{1}^{f})^{1/2} u_{1}^{f}(z)\;\;(\lambda_{2}^{f})^{1/2}u_{2}^{f}(z)\;\;\cdots\;\;\right]^{\top},\label{eq:phi_defn}
\end{equation}
with range we denote
\begin{equation}
\mathcal{M}\coloneqq\left\{\phi(z);z\in\Zc\right\}.\label{eq:M_defn}
\end{equation} 
  The set $\mathcal{M}$ can be checked to be a subset of $\ell_2$ using $\|\phi(z)\|^2_{\ell_2}=f(z,z)$, and using the compactness of $\mathcal{Z}$ and the continuity assumption  \ref{ass:cont_covar} which imply  $\sup_{z\in\Zc}f(z,z)<\infty$.

%Our simulated examples should not be construed to be assumptions

% \noindent We stress the perspective here that $f$ and $\phi$ are defined implicitly by the ingredients of the LMM.

We denote by $r$ the rank of $f$, that is the largest $k\geq1$ such that $\lambda_k^f>0$, with $r\coloneqq \infty$ if $\lambda_k^f>0$ for all $k\geq1$. When $r<\infty$ we abuse notation slightly by writing
 \begin{equation}
\phi(z)\coloneqq\left[(\lambda_{1}^{f})^{1/2} u_{1}^{f}(z)\;\;(\lambda_{2}^{f})^{1/2}u_{2}^{f}(z)\;\;\cdots\;\;(\lambda_{r}^{f})^{1/2} u_{r}^{f}(z)\right]^{\top}.
\end{equation}

We stress two points. First, the central purpose of the LMM is to explain and describe manifold structure in data as a general statistical phenomenon. The breadth of this objective necessitates a flexible modelling paradigm and, except when considering examples, we do not make specific distributional or functional assumptions, such as Gaussianity. The assumptions in this paper, involve more general concepts, such as continuity, smoothness or stationarity. Second, we stress the perspective here that $f$ and $\phi$ are derived quantities, defined implicitly by the ingredients of the LMM, rather than model parameters or hyperparameters whose values need to be chosen.

\section{Connecting statistical and geometric properties of the LMM}\label{sec:connecting}
In this section we explain how  statistical properties of the LMM allow us to connect the geometry of the data vectors, $\Y_1,\ldots,\Y_n$, which can be thought of as a point cloud in $\mathbb{R}^p$, to the structure of $\mathcal{M}$, and in turn the latent metric space $\Zc$. This is important for two reasons. Firstly, it shows how manifold structure in data emerges from elementary  statistical properties of the LMM, thus clarifying in what sense and why the Manifold Hypothesis holds. Secondly, it forms the basis for data analysis procedures we detail in section \ref{sec:methodology}. We proceed in four main steps:
\begin{itemize}[leftmargin=0.4cm]
\item Section \ref{subsec:intro_Phi_W} shows how inner-products between data vectors, say $\Y_i,\Y_j$, relate to inner products between $\phi(Z_i),\phi(Z_j)$. Since $\phi(Z_1),\ldots,\phi(Z_n)$ are i.i.d. and valued in $\mathcal{M}$, recall \eqref{eq:M_defn}, this gives our first indication that the geometry of the point cloud $\Y_1,\ldots,\Y_n$ will reflect the shape of $\mathcal{M}$.  
\item Section \ref{subsec:intro_manifold} shows that under a simple distinguishability assumption, the feature map $\phi$ is a homeomorphism.  Informally, this means we can think of $\mathcal{M}$ as being equivalent to $\Zc$ up to some continuous, invertible distortion such as bending, twisting or stretching. Formally, we can say $\mathcal{M}$ is a topological manifold.
\item Section \ref{sec:isometry} shows that when $\Zc$ is a subset of Euclidean space, conditions closely related to weak stationarity of the random function $X_j$ imply $\phi$ is an \emph{isometry}. This means a very special form of geometric relationship holds between $\mathcal{M}$ and $\Zc$, in which distances between points in $\Zc$, say $Z_i$ and $Z_j$, are faithfully represented by distances measured \emph{along the manifold} $\mathcal{M}$ between $\phi(Z_i)$ and $\phi(Z_j)$, rather than by straight-line distances of the form $\|\phi(Z_i)-\phi(Z_j)\|_{\ell_2}$.
\item Section \ref{sec:smoothness} shows that if the kernel is sufficiently smooth,  most of the structure of $\mathcal{M}$ is captured in a low-dimensional subspace. This hints towards the potential effectiveness of PCA (step 2 in the workflow) for manifold exploration. 
\end{itemize}
Remarkably, we shall draw the conclusions in the second and third points above without any explicit knowledge of the eigenvalues and eigenfunctions which appear in the definition of $\phi$, and which thus define $\mathcal{M}$.

\subsection{Relating data inner products to feature map inner products}\label{subsec:intro_Phi_W}

We have not made any assumptions about the functional form of $z\mapsto X_j(z)$,  $j=1,\ldots,p$, in the LMM, other than \ref{ass:cont_covar}. Nevertheless, the following proposition shows that a linear relationship holds between  $\Y_i-\sigma\E_i$ and $\phi(Z_i)$.
\begin{prop}\label{prop:Phi_W_expansion} Assume \ref{ass:cont_covar}. Then 
under the LMM with $r\in\{1,2,\ldots,\}\cup\{\infty\}$, the matrix $\W\in\mathbb{R}^{p\times r}$ with elements
\begin{equation}\label{eq:W_jk_defn}
\W_{jk}\coloneqq \frac{1}{(p\lambda_k^f)^{1/2}}\int_{\Zc}X_j(z)u_k^f(z)\mu(\mathrm{d}z)
\end{equation}
satisfies
\begin{equation}\label{eq:Phi_W_expansion}
\mathbf{Y}_i \stackrel{m.s.}{=}  p^{1/2}\W \phi(Z_i) + \sigma \E_i,\quad i=1,\ldots,n,\quad\qquad \mathbb{E}[\mathbf{W}^\top\mathbf{W}]=\mathbf{I}_r,
\end{equation}
where $\I_r$ is the identity matrix with $r$ rows and columns.
\end{prop}
\noindent The qualification ``$\stackrel{m.s.}{=}$'' in \eqref{eq:Phi_W_expansion} indicates that the infinite summations constituting the matrix-vector product $\W\phi(Z_i)$ in the case $r=\infty$ converge in the mean-square sense. The proof of proposition \ref{prop:Phi_W_expansion}, in appendix \ref{sec:Proofs-and-supporting_for_manifold}, entails a generalised form of Karhunen-Lo\`{e}ve expansion of $X_{1},\ldots,X_{p}$.  

The identities in \eqref{eq:Phi_W_expansion} can be interpreted as meaning that $p^{-1/2}\Y_i$ is a noisy, random projection of $\phi(Z_i)$. 
Indeed we can use \eqref{eq:Phi_W_expansion} together with the defining properties of the LMM in section \ref{sec:Model} to describe the behaviour of the inner-product between $\Y_i,\Y_j\in\mathbb{R}^p$ when the randomness in $X_1,\ldots,X_p$ and $\E$ is averaged out:
\begin{align}
\frac{1}{p}\mathbb{E}[\langle\Y_i,\Y_j\rangle |Z_i,Z_j] &= \langle\phi(Z_i),\mathbb{E}[\W^\top\W]\phi(Z_j)\rangle_{\ell_2} +0+0+\sigma^2 \frac{1}{p}\mathbb{E}[\langle\E_i,\E_j\rangle]\nonumber\\
&=\langle\phi(Z_i),\phi(Z_j)\rangle_{\ell_2} + \sigma^2\I[i=j]\label{eq:Y_inner_prod_cond_exp}.
\end{align}

The quantity $p^{-1}\langle\Y_i,\Y_j\rangle$ is an arithmetic mean of $p$ random variables. If, conditionally on $Z_i,Z_j$, the summands in $p^{-1}\langle\Y_i,\Y_j\rangle$ are weakly dependent and have moments bounded uniformly in $p$, then by a law of large numbers argument $p^{-1}\langle\Y_i,\Y_j\rangle$ will be close to its conditional expectation \eqref{eq:Y_inner_prod_cond_exp} with high probability when $p$ is large (see proposition \ref{prop:mixing} in appendix \ref{sec:Proofs-and-supporting_for_manifold} for details). Moreover, write $\W\equiv[\W_1|\cdots|\W_p]^\top$ and note from \eqref{eq:W_jk_defn} that the only randomness in $\W_j$ arises from $X_j$. So if $X_1, \ldots, X_p$ were assumed weakly dependent,  $\W_1, \ldots, \W_p$ would be too. Then, again assuming moments bounded uniformly in $p$, by a law of large numbers argument the sum $\sum_{j=1}^p\W_j \W_j^\top$ will be close to its expectation with high probability when $p$ is large, i.e., 
$$\W^\top\W = \sum_{j=1}^p\W_j \W_j^\top \approx \mathbb{E}[\W^\top\W] = \mathbf{I}_r.$$
We therefore conclude that, subject to suitable weak dependence and moment conditions, \begin{equation}\label{eq:inner_prod_convg}
\left|p^{-1}\langle\Y_i,\Y_j\rangle -  \langle\phi(Z_i),\phi(Z_j)\rangle_{\ell_2} - \sigma^2\I[i=j]\right|\to 0, \quad \text{as} \quad p\to\infty,
\end{equation}
in probability. In this sense the geometry of the collection of high-dimensional data vectors $\Y_1,\ldots,\Y_n$ reflects that of $\phi(Z_1),\ldots,\phi(Z_n)$, subject to some distortion depending on the noise level $\sigma$. Moreover, if \eqref{eq:inner_prod_convg} holds, then $|p^{-1}\|\Y_i-\Y_j\|^2 -  \|\phi(Z_i)-\phi(Z_j)\|_{\ell_2}^2 -2\sigma^2|\to 0$ as $p\to\infty$.  
% This stands in contrast to the behaviour established in \cite{hall2005geometric} that  $p^{-1}\|\Y_i-\Y_j\|^2\to\text{const.}$ when $\Y_i$ and $\Y_j$ are independent and identically distributed. 
 
In section \ref{subsec:PCA-scores}, we shall complement the above reasoning with theorem \ref{thm:consistency_summary} which shows that when the noise level $\sigma$ is fixed, and $n\to\infty$ and $p/n\to\infty$ simultaneously, using PCA to reduce dimension of $\Y_1,\ldots,\Y_n$ allows $\phi(Z_1),\ldots,\phi(Z_n)$ to be recovered, up to an orthogonal transformation.

\subsection{Relating distinguishability of latent variables to homeomorphism}\label{subsec:intro_manifold}
A \emph{homeomorphism} between two metric spaces is a mapping which is continuous, bijective and has a continuous inverse. If such a mapping exists the two metric spaces are said to be homeomorphic, or \emph{topologically equivalent}. To develop some intuition, one can think about the case in which the metric spaces in question are subsets of the three dimensional Euclidean world around us. In this situation mappings which qualify as homeomorphisms include transformations of shape by bending, twisting, stretching and folding, but not cutting, puncturing or joining \citep{bing1960}. Topological equivalence implies the two metric spaces in question must exhibit the same number of connected components, the same number of 1-dimensional loops and more generally the same number of $k$-dimensional ``holes'' as each other. Detecting such features using data is the purpose of persistent homology methods within the field of Topological Data Analysis \cite{carlsson2009topology,chazal2021introduction}. But there is more to a topological structure than its homology; for example, in the transcriptomics application (introduction and Section~\ref{subsec:ex_planaria}), the hypothesized underlying structure has interesting, `tree-like', topology but no interesting homology. 

We shall now see that, with only a little more structure added to the LMM, $\phi$ is homeomorphism between $\Zc$ and $\mathcal{M}$, where the distance on $\mathcal{M}$ is $\|\cdot-\cdot\|_{\ell_2}$. The first requirement, continuity of $\phi$, means that  $d_{\Zc}(z, z^\prime)\to 0$ implies  $\|\phi(z) - \phi(z^\prime)\|_{\ell_2}\to 0$. This holds due to the identities:
\begin{align*}
\|\phi(z)-\phi(z^\prime)\|_{\ell_2}^2 &=\|\phi(z)\|_{\ell_2}^2 +\|\phi(z^\prime)\|_{\ell_2}^2 -2\langle\phi(z),\phi(z^\prime)\rangle_{\ell_2}\\
& = f(z,z)+f(z^\prime,z^\prime)-2f(z,z^\prime),
\end{align*}
combined with continuity of $f$ under \ref{ass:cont_covar}. By its definition, $\phi:\Zc\to\mathcal{M}$ is automatically surjective, and if $\phi$ is one-to-one, its inverse is automatically continuous due to a general result in the theory of metric spaces \citep[Prop. 13.26]{sutherland2009introduction} concerning the inverse of a continuous mapping with compact domain. The question of whether or not $\phi$ is a homeomorphism thus reduces to whether or not it is one-to-one. Consider the following assumption.
\begin{assump}\label{ass:injective}
For each $z,z^{\prime}\in\mathcal{Z}$ such that $z\neq z^{\prime}$, $\sum_{j=1}^p\mathbb{E}[|X_j(z)-X_j(z^\prime)|^2]>0$.
\end{assump}
% \begin{assump}\label{ass:injective}
% For each $z,z^{\prime}\in\mathcal{Z}$ such that $z\neq z^{\prime}$,
% there exists $\xi\in\mathcal{Z}$ such that $f(z,\xi)\neq f(z^{\prime},\xi)$.
% \end{assump}
\begin{prop}
\label{prop:homeomorphism} Assume \ref{ass:cont_covar}. Then 
 $\phi:\Zc\to\mathcal{M}$ is a homeomorphism if and only if \ref{ass:injective} holds.
\end{prop}
\noindent Assumption \ref{ass:injective} can be interpreted as a ``distinguishability'' condition, requiring that points in $\Zc$ can be distinguished from each other in terms of the random functions $X_j$. This distinguishability can also be stated in terms of the kernel $f$; in 
the proof of proposition \ref{prop:homeomorphism} in appendix \ref{sec:Proofs-and-supporting_for_manifold}, we show that \ref{ass:injective} is equivalent to the condition: for each $z,z^{\prime}\in\mathcal{Z}$ such that $z\neq z^{\prime}$,
 there exists $\xi\in\mathcal{Z}$ such that $f(z,\xi)\neq f(z^{\prime},\xi)$.

The term \emph{topological manifold} conventionally means some set such that each point in that set has a neighbourhood which is homeomorphic to some subset of Euclidean space. We note that the relationship between $\mathcal{M}$ and $\Zc$ is of a similar nature, except that $\mathcal{M}$ is globally rather than only locally homeomorphic to $\Zc$, and the metric space $\Zc$ need not be Euclidean.  Putting these differences aside, we shall simply call $\mathcal{M}$ a manifold from now on.

When $\mathcal{M}$ and $\Zc$ are homeomorphic, they must have the same \emph{covering dimension}---see \citep[Ch.3]{pears_dimension_ttheory} for background---this is an abstract topological notion dimension, which generalises the usual notion of dimension of Euclidean space. In this sense, we can say that when $\Zc$ is low-dimensional, $\mathcal{M}$ is low-dimensional too.

\subsection{Relating stationarity to isometry}\label{sec:isometry}
Weak stationarity of any one of the random functions $X_j$ in the LMM would mean that:
\begin{itemize}[leftmargin=0.4cm]
\item $\mathbb{E}[X_j(z)]$ is constant in $z$, and 
\item $\mathbb{E}[(X_j(z)-\mathbb{E}[X_j(z)])(X_j(z^\prime)-\mathbb{E}[X_j(z^\prime)])]$ is a function only of distance between $z$ and $z^\prime$.
\end{itemize}
If all the random functions $X_1,\ldots,X_p$ were to have this property, it would follow from the definition of $f$ in \eqref{eq:f_defn} that  $f(z,z^\prime)$ must also be a function only of distance between $z$ and $z^\prime$. We shall now see that this leads to an \emph{isometric} relationship between $\mathcal{M}$ and $\Zc$.  To define isometry, it's convenient to work in the following setting:
\begin{assump}\label{ass:diffeo_Rd_alt} $\Zc$ is a compact subset of $\mathbb{R}^d$, and there exists a continuous path in $\Zc$ of finite length between any two points in $\Zc$.
\end{assump}
\noindent The precise mathematical definition of a path and its length  are given in appendix \ref{subsec:isometry_proofs}. In the setting of \ref{ass:diffeo_Rd_alt}, we denote by $d^{\mathrm{geo}}_{\mathcal{Z}}(z,z^\prime)$ the shortest path length, or geodesic distance, in $\mathcal{Z}$. This is the infimum of the lengths of all paths in $\mathcal{Z}$ with end-points $z,z^\prime$  (see appendix \ref{subsec:isometry_proofs} for details). If $\mathcal{Z}$ is convex,  the shortest path between two points is a straight line and $d^{\mathrm{geo}}_{\mathcal{Z}}(z,z^\prime)=\|z-z^\prime\|_{\mathbb{R}^d}$.  For $x,x^\prime\in\mathcal{M}$, the shortest path length, or geodesic distance, in $\mathcal{M}$ is denoted $d^{\mathrm{geo}}_{\mathcal{M}}(x,x^\prime)$ and defined analogously to $d^{\mathrm{geo}}_{\mathcal{Z}}(z,z^\prime)$. Even when $\Zc$ is convex, in general $\mathcal{M}$ is not convex and $d^{\mathrm{geo}}_{\mathcal{M}}(x,x^\prime)$ is not equal to the straight-line distance $\|x-x^\prime\|_{\ell_2}$. 

We shall say \emph{isometry} holds between $\Zc$ and $\mathcal{M}$ if
\begin{equation}\label{eq:isometry-definition}
d^{\mathrm{geo}}_{\mathcal{M}}(\phi(z),\phi(z^\prime) )= d^{\mathrm{geo}}_{\mathcal{Z}}(z,z^\prime),\qquad\forall z,z^\prime\in\Zc. 
\end{equation}
Compared to homeomorphism, this isometry condition imposes more of a constraint on the relationship between $\Zc$ and $\mathcal{M}$. One can interpret isometry as allowing $\phi$ to transform $\Zc$ into $\mathcal{M}$ by bending, but not by stretching or compressing, since that would violate the equality of shortest path lengths.

The following proposition shows that isometry holds up to a scaling constant when, for $z,z^\prime$ close to each other, $f(z,z^\prime)$ depends only on the Euclidean distance between $z$ and $z^\prime$. In contrast, weak-sense stationarity involves the more stringent requirement that such dependence holds \emph{for all} $z,z^\prime$. Define  $\mathcal{D}\coloneqq\{(z,z);z \in\Zc\}\subset \Zc\times\Zc$.
\begin{prop}\label{prop:sq_euc_kernels_alt} Assume \ref{ass:injective} and \ref{ass:diffeo_Rd_alt}.
If $f(z,z^{\prime})=g(\|z-z^{\prime}\|_{\mathbb{R}^{d}}^{2})$ for all $z,z^\prime$
in an open neighbourhood of $\mathcal{D}$ where $g$
is twice continuously differentiable and  $g^{\prime}(0)<0$, then
\begin{equation}\label{eq:euc_isometry_alt}
d_{\mathcal{M}}^{\mathrm{geo}}(\phi(z),\phi(z^\prime)) = \sqrt{-2g^{\prime}(0)} d_{\Zc}^{\mathrm{geo}}(z,z^\prime).
\end{equation}
\end{prop}
\noindent The following proposition complements proposition \ref{prop:sq_euc_kernels_alt} by addressing the special case in which $\Zc$ is a sphere. 
\begin{prop}\label{prop:inner_prod_kernels_alt}
Assume \ref{ass:injective}. If $\Zc=\{z\in\mathbb{R}^{d}:\|z\|_{\mathbb{R}^d}=1\}$ and $f(z,z^\prime)=g(\langle z,z^\prime\rangle_{\mathbb{R}^d})$ for all $z,z^\prime$ in an open neighbourhood of $\mathcal{D}$ where $g$
is twice continuously differentiable and $g^{\prime}(1)>0$, then
\begin{equation}\label{eq:sphere_isometry_alt}
d_{\mathcal{M}}^{\mathrm{geo}}(\phi(z),\phi(z^\prime)) = \sqrt{g^{\prime}(1)} d_{\Zc}^{\mathrm{geo}}(z,z^\prime).
\end{equation}
\end{prop}
\noindent The proofs of propositions \ref{prop:sq_euc_kernels_alt} and \ref{prop:inner_prod_kernels_alt} are at the end of appendix \ref{subsec:isometry_proofs}.

\subsection{Relating smoothness to concentration within a low-dimensional subspace}\label{sec:smoothness}
When the latent domain $\Zc$ is a subset of $\mathbb{R}^d$, we will say that $f$ is smooth if it can be expressed as the restriction of a smooth function on $\mathbb{R}^d \times \mathbb{R}^d$ to $\Zc \times \Zc$. How smooth $f$ is affects how much of the manifold $\mathcal{M}$ we can capture using only the first few coordinates. For some $s < r$, consider the truncated map
\[\phi_{s}(z)\coloneqq\left[(\lambda_{1}^{f})^{1/2} u_{1}^{f}(z)\;\;\cdots\;\;(\lambda_{s}^{f})^{1/2} u_{s}^{f}(z) \;\; 0  \;\;\cdots\;\;\right]^{\top}.\]
The eigenvalues give us a measure of how well $\mathcal{M}_{s}  \coloneqq  \phi_{s}(\Zc)$ approximates $\mathcal{M}$ through the mean square error
\begin{equation}\mathbb{E}[\|\phi(Z_i) - \phi_{s}(Z_i)\|^2_{\ell_2}] =\sum_{k>s}\lambda_k^f \mathbb{E}\left[|u_k^f(Z_i)|^2\right]= \sum_{k > s} \lambda^f_k.\label{eq:phi_s_error}
\end{equation}
The rate of decay of the eigenvalues is known to be related to the smoothness of the kernel \cite{takhanov2023speed} and so, under smoothness assumptions, \eqref{eq:phi_s_error} tells us that for $s$ suitably large the first $s$ coordinates of $\phi$ can provide a good approximation to $\mathcal{M}$, even if $r = \infty$. 
When $s\leq p$, such smoothness also implies each  vector $\mathbf{Y}_i$ will be concentrated within the (at most) $s$-dimensional subspace of $\mathbb{R}^p$  spanned by the first $s$ columns of $\mathbf{W}$. Indeed recalling the identity $\mathbf{Y}_i \stackrel{m.s.}{=}p^{1/2}\mathbf{W}\phi(Z_i)+\sigma\mathbf{E}_i$ from proposition \ref{prop:Phi_W_expansion} we have:
\begin{align*}
\mathbb{E}\left[\| \mathbf{Y}_i - p^{1/2}\mathbf{W}\phi_s(Z_i)\|^2 \right] &= \mathbb{E}\left[\|  p^{1/2}\mathbf{W}\phi(Z_i)-p^{1/2}\mathbf{W}\phi_s(Z_i)+\sigma\mathbf{E}_i\|^2 \right]\\
 &= p\mathbb{E}[\|\phi(Z_i) - \phi_{s}(Z_i)\|^2_{\ell_2}] + \sigma^2\mathbb{E}[\|\mathbf{E}_i\|^2]=p\sum_{k>s}\lambda_k^f + p\sigma^2,
\end{align*}
where the independence of $\mathbf{W}$, $\phi(Z_i)$, $\mathbf{E}_i$, the second equality in \eqref{eq:Phi_W_expansion} and the properties $\mathbb{E}[\mathbf{E}_{ij}]=0$, $\mathbb{E}[\mathbf{E}_{ij}^2]=1$  have been used.
An LMM with smooth kernel can therefore produce a data matrix $\+Y$ which is `approximately low-rank', a common feature of real data \citep{udell2019big}, hinting that PCA may be a useful tool to help recover $\phi(Z_1),\ldots,\phi(Z_n)$.

%In fact, if we are willing to assume that $f$ is polynomial, then $r < \infty$ \citep{rubin2020manifold} and $\mathcal{M}$ is entirely contained within a finite-dimensional subspace. The complexity of this polynomial (e.g. degree) will dictate how large this dimension must be.

% If we take $\mathcal{M}$ to be approximately contained within the first $\hat r$ coordinates then, subject to regularity of $\W$ and assuming $p > \hat r$, the data vectors $\mathbf{Y}_i \stackrel{m.s.}{=} p^{1/2} \W \phi(Z_i) + \sigma \E_i$ will lie close an $\hat r$-dimensional subspace of $\mathbb{R}^p$. 
% Although our model is different, the arguments presented here are close 
% Our model is different but the explanation is essentially similar to \citet{udell2019big}.

\subsection{A visual example}
To illustrate some of the concepts from section \ref{sec:connecting}, we consider a case in which  $\mathcal{Z}$ is a torus embedded
in $\mathbb{R}^{3}$, satisfying  \ref{ass:diffeo_Rd_alt} with $d=3$. We take $\mu$ to be the uniform distribution on the torus, and $Z_{1},\ldots,Z_{4000}$ simulated from $\mu$
are shown in figure \ref{fig:torus_Zs}. The colouring of the points
in this figure emphasises that the torus is the Cartesian product of two circles, and the locations on the torus can be parameterised in terms of angles around
these two circles. 
\begin{figure}[h!]
\includegraphics[width=1\columnwidth]{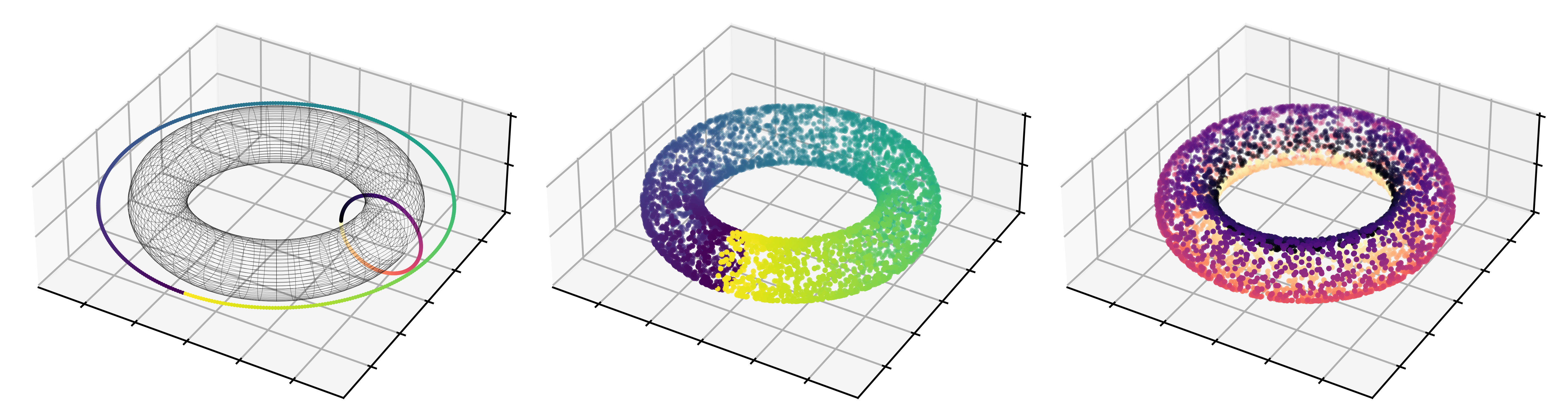}\caption{\label{fig:torus_Zs} Torus example. Left: grey wireframe of $\Zc$,
a torus, with colour bars indicating coordinates with respect to two circles. Both the middle and right plots show the same $n=4000$
points, $Z_1,\ldots,Z_{4000}$, which are sampled uniformly on the torus, coloured by their
coordinates with respect to each of the two circles.}
\end{figure}
\begin{figure}[h!]
\includegraphics[width=1\columnwidth]{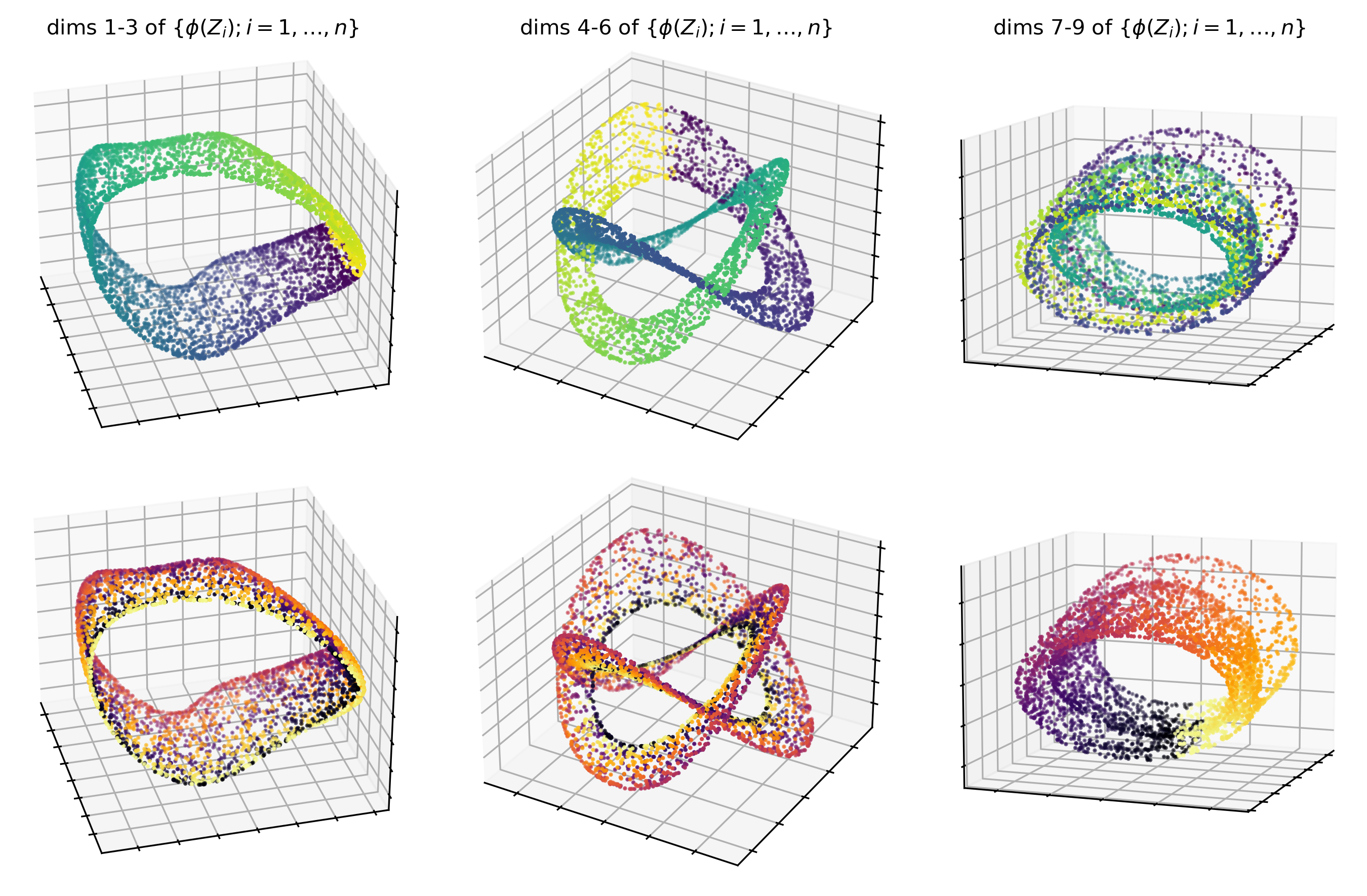}
\caption{\label{fig:torus_pcs}Torus example. Both the top and bottom rows
show the first 9 dimensions of $\phi(Z_i)$, $i=1,\ldots,4000$.
In each row, points are coloured according to the coordinates of the
underlying points $Z_{1},\ldots,Z_{n}$ with respect to the two
circles shown in figure \ref{fig:torus_Zs}. Numerical scales are
omitted to de-clutter the plots.}
\end{figure}

We assume $X_1\ldots,X_p$ are i.i.d., zero-mean Gaussian processes with common covariance function $\exp(-\|z-z^{\prime}\|_{\mathbb{R}^{3}}^{2})=f(z,z^{\prime})$,
which satisfies \ref{ass:cont_covar}. Figure \ref{fig:torus_pcs} shows numerical approximations to the first 1-3, 4-6 and 7-9 dimensions of $\phi(Z_i)$, $i=1,\ldots,4000$ (these approximations were obtained using PCA, the details of which are given later in section \ref{subsec:PCA-scores}). The only difference between the two rows of plots in figure \ref{fig:torus_pcs}
is the colouring of the points; the colouring in the top row is the
colouring of the corresponding points in the middle plot in figure
\ref{fig:torus_Zs}, similarly the colouring in the bottom row matches
that in the plot on the right of figure \ref{fig:torus_Zs}.
\begin{figure}[h!]
%\begin{wrapfigure}{r}{0.3\textwidth} 
\begin{center}     
\includegraphics[trim={0.2cm 0.7cm 0 0.2cm}, width=0.28\textwidth]{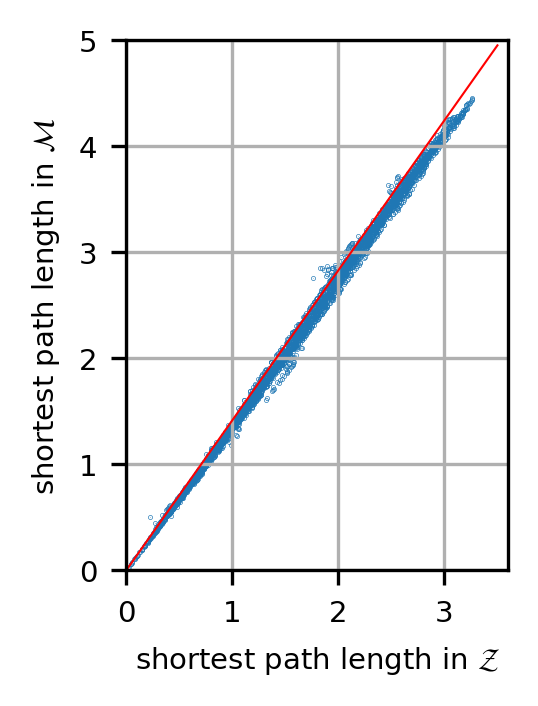}   
\end{center}   
\caption{Torus example. Blue: numerical shortest path lengths between points in $\mathcal{M}$ vs. between the corresponding points in $\Zc$. Red: theoretical scaling relationship $\sqrt{2}$.  \label{fig:torus_path_lengths}} 
\end{figure}
%\end{wrapfigure}

It is clear from figure \ref{fig:torus_pcs} that the global shape of $\mathcal{M}$, when viewed three dimensions at a time, is qualitatively different to the global shape of $\Zc$. However, assumption \ref{ass:injective} holds, so by lemma \ref{prop:homeomorphism}, we know $\mathcal{M}$ is topologically
equivalent to $\Zc$, and by proposition \ref{prop:sq_euc_kernels_alt}, $\phi$ isometry holds, up to a scaling factor of $\sqrt{-g^\prime(0)}=\sqrt{2}$ for the $f$ in question.  
 This tells us that shortest path lengths in $\mathcal{M}$ are equal to the corresponding path lengths $\Zc$, up to a factor of $\sqrt{2}$. Figure \ref{fig:torus_path_lengths} shows comparison of these shortest path lengths, computed numerically from the points in figures \ref{fig:torus_pcs} and \ref{fig:torus_Zs} using a nearest neighbour graph as detailed in section \ref{subsec:nn_graph}. We see a close approximation to the theoretical scaling factor of $\sqrt{2}$, shown by the red line. 

 Overall this example illustrates that if we are interested in discovering the topological or geometric structure of $\Zc$ based on observations of $\mathcal{M}$, we should not pay attention to the global shape of $\mathcal{M}$ that we perceive visually, because that depends on both $\Zc$ and $\phi$.  However, when homeomorphism holds, we can in principle recover the abstract topological structure of $\Zc$ and its homological features such as number of connected components, number of holes, etc., from $\mathcal{M}$. Moreover, when isometry holds, at least up to a constant scaling factor, we can gain insight into the geometry of $\Zc$ from shortest paths in $\mathcal{M}$.

\section{Methodology}\label{sec:methodology}
In this section, properties of the LMM are used to explain and justify the workflow outlined in section \ref{sec:Introduction}. Discussion of the step 1. is postponed until after discussion of step 2.

\subsection{Linear dimension reduction by PCA}\label{subsec:PCA-scores}
Given data $\mathbf{Y}\in\mathbb{R}^{n\times p}$ and $s\leq\min\{p,n\}$,
let the columns of $\mathbf{V}_{\mathbf{Y}}\in\mathbb{R}^{p\times s}$
be orthonormal eigenvectors associated with the $s$ largest eigenvalues
of $\mathbf{Y}^{\top}\mathbf{Y}\in\mathbb{R}^{p\times p}$. The \emph{dimension-$s$
PCA embedding} is the collection of vectors $\zeta_{1},\ldots,\zeta_{n}$, defined by:
\begin{equation}\label{eq:pc_scores_defn}
[\zeta_{1}|\cdots|\zeta_{n}]^{\top}\coloneqq\mathbf{Y}\mathbf{V}_{\mathbf{Y}},
\end{equation}
so for each $\zeta_{i}=\mathbf{V}_{\mathbf{Y}}^\top \Y_i$ 
is a vector in $\mathbb{R}^{s}$. These quantities are sometimes called principal component scores \cite{lee2010convergence,shen2012high,hellton2017and}.  
When performing PCA in practice, one often centers the data about their sample mean. For simplicity of presentation we do not consider such centering here, although we do not require population centering, that is we do not assume $\mathbb{E}[\Y_i]=0$.

The following assumptions about the LMM are introduced to enable theoretical analysis of the PCA embedding.
\begin{assump}\label{ass:independence}
The random functions $X_1,X_2, \ldots$ are independent.
\end{assump}
\begin{assump}\label{ass:moments_q_equals_1}
$\sup_{j\geq 1}\sup_{z\in\mathcal{Z}}\mathbb{E}[|X_{j}(z)|^{4}]<\infty$
and $\sup_{j \geq 1}\sup_{i\geq 1}\mathbb{E}[|\mathbf{E}_{ij}|^{4}]<\infty$.
\end{assump}
\begin{assump}\label{ass:finite rank}
For each $p\geq 1$, the rank $r$ of the mean correlation kernel $f$ defined in \eqref{eq:f_defn} is finite, and $r$ and   $1/\lambda_r^f$ are bounded as $p\to\infty$.
\end{assump}

\begin{thm}
\label{thm:consistency_summary}
Assume \ref{ass:cont_covar}, \ref{ass:independence}-\ref{ass:finite rank} and let 
$r$ be as therein. Let $\zeta_{1},\ldots,\zeta_{n}$
be the dimension-$r$ PCA embedding of $\Y\in\mathbb{R}^{n\times p}$ under the LMM.
Then there exists a random orthogonal matrix $\mathbf{Q\in}\mathbb{R}^{r\times r}$ depending on $n$ and $p$ such
that
\begin{equation}\label{eq:consistency_summary}
\max_{i=1,\ldots,n}\left\Vert p^{-1/2}\mathbf{Q}\zeta_{i}-\phi(Z_{i})\right\Vert_{2} \in O_{\mathbb{P}}\left(\frac{1}{\sqrt{n}}+\sqrt{\frac{n}{p}}\right)
\end{equation}
as $n\to\infty$ and $p/n\to\infty$ simultaneously, where $\|\cdot\|_2$ is the Euclidean norm.
\end{thm}
\noindent Theorem \ref{thm:consistency_summary} is a corollary to a more detailed non-asymptotic concentration result for the PCA embedding, theorem \ref{thm:uniform_consistency}, stated and proved in appendix
\ref{sec:Proof-and-supporting_consistency}.

\subsubsection*{Interpretation of theorem \ref{thm:consistency_summary}} Theorem \ref{thm:consistency_summary} implies that for any $\epsilon>0$, the probability that $\max_{i=1,\ldots,n}\left\Vert p^{-1/2}\mathbf{Q}\zeta_{i}-\phi(Z_{i})\right\Vert_{2}>\epsilon$ converges to zero when $n$ and $p/n$ grow simultaneously. In that sense $\phi(Z_1),\ldots, \phi(Z_n)$ can be recovered from $p^{-1/2}\zeta_1,\ldots,p^{-1/2}\zeta_n$, up to an orthogonal transformation, i.e., a transformation which preserves distances and inner-products. We see that computing the PCA embedding achieves a form of de-noising and signal extraction: each $\zeta_i$ depends on all three sources of randomness in the LMM, but $\phi(Z_i)$ clearly depends only on the random latent variable $Z_i$.

The result has positive implications for different forms of unsupervised learning, such as clustering, topological data analysis or manifold learning in the regime $n, p/n \rightarrow \infty$.
Viewed as sets, the point clouds $\{p^{-1/2}\zeta_{i}\}_{i = 1, \ldots, n}$ and $\{\phi(Z_i)\}_{i={1, \ldots, n}}$ converge to each other in Hausdorff distance, up to $\+Q$, implying convergence of topological summaries such as persistence diagrams \citep{wasserman2018topological}, and so on. Broadly speaking, 
we can consider $p^{-1/2}\zeta_1,\ldots,p^{-1/2}\zeta_n$ as proxies for $\phi(Z_1), \ldots, \phi(Z_n)$ when estimating features of $\mathcal{M}$. The LMM then gives us a way to translate such estimates into statements about the latent domain $\Zc$ (see section~\ref{sec:connecting}).

\subsubsection*{Discussion of assumptions \ref{ass:independence}-\ref{ass:finite rank}}

In the proof of theorem \ref{thm:uniform_consistency} and hence theorem \ref{thm:consistency_summary}, the independence assumption \ref{ass:independence}
and the moment assumption \ref{ass:moments} are used when analysing $p^{-1}\Y \Y^\top$ via matrix a polynomial moment concentration inequality from \cite{paulin2016efron}. The moment assumption  \ref{ass:moments} is not particularly restrictive. From a modelling point of view relaxing \ref{ass:independence} to some form of weak dependence or mixing condition would be desirable, but the authors do not know of any suitable polynomial moment matrix concentration inequalities which are applicable in that situation.

Concerning assumption \ref{ass:finite rank}, that $f$ has finite rank: recall from  section~\ref{sec:smoothness} that the eigenvalues tend to tail off quickly when $f$ is smooth, in which case assumption \ref{ass:finite rank} might be taken to hold approximately. Moreover, if $f$ is polynomial \citep{rubin2020manifold} or piecewise polynomial, or if $\Zc$ consists of finitely many points, then $f$ has strictly finite rank. As a result, assumption \ref{ass:finite rank} is mild enough to include any function $f$ which is obtainable from standard numerical or function approximation schemes (e.g. Taylor expansion, polynomial splines, etc).

The condition that  $r$ is bounded as $p\to\infty$ in \ref{ass:finite rank} can be understood as constraining the functional complexity of  $f$ as $p$ grows. The condition that  $1/\lambda_r^f $ is bounded as $p\to\infty$ means that the additive noise whose scale is specified by the constant $\sigma$ cannot overwhelm the ``signal'' in the LMM. If $X_1,X_2,\ldots$ are identically distributed then $f$, and hence $r$ and $\lambda_r^f$, are automatically constant in $p$. The assumption in  theorem \ref{thm:consistency_summary} that the dimension of the PCA embedding is equal to the finite kernel rank $r$ is an idealisation, although a very common type of assumption in uniform consistency results for spectral embedding, e.g. \cite{lyzinski2016community}.

\subsection{Choosing the PCA dimension} \label{subsec:Choosing-r}
 Our model and theory motivate a new method for choosing the embedding dimension, $\hat r$. Before proceeding, we should make clear that the hat notation in $\hat r$ is meant loosely: we seek a choice which achieves a good bias/variance trade-off in practice, and this may or may not coincide with the true rank of the kernel, $r$. Moreover, we do not claim that there is a `best' choice: different tasks benefit from different choices. In particular, if using PCA for prediction purposes we simply recommend cross-validation, as is common practice. For more exploratory analyses, as conducted here, we propose the following approach instead.
 
Assuming $n$ is even split the data $\Y$ into two, $\Y^{(1)},\Y^{(2)}\in\mathbb{R}^{n/2\times p}$, and for each candidate dimension $\rho$, take the orthogonal projection of the rows of $\Y^{(1)}$ onto the $\rho$ principal eigenvectors of ${\Y^{(1)}}^\top\Y^{(1)}$ --- the resulting $n/2$ vectors are $p$-dimensional, just constrained to a $\rho$-dimensional subspace. Next, measure how much this projection step has brought the first half closer to the second, using Wasserstein distance. Select $\hat{r}$ to be the $\rho$ achieving the smallest distance. The procedure is described precisely in algorithm~\ref{alg:dim_select}. The algorithm accommodates the possibility that there may be some maximal value $\rho_{\text{max}} <\min(n,p)$ that one is willing to consider, e.g. for computational reasons.

When calculating the Wasserstein distance between two sets of $n$ samples in $\mathbb{R}^p$, the overall computational complexity can be up to $O(n^2(p+ n\log n))$. One way to reduce the burden of computing Wasserstein distances over very large dimensional point sets, is to reduce the data to $\rho_{\text{max}} \ll p$ dimensions using PCA. Reducing the complexity further using approximations is well studied \cite{cuturi2013sinkhorn,arjovsky2017wasserstein,bonneel2015sliced,arjovsky2017wasserstein}. A variety of methods are available in the
Python packages \verb|POT| \citep{flamary2021pot} and  \verb|Geomloss| \citep{feydy2019interpolating}.

% One way to limit the dependence on $p$ in this overall complexity is to reduce the data $\mathbf{Y}_1,\ldots,\mathbf{Y}_n$ to $\rho_{\text{max}}$ dimensions using PCA before applying algorithm~\ref{alg:dim_select}. 
% % Different measures of distance between point clouds can be chosen to reflect different inferential goals.
% here is nothing fundamentally special about the Wasserstein distance in this approach, or indeed its exact configuration, i.e., which transport cost to use; 

To understand how $\hat{r}$ might relate to $r$, let us make a few simplifying assumptions. Suppose $r < \infty$, so that (with exact equality)
\[\+Y_i = p^{1/2} \+W \phi(Z_i) + \sigma \+E_i,\]
and that the second Wasserstein distance is used, that is $d_\rho = \mathcal{W}_2(\mathbf{Y}^{(1)}\boldsymbol{\Pi}_\rho, \mathbf{Y}^{(2)})$ where
\[\mathcal{W}^2_2(\mathbf{A},\mathbf{B}) \coloneqq \min_\pi \frac{1}{m} \sum \lVert \mathbf{A}_i - \mathbf{B}_{\pi(i)} \rVert^2_2, \quad  \+A, \+B \in \mathbb{R}^{m \times d},\]
where $\+A_i$ and $\+B_i$ are the $i$th rows of $\+A$ and $\+B$, and where the minimum is over all permutations of the integers $1, \ldots, m$.%, with $n = 2m$ (assumed even). 

The second Wasserstein distance is particularly amenable to mathematical analysis because of the following property, which can be checked by direct calculation. If there exist $\hat{\+A}_1, \ldots, \hat{\+A}_m$ such that for all $i,j$ we have
$\langle{\+A}_i - \hat{\+A}_i, \+B_j \rangle = 0$, then
\begin{equation}\label{eq:wass_id}\mathcal{W}_2^2(\+A, \+B) = \frac{1}{m}\lVert \+A - \hat{\+A} \rVert^2_{\text{F}}  + \mathcal{W}_2^2(\hat{\+A}, \+B).
\end{equation}
To see why algorithm~\ref{alg:dim_select} might reject overly values of $\rho$, suppose $\rho > r$ and consider the projection errors $\boldsymbol{\mathcal{E}}^{(1)} = \mathbf{Y}^{(1)}(\boldsymbol{\Pi}_\rho-\boldsymbol{\Pi}_r)$.
With $\boldsymbol{\mathcal{E}}^{(1)}_i$ denoting the $i$th row of $\boldsymbol{\mathcal{E}}^{(1)}$ and the superscript $(k)$ indicating random objects associated with $\+Y^{(k)}$, suppose
\begin{equation}
  \frac{1}{p}\langle\: \boldsymbol{\mathcal{E}}^{(1)}_i, \+E^{(2)}_j\rangle  \approx 0, \quad \text{and} \quad 
  \frac{1}{p} \langle\: \boldsymbol{\mathcal{E}}_i^{(1)}, \sqrt{p} \+W \phi(Z^{(2)}_j) \rangle  \approx 0, \quad i,j = 1, \ldots, n/2. \label{eq:orthogonality}
\end{equation}
%To see the rationale for algorithm~\ref{alg:dim_select} in the case that  $\rho>r$, 
Then, using \eqref{eq:wass_id},
\[\frac{1}{p} d^2_\rho \approx \frac{2}{np} \lVert \boldsymbol{\mathcal{E}}^{(1)} \rVert^2_{\text{F}} + \frac{1}{p}d^2_r,\]
where $\lVert \boldsymbol{\mathcal{E}}^{(1)} \rVert^2_{\text{F}}$ is non-decreasing in $\rho - r$,
and it follows that we should expect $d_\rho > d_r$ when $\rho$ is large relative to $r$.

Why should the approximations in \eqref{eq:orthogonality} hold? The first is the product of sample-splitting. If the $\+E_i$ are statistically independent, then $\+E^{(2)}$ is statistically independent of $\+V^{(1)}_{\rho}$. Combined with the fact that the elements of  $\+E^{(2)}_i$ are mean-zero and unit-variance, we therefore expect the $p^{-1}\+E^{(2)}_i$ to be approximately orthogonal to the subspace spanned by the columns of $\+V^{(1)}_{\rho}$, \emph{when $p$ is large relative to $n$}. The second approximation seems to be reasonable as long as $n$ and $p$ are large, and we have confirmed this by simulation. %Note that this is not equivalent to supposing $\+V^{(1)}_{r}$ is `approximately a constant', which would be unreasonable when $p \gg n$.

Now consider the case $\rho < r$. We have $$d_\rho^2 = \mathcal{W}_2^2(\mathbf{Y}^{(1)}\boldsymbol{\Pi}_\rho,\mathbf{Y}^{(2)}) = \frac{2}{n}\|\Y^{(2)}(\boldsymbol{\Pi}_\rho-\mathbf{I}_p)\|_{\mathrm{F}}^2 + \mathcal{W}_2^2(\mathbf{Y}^{(1)}\boldsymbol{\Pi}_\rho,\mathbf{Y}^{(2)}\boldsymbol{\Pi}_\rho).$$
By the Eckart-Young theorem, 
\begin{equation}\label{eq:eck-yng}
\|\mathbf{Y}^{(2)}(\boldsymbol{\Pi}_{\rho}-\mathbf{I}_p)\|_{\mathrm{F}}^{2}\geq \sum_{k>\rho}\lambda_{k}(\mathbf{Y}^{(2)\top}\mathbf{Y}^{(2)})
\end{equation}
where $\lambda_{k}(\mathbf{Y}^{(2)\top}\mathbf{Y}^{(2)})$ is the $k$th largest eigenvalue of $\mathbf{Y}^{(2)\top}\mathbf{Y}^{(2)}$.
The r.h.s. of \eqref{eq:eck-yng} is non-decreasing in $r-\rho$. To see that the term $\mathcal{W}_2^2(\mathbf{Y}^{(1)}\boldsymbol{\Pi}_\rho,\mathbf{Y}^{(2)}\boldsymbol{\Pi}_\rho)$ converges to zero as $p/n,n\to\infty$, it is convenient to consider the case $\sigma=0$.  In this situation, 
$$
\frac{1}{p}\mathcal{W}_2^2(\mathbf{Y}^{(1)}\boldsymbol{\Pi}_\rho,\mathbf{Y}^{(2)}\boldsymbol{\Pi}_\rho)\leq \frac{1}{p}\mathcal{W}_2^2(\mathbf{Y}^{(1)},\mathbf{Y}^{(2)})=\frac{1}{p}\mathcal{W}_2^2(\boldsymbol{\Phi}^{(1)}\mathbf{W}^\top,\boldsymbol{\Phi}^{(2)}\mathbf{W}^\top)
$$
where $\boldsymbol{\Phi}^{(k)}=[\phi(Z_1^{(k)})|\cdots|\phi(Z_{n/2}^{(k)})]^\top\in\mathbb{R}^{n/2\times r}$.
%,  and $Z_i^{(j)} \coloneqq Z_i$ for $j=1$ and $Z_i^{(j)}\coloneqq Z_{i+n/2}$ for $j=2$. 
Appealing to the same arguments as in section \ref{subsec:intro_Phi_W}, as $p/n\to\infty$, $\frac{1}{p}\mathcal{W}_2^2(\boldsymbol{\Phi}^{(1)}\mathbf{W}^\top,\boldsymbol{\Phi}^{(2)}\mathbf{W}^\top)$ is concentrated about $\mathcal{W}_2^2(\boldsymbol{\Phi}^{(1)},\boldsymbol{\Phi}^{(2)})$, and the latter converges to zero as $n\to\infty$ because the rows of  $\boldsymbol{\Phi}^{(1)} $ and $\boldsymbol{\Phi}^{(2)}$ are i.i.d. random vectors in $\mathbb{R}^r$. It follows that we should expect $d_{\rho} > d_{r}$ when $\rho$ is small relative to $r$, and so overall that $d_{\rho}$ will have a minimum near $\rho=r$.

A general rule we could draw from these arguments, and which we see in practice, is that to recommend substantial dimension reduction the algorithm wants to see a large $p$ relative to $n$, and noise. Conversely, if the noise level is low or if $n$ is large relative to $p$, then $d_\rho$ may keep decreasing with $\rho$, which we tend to interpret as contraindication against PCA.

\subsubsection*{Method comparison}
We now explore the performance of Algorithm~\ref{alg:dim_select} in a few simulated examples. We consider four configurations, where each configuration refers to a choice of latent space $\mathcal{Z}$ and corresponding kernel $f$. In the first configuration, the latent space comprises six distinct elements. The latent spaces in the remaining configurations are different subsets of $\mathbb{R}^2$. In each configuration, we draw $n = 500$ points $Z_i$ uniformly on $\mathcal{Z}$, and the resulting point sets are shown in figures~\ref{fig:dimension_selection}a)1-4.

In the first configuration, we choose an arbitrary $6 \times 6$ positive-definite matrix to represent the kernel. In the second, $f(x,y)= (x^\top y +1)^2$, which has rank 6; in the third, $f(x,y) = \{\cos(x^{(1)} - y^{(1)}) + \cos(x^{(2)} - y^{(2)}) +2\}$, which has rank 5; and in the fourth, $f(x,y)=\exp(-\|x-y\|_{\mathbb{R}^2}^{2}/2)$, which has infinite rank.

We simulate a $500 \times 1000$ data matrix $\mathbf{Y}$ in each configuration, where the $p=1000$ random functions are independent, zero mean Gaussian processes with the same covariance kernel $f$, and the errors $\mathbf{E}_{ij}$ are independent and standard normal.

% By eye, it is impossible to distinguish any of the structure of $\mathcal{Z}$ by visualising selected pairs or triples of data coordinates from $\mathbf{Y}$ (figure~\ref{fig:indistinguishability}, appendix), but the first two principal components of the data, shown in figures~\ref{fig:dimension_selection}b)1-4, draw out some structure. For example, in figure~\ref{fig:dimension_selection}b)1, we can distinguish 6 clusters. In the remaining configurations, this two-dimensional view is insufficient and obscures obvious topological and geometric features of $\mathcal{Z}$, for example, the hole in configuration 4, or the `Z' shape in configuration 3. A plot of the second and third principal components (figure~\ref{fig:PCA_other_dimensions}, appendix) gives a better view of $\mathcal{Z}$ in configurations 2-4 but in isolation could tell the wrong story in configuration 1, suggesting 5 rather than 6 clusters.

In figures~\ref{fig:dimension_selection}c)1-4 we show the Wasserstein error (log-scale), i.e., the distance computed in Algorithm~\ref{alg:dim_select}, for dimensions up to $\rho_{\text{max}} = 30$. Reassuringly, the optimum roughly coincides with the true rank of the kernel when finite (dashed black line, configurations 1-3) and at the same time it is interesting that a non-degenerate optimum is still found under infinite rank (configuration 4). If we lower the noise, the optimal dimension increases (figure~\ref{fig:Wasserstein_varying_sigma}, Appendix), reflecting the afore-mentioned bias/variance trade-off.

For comparison, we also show the dimensions selected using the ladle \citep{luo2016combining} and elbow methods \citep{zhu2006automatic}, as implemented in the R packages `dimension' (on GitHub: \url{https://github.com/WenlanzZ}) and `igraph' (on The Comprehensive R Archive Network), respectively. The ladle and Wasserstein methods seem to make similar choices, although as implemented the ladle method is computationally costly, which has precluded more simulations or going beyond $\max(n,p) = 1000$ to allow a more comprehensive comparison. We would advise against the elbow method for dimension selection under the LMM, as it appears to favour dangerously low dimensions.

In configurations 3 and 4, there is isometry between $\mathcal{M}$ and $\mathcal{Z}$. As a result, we can aim to recover the path lengths in $\Zc$ amongst $Z_1, \ldots, Z_n$ from $p^{-1/2} \zeta_{1}, \ldots,p^{-1/2} \zeta_{n}$ -- see section \ref{subsec:nn_graph} for details. This method yields infinite distances when the $k$-nearest neighbour graph isn't connected. Dealing with this issue in a systematic way is awkward, and we settled on the following solution. Picking $\epsilon$ as the $5\%$ quantile of the $\mathbb{R}^2$ Euclidean distance matrix between the $Z_i$, we place an edge between any pair of points within distance $\epsilon$, weighted by Euclidean distance, and approximate the geodesic distance between two points as the corresponding weighted graph distance. Any infinite distance remaining is replaced with the original Euclidean distance. The blue line in figures~\ref{fig:dimension_selection}d)3-4 shows the entrywise mean square error between the estimated geodesic distance matrices of $Z_1, \ldots, Z_n$ and  $p^{-1/2} \zeta_{1}, \ldots,p^{-1/2} \zeta_{n}$, for different choices of $r$. The optimum roughly coincides with the dimensions selected by the ladle and Wasserstein methods.

Because of the isometric relationship between $\mathcal{M}$ and $\mathcal{Z}$ in configurations 3 and 4, we might also hope that the persistence diagrams of their Rips filtrations would be similar. %, although even asymptotically we should not expect an exact match. 
The red line in figures~\ref{fig:dimension_selection}d)3-4 shows the bottleneck distance between the persistence diagrams of the Rips filtrations of $Z_1, \ldots, Z_n$ and $p^{-1/2} \zeta_{1}, \ldots,p^{-1/2} \zeta_{n}$, as implemented in the R package `TDA', for different choices of $r$. In this metric, the optimal dimension (lowest bottleneck distance) is lower than that suggested by the ladle and Wasserstein methods, but we do not know to what extent this should be expected in general. The scales of the log-Wasserstein error, geodesic distance error, and bottleneck distance are not comparable and in figures~\ref{fig:dimension_selection}d) we have recentered and rescaled the curves to make their maxima and minima agree.

In figures~\ref{fig:dimension_selection}e)1-4 we show the persistence diagrams of the Rips filtrations of $p^{-1/2} \zeta_{1}, \ldots, p^{-1/2} \zeta_{n}$ computed on the basis of the dimension selected by the Wasserstein method (Algorithm~\ref{alg:dim_select}), using the R package `TDA'. Recall that in persistent homology the significance of a topological feature is quantified by its persistence, death minus birth, which is the vertical distance between the point (birth,death) to the diagonal $x=y$. Following \citet{fasy2014confidence} we draw a line parallel to $x=y$ to separate the signal from the noise, picking $y = x + 0.2$ by eye. In each figure, we report the number of connected components, $\hat \beta_0$, and holes, $\hat \beta_1$, estimated by this heuristic. The true corresponding values for $\mathcal{Z}$ are respectively (6,0), (1,8), (1,0), and (1,1).

\begin{algorithm}[t]
 \caption{PCA dimension selection}\label{alg:dim_select}
% \vspace{-1em}
 \begin{flushleft}
  \textbf{Input:} data matrix $\mathbf{Y}\in\mathbb{R}^{n\times p}$; \\
  \hphantom{\textbf{Input:} }maximal dimension $\rho_{\text{max}} \leq \min(n,p)$, default $\rho_{\text{max}} = \min(n,p)$.
 \end{flushleft}
 \vspace{-1em}
\begin{algorithmic}[1]
  \STATE Split the data as $\mathbf{Y}^{(1)} \coloneqq \mathbf{Y}_{1:\lceil n/2 \rceil, 1:p}$, $\mathbf{Y}^{(2)} \coloneqq \mathbf{Y}_{(\lceil n/2 \rceil + 1): n, 1:p}$
  \FOR{$\rho\in\{1,...,\rho_{\text{max}}\}$} 
  \STATE Let $\mathbf{V}^{(1)}_{\rho} \in \mathbb{R}^{p\times \rho}$ denote the matrix of orthogonal eigenvectors associated with the $\rho$ largest eigenvalues of $\mathbf{Y}^{(1)\top}\mathbf{Y}^{(1)}$.
  \STATE Project  $\mathbf{Y}^{(1)}$ onto the linear span of  the columns of $\mathbf{V}^{(1)}_{\rho}$, that is, compute: $\mathbf{Y}^{(1)}\boldsymbol{\Pi}_\rho $ where $\boldsymbol{\Pi}_\rho\coloneqq\mathbf{V}^{(1)}_{\rho} \mathbf{V}^{(1)\top}_{\rho}$.
  \STATE Compute Wasserstein distance $d_{\rho}$ between $\mathbf{Y}^{(1)}\boldsymbol{\Pi}_\rho$ and $\mathbf{Y}^{(2)}$ as point sets in $\mathbb{R}^p$.
  \ENDFOR
\end{algorithmic}
\vspace{-1em}
 \begin{flushleft}
  \textbf{Output:} selected dimension $\hat r = \text{argmin } \{d_{\rho}\}$.
 \end{flushleft}
 \vspace{-1em}
 \end{algorithm}

% \begin{algorithm}[t]
%  \caption{PCA dimension selection}\label{alg:dim_select}
% % \vspace{-1em}
%  \begin{flushleft}
%   \textbf{Input:} data matrix $\mathbf{Y}\in\mathbb{R}^{n\times p}$.
%  \end{flushleft}
%  \vspace{-1em}
% \begin{algorithmic}[1]
%   \STATE Split the data as $\mathbf{Y}^{(1)} \coloneqq \mathbf{Y}_{1:\lceil n/2 \rceil, 1:p}$, $\mathbf{Y}^{(2)} \coloneqq \mathbf{Y}_{(\lceil n/2 \rceil + 1): n, 1:p}$
%   \FOR{$\rho\in\{1,...,\min(n,p)\}$} 
%   \STATE Let $\mathbf{V}^{(1)}_{\rho} \in \mathbb{R}^{p\times \rho}$ denote the matrix of orthogonal eigenvectors associated with the $\rho$ largest eigenvalues of $\mathbf{Y}^{(1)\top}\mathbf{Y}^{(1)}$
%   \STATE Project  $\mathbf{Y}^{(1)}$ onto the linear span of  $\mathbf{V}^{(1)}_{\rho}$,  $\hat{\mathbf{X}}^{(1)} \coloneqq \mathbf{Y}^{(1)} \mathbf{V}^{(1)}_{\rho} \mathbf{V}^{(1)\top}_{\rho}$
%   \STATE Compute Wasserstein distance $d_{\rho}$ between $\hat{\mathbf{X}}^{(1)}$ and $\mathbf{Y}^{(2)}$ (as point sets in $\mathbb{R}^p$)
%   \ENDFOR
% \end{algorithmic}
% \vspace{-1em}
%  \begin{flushleft}
%   \textbf{Output:} selected dimension $\hat r = \text{argmin } \{d_{\rho}\}$.
%  \end{flushleft}
%  \vspace{-1em}
%  \end{algorithm}

 \begin{figure}%[tbhp]
    \centering
    \includegraphics[trim={1cm 1cm 2cm 1cm}, width=\textwidth]{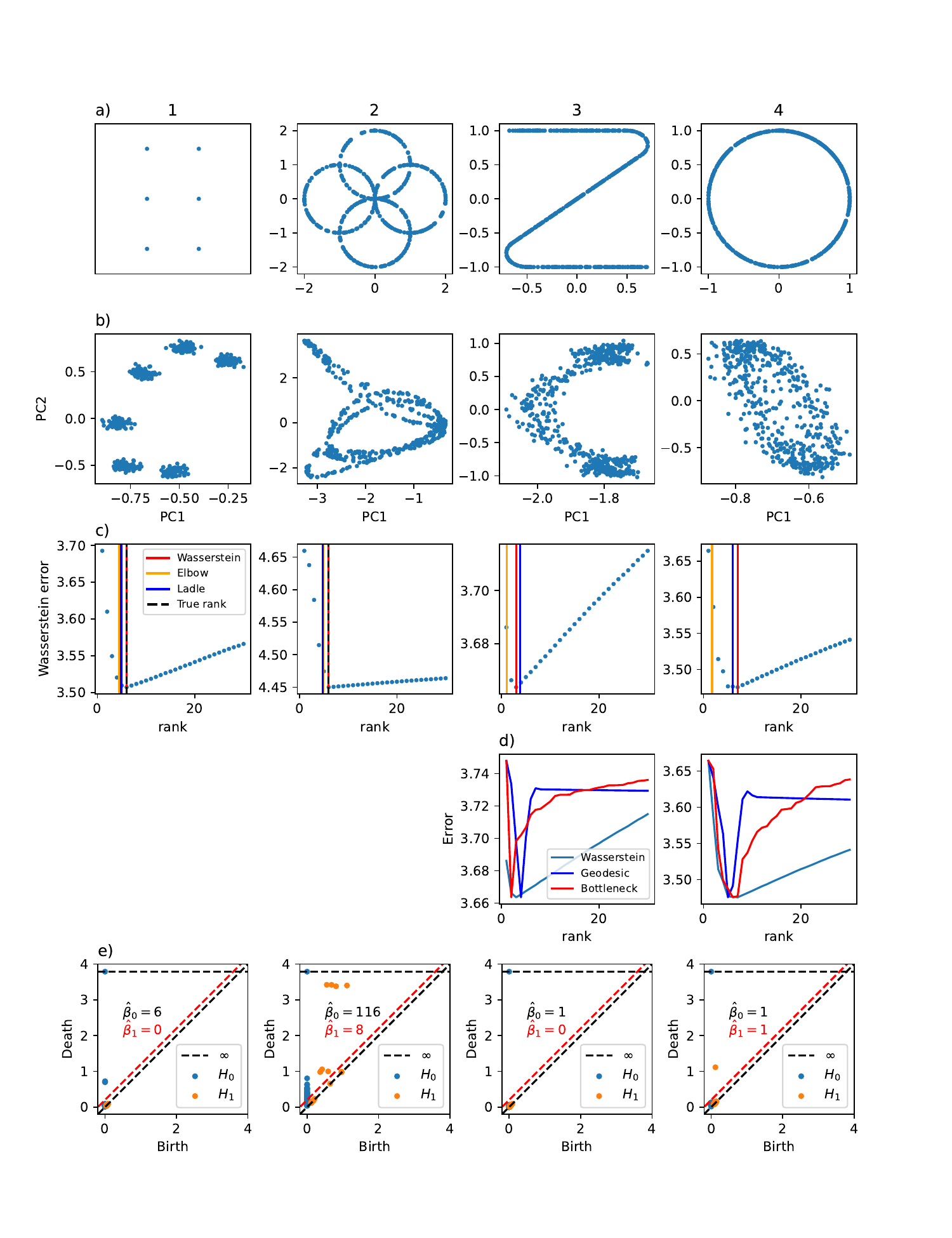}
    \caption{PCA dimension selection. Columns 1-4: different latent space/kernel configurations. 1-3 are finite rank, 4 infinite rank; configurations 3 and 4 are isometric. Row a: sampled positions ($n=500$); b: first two principal components ($p=1000$); c: the dimension selected by different methods, and the true rank when finite; d: error in geodesic distance and persistence diagram estimation (bottleneck distance) for the isometric configurations; e: persistence diagrams showing partial recovery of true topological features. Further details in main text.}
    \label{fig:dimension_selection}
\end{figure}

\subsection{Spherical projection}\label{sec:spherical_proj}
When performing data analysis we may wish to consider the assumption that $f$ belongs to one of the families of kernels in proposition \ref{prop:sq_euc_kernels_alt} or \ref{prop:inner_prod_kernels_alt}, because of their stationarity interpretation, and because the associated isometry properties would justify use of the PCA embedding to recover geometric features of $\Zc$. However, all these kernels have the property that \begin{equation} p^{-1}\sum_{j=1}^p\mathbb{E}[|X_j(z)|^2] = f(z,z)= \mathrm{const.},\label{eq:f(z,z)_equals_const}
\end{equation}
which from a modelling point of view may be restrictive.  We shall now show that the spherically projected PCA embedding has a model-based interpretation which allows this restriction to be loosened.

Suppose we are given random functions $X_1,\ldots,X_p$ such that $f(z,z)$ is constant in $z\in\Zc$. Without loss of generality assume this constant is $1$. As usual let $\phi$ be the feature map associated with $f$. Note that in this situation $\mathcal{M}$ is a subset of the unit hypersphere $\{x\in\ell_2:\|x\|_{\ell_2}=1\}$. With $\alpha_1,\ldots,\alpha_n$ being i.i.d. random variables whose distribution is supported on  a compact set $\mathcal{A}\subset\mathbb{R}_+$, define the model:
\begin{equation}
\Y_{ij}=\alpha_i X_j(Z_i)+\sigma\E_{ij}.\label{eq:LMM_extended}
\end{equation}
This can be viewed as a particular form of LMM with extended latent space $\Zc^{\mathrm{ext}}\coloneqq \mathcal{A}\times\Zc$ and extended random functions $X^{\mathrm{ext}}_j(\alpha,z)\coloneqq \alpha X_j(z)$. Its mean correlation kernel is:
$$
f^{\mathrm{ext}}(\alpha,z,\alpha^\prime,z^\prime)\coloneqq\frac{1}{p}\sum_{j=1}^p \mathbb{E}\left[X^{\mathrm{ext}}_j(\alpha,z)X^{\mathrm{ext}}_j(\alpha^\prime,z^\prime)\right] = \alpha\langle\phi(z),\phi(z^\prime)\rangle_{\ell_2}\alpha^\prime,
$$
the Mercer feature map of $f^{\mathrm{ext}}$ satisfies: $\phi^{\mathrm{ext}}(\alpha,z)=\alpha \phi(z)$, and we have $p^{-1}\sum_{j=1}^p\mathbb{E}[|X^{\mathrm{ext}}_j(\alpha,z)|^2] =\alpha^2$,
allowing more flexibility than \eqref{eq:f(z,z)_equals_const}.

Now suppose assumption \ref{ass:finite rank} holds, let $\zeta_1,\ldots,\zeta_n$ be the dimension-$r$ PCA embedding computed from $\Y$ under the extended LMM \eqref{eq:LMM_extended} and consider the spherical projection $\zeta_i^{\mathrm{sp}}\coloneqq \zeta_i/\|\zeta_i\|_2$. Using the identities $\|\phi^{\mathrm{ext}}(\alpha,z)\|_2=\alpha \|\phi(z)\|_2 = \alpha f(z,z)^{1/2}=\alpha$, and applying the triangle inequality several times gives:
\begin{align*}
\left\|\mathbf{Q}\zeta_i^{\mathrm{sp}}- \phi(Z_i)\right\|_2& = \left\|\frac{p^{-1/2}\mathbf{Q}\zeta_i}{p^{-1/2}\|\zeta_i\|_2}- \frac{\phi^{\mathrm{ext}}(\alpha_i,Z_i)}{\|\phi^{\mathrm{ext}}(\alpha_i,Z_i)\|_2}\right\|_2 \\
&\leq 2\frac{\|p^{-1/2}\mathbf{Q}\zeta_i -\phi^{\mathrm{ext}}(\alpha_i,Z_i) \|_2}{\alpha_i - \|p^{-1/2}\mathbf{Q}\zeta_i -\phi^{\mathrm{ext}}(\alpha_i,Z_i) \|_2},
\end{align*}
where $\mathbf{Q}$ is any orthogonal matrix. Theorem \ref{thm:consistency_summary} could therefore be applied to the LMM \eqref{eq:LMM_extended} to establish that for the particular $\mathbf{Q}$ in that theorem, $\|p^{-1/2}\mathbf{Q}\zeta_i -\phi^{\mathrm{ext}}(\alpha_i,Z_i) \|_2\to 0$, which by the above inequality implies $\left\|\mathbf{Q}\zeta_i^{\mathrm{sp}}- \phi(Z_i)\right\|_2\to 0$. In summary, under the model \eqref{eq:LMM_extended}, $\phi(Z_1),\ldots,\phi(Z_n)$ can be recovered from the spherically projected embedding $\zeta_1^{\mathrm{sp}},\ldots,\zeta_n^{\mathrm{sp}}$, up to an orthogonal transformation.

\subsection{Nearest neighbour graph construction}\label{subsec:nn_graph}

Constructing a nearest neighbour graph from the PCA embedding allows us to approximate topological and geometric features of $\mathcal{M}$ and hence $\Zc$. In keeping with the workflow set out in section \ref{sec:Introduction}, we focus on the spherically projected embedding $\zeta_1^{\mathrm{sp}},\ldots,\zeta_n^{\mathrm{sp}}$ but very similar considerations apply to the raw embedding $\zeta_1,\ldots,\zeta_n$. Noting that $\|\zeta_i^{\mathrm{sp}}\|_2=1$ for all $i$, we denote by $d_{\mathbb{S}}(\zeta_i^{\mathrm{sp}},\zeta_j^{\mathrm{sp}})\coloneqq\arccos(\langle\zeta_i^{\mathrm{sp}},\zeta_j^{\mathrm{sp}}\rangle_2)$ the circular arc distance on the unit hypersphere.

There are two popular types of nearest neighbour graph: the $\epsilon$-nn and $k$-nn graphs, both of which are undirected, weighted graphs with $n$ vertices, identified with $\zeta_1^{\mathrm{sp}},\ldots,\zeta_n^{\mathrm{sp}}$.  There is an edge between $\zeta_i^{\mathrm{sp}}$ and $\zeta_j^{\mathrm{sp}}$ in the $\epsilon$-nn graph if $d_{\mathbb{S}}(\zeta_i^{\mathrm{sp}},\zeta_j^{\mathrm{sp}})\leq \epsilon$, and in the $k$-nn graph if $\zeta_i^{\mathrm{sp}}$ is one of the $k$-nearest (with respect to $d_{\mathbb{S}}$) neighbours of  $\zeta_j^{\mathrm{sp}}$ or vice versa. In both types of graph, if there is an edge between $\zeta_i^{\mathrm{sp}}$ and $\zeta_j^{\mathrm{sp}}$ it is assigned weight $d_{\mathbb{S}}(\zeta_i^{\mathrm{sp}},\zeta_j^{\mathrm{sp}})$. A number of algorithms for identifying nearest neighbours exactly or approximately are available, for example in the Python library \verb|scikit-learn| \cite{scikit-learn}.

Recalling \eqref{eq:d_M_and_d_Z_defn}, nearest neighbour graph distances can be used to approximate shortest path lengths in $\mathcal{M}$:
$$
d_\mathcal{M}^{\mathrm{geo}}(\phi(Z_i) ,\phi(Z_j)) \approx \mathbf{D
}_{\mathcal{M}}^{ij}\coloneqq \min_{x_1,\ldots,x_m} \left\{d_{\mathbb{S}}(x_1,x_2)+\cdots+d_{\mathbb{S}}(x_{m-1},x_m)\right\},
$$
where the minimum is over all paths in the nearest neighbour graph connecting $x_1=\zeta_i^{\mathrm{sp}}$ and $x_m = \zeta_j^{\mathrm{sp}}$. If there are no such paths, $\mathbf{D
}_{\mathcal{M}}^{ij}=\infty$ by convention.

Various fast algorithms for computing shortest paths and shortest path lengths are available, for example in the Python library \verb|NetworkX| \cite{hagberg2008exploring}, which can be parallelised \verb|nx-parallel| (https://github.com/networkx/nx-parallel). If further speed-up in computing shortest paths length is needed, one might consider contraction hierarchies \cite{geisberger2008contraction}.

The use of nearest neighbour graphs to approximate path lengths on manifolds is well studied and is the first step in the Isomap procedure \cite{tenenbaum2000global}. The theoretical accuracy of such approximations has been analysed by
\cite{bernstein2000graph,trosset2020rehabilitating}.  In particular \cite{trosset2020rehabilitating} note that the $k$-nn graph is often preferred in practice although its analysis is more complicated. They also note that choosing a single value for $\epsilon$ or $k$ is a difficult problem in general. Where possible in the examples of section \ref{sec:examples} we compute and analyse the  nearest neighbour graph over a range of values for $\epsilon$ or $k$, rather than selecting one single value. This approach is similar in spirit to the computation of $\epsilon$-nn graphs over a range of $\epsilon $ values in persistent homology techniques \cite{carlsson2009topology,chazal2021introduction}.

\section{Examples}\label{sec:examples}
In the following three real data examples, we will assume the LMM holds and explore hypotheses about the latent domain, $\Zc$, feature map, $\phi$, and the manifold underlying the data, $\mathcal{M} = \phi(\Zc)$. In all examples, we have access to background information. In some cases, such as the first hypothesis in the image and transcriptomics examples, the background points us towards generic hypotheses, such as `$\mathcal{M}$ is a loop' or `$\mathcal{M}$ is a tree'. In others, e.g. the second hypothesis in the image and transcriptomics examples, the information is used more explicitly to obtain trial values for $\Zc$, realisations $z_i$ of $Z_i$, and to estimate parts of the model. 

In his famous book ``Exploratory Data Analysis'' \citep{edatukey}, Tukey observed that there was a readily accepted division in the process of criminal justice between ``the search for the evidence --- the responsibility of the police and other investigative forces --- and the evaluation of the evidence's strength --- a matter for juries and judges''. In terms of this analogy, the workflow we present is directed towards the former activity --- the search for clues, indications, appearances. In the examples to follow, we do not attempt to formally evaluate the strength of the evidence presented, beyond baseline comparisons against uniform models. We regard this confirmatory analysis as an important but distinct undertaking requiring different techniques.

The code and data used are available here: \url{https://github.com/anniegray52/explore_manifold_hyp}

\subsection{Images}\label{subsec:ex_images}
We return to the data set of images described in section \ref{sec:Introduction}. Recall there are $n=72$ images, each consisting of $p=110592$ grey-scale pixels, taken from angles $0,5,10,\ldots,355$ degrees around the circumference of a circle. We will assume XY coordinates for the camera positions, $\cos(\theta_i), \sin(\theta_i)$ for each angle $\theta_i$ (converted to radians). In this context, we will first consider the hypothesis:
\begin{enumerate}
\item[1.] $\mathcal{Z}$ is a circle and $\phi$ is a homeomorphism. An informal implication is: The data lie close to a loop.
\end{enumerate}
Finding this hypothesis to be tenable given the data, we will consider the stronger hypothesis:
\begin{enumerate}
\item[2.] $\mathcal{Z}$ is the circle of camera positions, $z_i = (\cos(\theta_i), \sin(\theta_i))$ are the (known) camera positions, and $\phi$ is a scaled isometry.
An informal implication is: Distances along the loop correspond to distances along the circle between camera positions. 
\end{enumerate}

The first step of the workflow is to apply the dimension selection method. As per figure \ref{fig:cars_EDA}a), using $\rho_{\text{max}} = 35$, this results in $\hat{r}=11$. The kernel density estimate in \ref{fig:cars_EDA}b) demonstrates the variation in the magnitudes of the dimension-$\hat{r}$ PCA embedding vectors $\|\zeta_1\|_2, \ldots, \|\zeta_n\|_2$; in all subsequent steps we instead work with the spherically projected embedding $\zeta_1^{\mathrm{sp}},\ldots,\zeta_n^{\mathrm{sp}}$ as per \eqref{sec:spherical_proj}. %, in order to explore the possibility of isometry. 
\begin{figure}[h!]
\includegraphics[width=1\columnwidth]{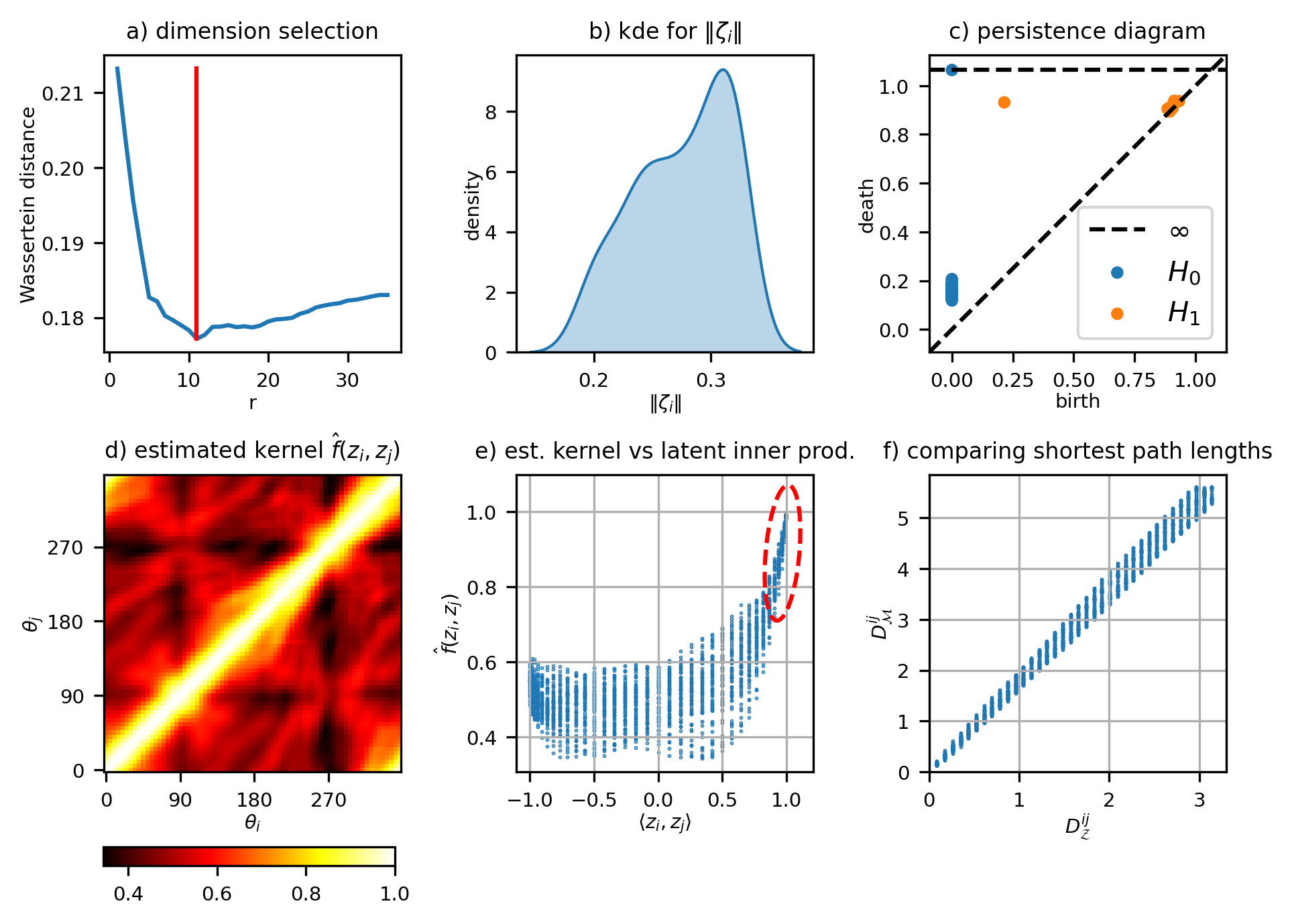}
\caption{\label{fig:cars_EDA}Images example. a) Wasserstein dimension selection; red vertical line indicates   minimum at $\hat{r}=11$. b) Kernel density estimate for the magnitudes of the PCA embedding vectors. c) Persistence diagram shows evidence of a single ``loop'' in the embedding. d) Estimated kernel as a function of latent positions in angular form $\theta_i = \arctan(z_i^{(2)}/z_i^{(1)})$. e) Estimated kernel as a function of latent inner product $\langle z_i,z_j\rangle$, the red dashed ellipse highlights $\hat{f}(z_i,z_j)$ in the region $\langle z_i,z_j\rangle\approx 1$. f) Evidence of a linear relationship between shortest path lengths computed from the nearest neighbour graph $\mathcal{G}$ ($y$-axis), and from the latent positions ($x$-axis).}
\end{figure}

We now consider the first hypothesis, which would be mathematically justified by assumption \ref{ass:injective}, with $\mathcal{Z}$ (say) the unit circle on $\mathbb{R}^2$ ($d_{\Zc}$ the Euclidean metric). Then indeed $\phi$ would be a homeomorphism and $\mathcal{M}$ would be topologically equivalent to a circle. Conveniently in this example, the presence of loops or holes in data point-clouds can be assessed using persistent homology techniques \cite{carlsson2009topology,chazal2021introduction}. Figure \ref{fig:cars_EDA}c) shows a persistence diagram computed from the spherically projected PCA embedding using the Python package \verb|Ripser.py| \cite{ctralie2018ripser}.  The blue dot on the horizontal dashed line is indicative of a single connected component with persists over a large range of length scales. The isolated single orange dot close to the horizontal dashed line is indicative of a single ``loop'' in the embedding, also persisting over a large range of length scales. We did not find any higher dimensional features, checking up to dimension 6, beyond which computation was prohibitive. This is all consistent with the hypothesis.

We now consider the second hypothesis, which by proposition \ref{prop:inner_prod_kernels_alt} would be mathematically justified by the mean correlation kernel being of the form $f(z,z^\prime)=g(\langle z, z^\prime\rangle)$ in some neighbourhood of $z=z^\prime$ (or equivalently $\langle z,z^\prime\rangle =1$), and $g^\prime(1)>0$ (recall proposition \ref{prop:inner_prod_kernels_alt}). Assuming that  $z_i = (\cos(\theta_i), \sin(\theta_i))$, figure \ref{fig:cars_EDA}d) shows $\hat{f}(z_i,z_j)\coloneqq \langle\zeta_{i}^{\mathrm{sp}},\zeta_{j}^{\mathrm{sp}} \rangle$, which we regard as an estimator of $f(z_i,z_j)$, plotted as a function of $\theta_i$ and $\theta_j$. The fairly constant width of the pronounced yellow/white diagonal stripe in this plot  admits the interpretation that indeed $f(z,z^\prime)=g(\langle z, z^\prime\rangle)$ in a neighbourhood of $z=z^\prime$. To examine this in more detail, figure \ref{fig:cars_EDA}e) plots values of $\hat{f}(z_i,z_j)$ against $\langle z_i,  z_j\rangle$ over all $i,j=1,\ldots,n$. The red dashed ellipse highlights that $\hat{f}(z_i,z_j)$ is approximately an increasing function of $\langle z_i, z_j \rangle$ in a neighbourhood of $z_i=z_j$, which is consistent with $g^\prime(1)>0$. Informed by \eqref{eq:sphere_isometry_alt} we thus anticipate there is isometry between $\mathcal{M}$ and $\mathcal{Z}$, up a scaling factor of $g^\prime(1)^{1/2}$. To see if the data allow for such a relationship we compute the $k$-nn graph $\mathcal{G}$ as per section \ref{subsec:nn_graph} with $k=2$. This is the natural choice for $k$ if $\mathcal{M}$ is topologically equivalent to a circle. Figure \ref{fig:cars_EDA}f) shows shortest path lengths $\mathbf{D}_{\mathcal{M}}^{ij}$ in $\mathcal{G}$ plotted against shortest path lengths around the circle, denoted $\mathbf{D}_{\mathcal{Z}}^{ij}$, over $i,j=1,\ldots,n$. The clear linear relationship is consistent with there being little deviation from isometry, up to a scaling factor, which by a straight line fit  in figure \ref{fig:cars_EDA}f) we can estimate: $g^\prime(1)\approx 3.18$.

% In summary, in this example, the data-inner products $p^{-1}\langle\Y_i,\Y_j\rangle$ are cross-correlations between time series. Through the LMM we give these inner-products a geometric interpretation, and  examining shortest paths allows us to compare and contrast the geometry of the data and underlying geography.  We find deviation from isometry in some geographic regions, implying nonstationarity.

\subsection{Single-cell transcriptomics }\label{subsec:ex_planaria}
We now revisit the planaria single-cell transcriptomics example introduced in section \ref{sec:Introduction}. Recall that here we have $p=5821$ dimensional gene expression data in $n=5000$ cells from adult planarians, and we also know cell-type labels for each of these cells, indicated by the different colours in figure \ref{fig:planaria_intro}. Adult planarians have a large number of pluripotent stem cells, known as neoblasts, that continuously differentiate into all adult cell types, resulting in a lineage tree that connects all the cells in the whole animal. We represent this lineage by a continuous tree (formally defined later) and suppose the cell types are named positions, $c_i$, on this tree.

In this context, we will first consider the hypothesis:
\begin{enumerate}
\item[1.] $\mathcal{Z}$ is a continuous tree and $\phi$ is a homeomorphism. An informal implication is: The data lie close to a tree.
\end{enumerate}
Finding this hypothesis to be tenable given the data, we will consider the stronger hypothesis:
\begin{enumerate}
\item[2.] $\mathcal{Z}$ is the lineage tree, $z_i = c_i$ are the (known) cell types, and $\phi$ is a homeomorphism. Informally: the tree represents the lineage of the cell types. 
\end{enumerate}

In graph theory, a tree is an undirected graph in which any two vertices are connected by a unique path. We consider an analogue of this concept which reflects the continuous nature of cell differentiation. Inspired by definitions in \cite{morgan1984valuations,janson2023real}, we say a metric space $\Zc$ is a \emph{continuous tree} if for any $z, z^\prime\in\Zc$ there exists a homeomorphism $\psi$ between $[0,1]$ and some subset of $\Zc$ such that $\psi(0) = z$ and $\psi(1) = z^\prime$ (this means a continuous path in $\Zc$ exists between $z$ and $z^\prime$), and all such homeomorphisms have the same image (this means the path is unique).

Following the workflow in section \ref{sec:Introduction}, dimension selection with $\rho_{\text{max}} = 50$ results in $\hat{r} = 14$. We then calculate the dimension-$\hat{r}$ embedding and its spherical projection $\zeta_1^{\mathrm{sp}},\ldots,\zeta_n^{\mathrm{sp}}$. For the remainder of this section we refer to the latter as the PCA embedding. 

We now consider the first hypothesis, which would be mathematically justified by assumption \ref{ass:injective}. Then $\mathcal{M}$ equipped with the $\ell_2$ distance also qualifies as a continuous tree, as the composition of two homeomorphisms is a homeomorphism. To gain some preliminary insight into the structure of the PCA embedding, figure \ref{fig:plan_paths}a) shows, in red, a histogram of inner products between all distinct pairs of embedding points $\zeta_i^{\mathrm{sp}},\zeta_j^{\mathrm{sp}}$. As a baseline comparison, we generated a random embedding consisting of the same number $n=5000$ points uniformly distributed on the $\hat{r}$-dimensional, unit-radius hypershere. Figure \ref{fig:plan_paths}a) shows, in black, a histogram of inner-products between all distinct pairs of points in this random embedding. We see this black histogram is symmetrical and concentrated around $0$. By contrast, the red histogram is not symmetrical and exhibits two peaks. The peak near an inner product value of $1$ indicates a substantial proportion of pairs of points  $\zeta_i^{\mathrm{sp}},\zeta_j^{\mathrm{sp}}$ which are much closer together than is observed in the random embedding. Many other pairs $\zeta_i^{\mathrm{sp}},\zeta_j^{\mathrm{sp}}$ have inner products between $-0.5$ and $0$, indicating they are spread out on the hypersphere, but not in the same way that uniformly random points are spread out. On the basis of this preliminary check we see no reason to rule out tree structure in the PCA embedding.
\begin{figure}[h!]
\includegraphics[width=1\columnwidth]{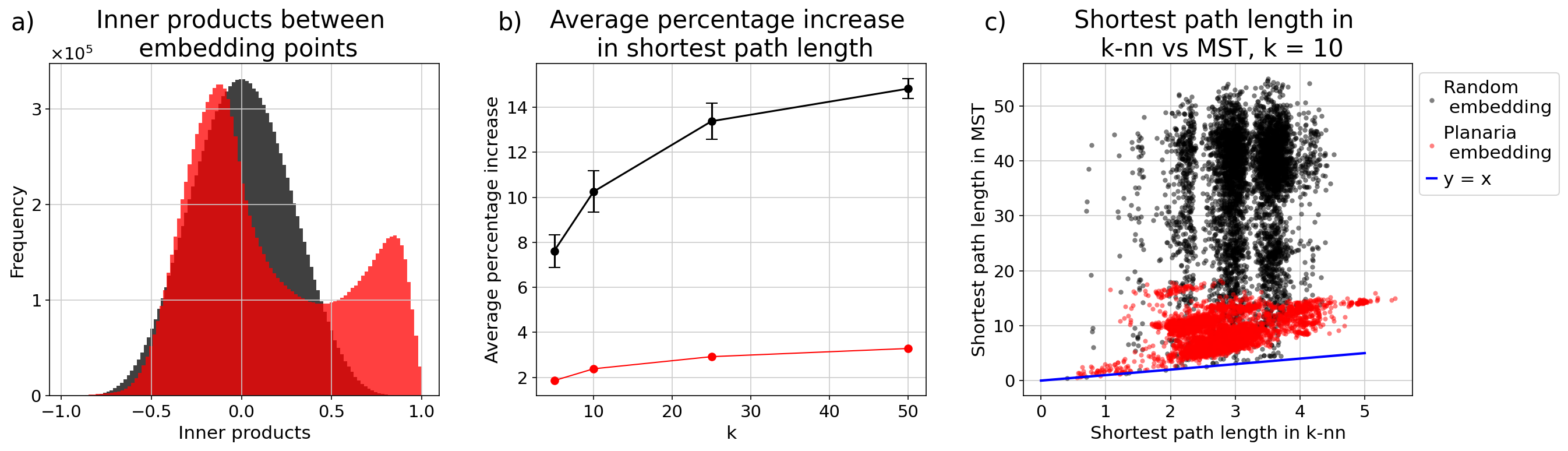}
\caption{\label{fig:plan_paths} Single-cell transcriptomics example. a) histogram of inner products between distinct points in the PCA and random embeddings. b) average percentage increase in shortest path length in the minimum spanning tree compared to the $k$-nn graph, over different values of $k$. Results for the random embedding are shown in black, over $10$ simulations with error bars indicated $2\times$standard error, c) comparing the shortest path lengths for samples in 10-nn graph and the MST.}
\end{figure}
We now quantify how ``tree-like'' the PCA embedding is in two steps. First we compute the $k$-nn graph of the PCA embedding as per section \ref{subsec:nn_graph}, and its minimum spanning tree. The latter is obtained by removing edges from the $k$-nn graph until a tree is formed, in such a way that the total edge length of the tree is minimal. Various fast algorithms for computing minimum spanning trees are available, we used the Python library \verb|NetworkX| \cite{hagberg2008exploring}. The second step is to compare shortest path lengths in the $k$-nn graph to  those in the minimum spanning tree. The shortest path length between any pair of vertices in the minimum spanning tree can only be greater than or equal to the shortest path length between those vertices in the $k$-nn graph. The percentage increase in shortest path length, when averaged over all pairs of vertices, serves as a univariate statistic which quantifies how tree-like the $k$-nn graph is. If the $k$-nn graph were a tree, this statistic would be exactly zero.

Figure \ref{fig:plan_paths}b) shows the average percentage increase in shortest path length, as a function of $k$. The red line shows the results for the PCA embedding. The black line and error bars show the same quantity computed from repeated simulations of the random embedding, serving as a baseline for comparison. We see that across all values of $k$, the average percentage increase in shortest path length is much lower for the PCA embedding than for the random embedding. This indicates that the minimum spanning tree is a close approximation to the $k$-nn graph of the PCA embedding. To take a finer-grained look, figure \ref{fig:plan_paths}c) shows shortest path lengths in the minimum spanning tree, versus in the $k$-nn graph with $k=10$, for a sample of 5000 pairs of vertices. The blue ``$y=x$'' line indicates the lower bound on path length increase which would be achieved if the $k$-nn graph were a tree. Overall, the hypothesis that the planaria data are tree-like seems tenable given the findings in figure \ref{fig:plan_paths}. 

% are consistent with the PCA embedding of the planaria data being tree-like. 

We now consider the second hypothesis. The left plot in figure \ref{fig:plan_class} shows a visualisation of the minimum spanning tree derived from the $k$-nn graph with $k=10$, with vertices coloured by the known cell type labels. This visualisation was obtained using the Scaleable Force Directed Placement graph layout algorithm \cite{hu2005efficient}. From the colouring by cell type, we see that biologically similar types, such as the three types of muscle cell, appear in localised branches of the tree.

\begin{figure}[h!]
\includegraphics[width=1\columnwidth]{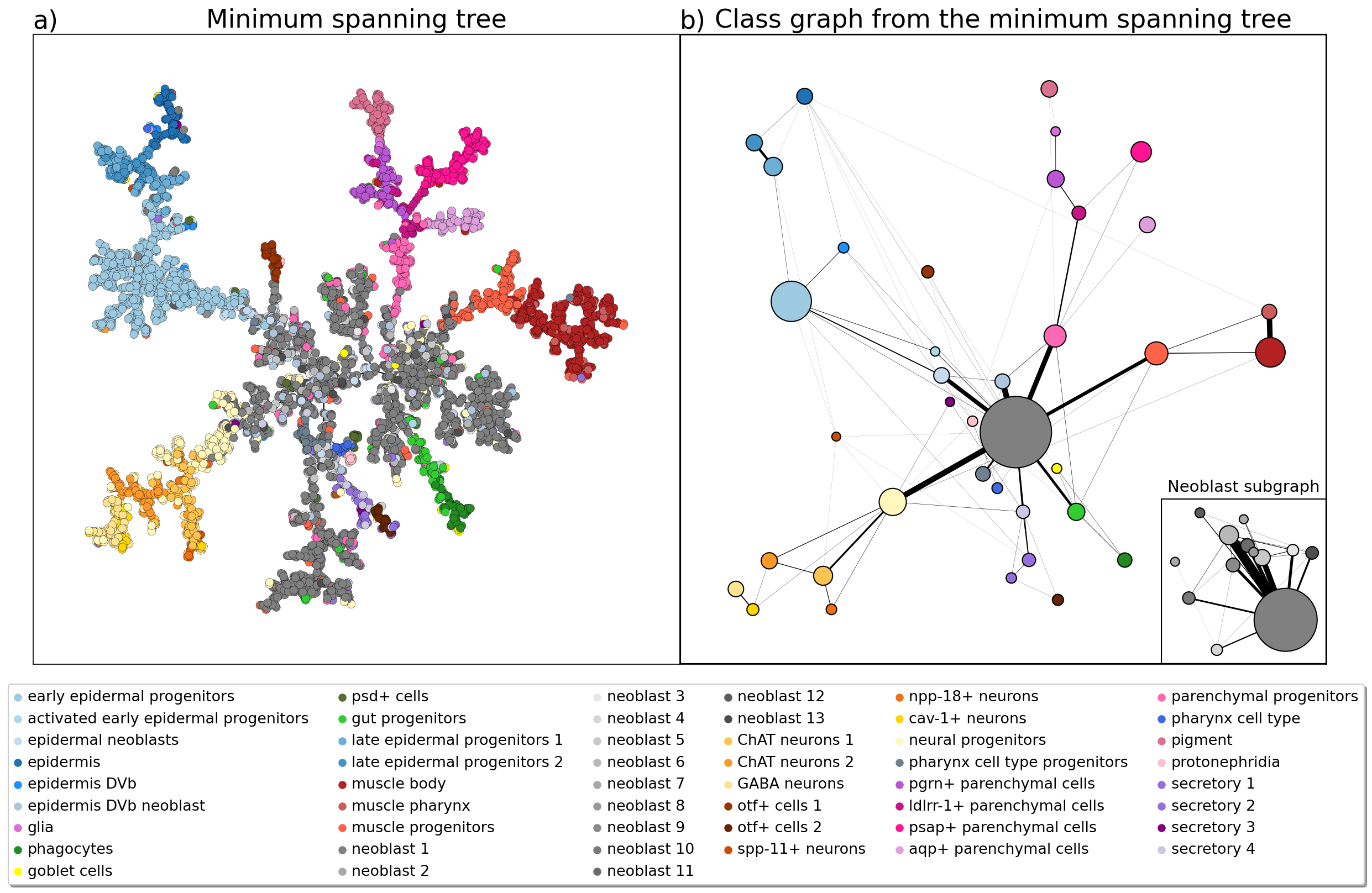}
\caption{\label{fig:plan_class}Single-cell transcriptomics example. a) minimum spanning tree computed from the spherically projected PCA embedding of the planaria data, colours indicate cell types. b) the class graph formed from the minimum spanning tree. All neoblast cell types are represented by a single dark grey node. The class subgraph consisting only of neoblast types is shown in the bottom right-hand corner inset.}
\end{figure}

We next construct a ``class graph'' which captures the relationships between cell types implied by the minimum spanning tree in figure \ref{fig:plan_class}a). In this class graph, each vertex corresponds to a cell type, and the undirected edge weight between any two vertices in the class graph is defined to be the total number of edges in the minimum spanning tree between cells of those two types. 

The class graph is shown in figure \ref{fig:plan_class}b). The size of each node represents the total number of cells of that type. The thickness of the edges reflects their weights in the class graph, although for visual clarity we do not draw some edges with very low weights.  Figure \ref{fig:plan_class} elucidates cell development, tracing the lineage from stem cells to progenitors and differentiated cell populations: neurons, muscle cells, protonephridia, epidermis, and secretory cells.

The original paper \cite{plass2018cell} provides a consolidated tree, which amalgamates various evidence types. The overall structure aligns with our nearest neighbour approach, with branches for individual known cell types, however, discrepancies exist in the form of minor variations in the differentiated cell populations. For example, cav-1+ neurons connect to ChAT neurons 1 rather than neural progenitors. The connections from the muscle pharynx to the muscle progenitors and from the epidermis to the epidermal lineage, found based on marker gene analysis in \cite{plass2018cell} coincide with the results from the nearest neighbour approach employed here. Acknowledging these differences, we refrain from delving further into minor disparities, given the current paper's intended scope.

% Additionally, a more interconnected secretory cell network is found in our class graph. In the consolidated tree of \cite{plass2018cell}, two connections are added based on marker gene analysis (genes known to be present in specific cell types): muscle pharynx to muscle progenitors and from epidermis to the epidermal lineage. In contrast, the nearest-neighbour approach employed here identifies these connections. 

For visual clarity, we draw a single node grouping together all the neoblast 1-13 cell types. The subgraph of the class graph corresponding to these neoblast types is shown in the inset of figure \ref{fig:plan_class}b), revealing a large number of neoblast 1 cells, linked by edges to most other neoblast cell types. This aligns with the results of the original authors, but contrasts with previous studies \cite{van2014single}, \cite{molina2021decoding}, which suggested distinct fates for various neoblast types. These disparities might be due to the unique ability of specialised neoblast cells to maintain pluripotency \cite{raz2021planarian} or the sensitivity of the single-cell transcriptomic method, as in \cite{plass2018cell}.

\subsection{Temperature time series}\label{subsec:ex_temps}
In this example the raw data are time series of average daily temperatures in $n=265$ towns and cities, on $p=1450$ days. The data originate from the Berkeley Earth project \cite{berkeleyearth}. Our objective is to explore the relationship between temperature deviations and geographic locations of the towns and cities. To do so we take the $i$th data vector $\Y_i$ to be the temperature time series for town or city $i$ centered about its long-run average. Thus geometry of the data point-cloud $\Y_1,\ldots,\Y_n$ as specified by the inner products $p^{-1}\langle\Y_i,\Y_j\rangle$ reflects the lag-zero cross-correlations amongst the time series.

In this context, we will first consider the hypothesis:
\begin{enumerate}
\item[1.] $\Zc$ is a geographic region, $z_i = (\text{latitute}_i,\text{longitude}_i)$ are the (known) geographical locations of towns or cities, $\phi$ is a scaled isometry. An informal implication is: geodesic distances in $\mathcal{M}$ reflect geographical distances.
\end{enumerate}
As a relaxation of the above we also ask if we can at least entertain:
\begin{enumerate}
\item[2.] $\Zc$ is a geographic region, $z_i = (\text{latitute}_i,\text{longitude}_i)$ are the geographical locations of towns or cities, $\phi$ is a scaled isometry in certain subregions.\end{enumerate}
%This will be found to hold approximately in e.g. central Europe.

Following the workflow from section \ref{sec:Introduction}, figure \ref{fig:temps_three_plots}a) shows the results of dimension selection with $\rho_{\text{max}} = 70$, figure \ref{fig:temps_three_plots}b) illustrates the variability of the magnitudes $\|\zeta_i\|_2$ of the non-projected embedding vectors, and we work henceforth with the spherically projected embedding as per section \ref{sec:spherical_proj}.

We now consider the first hypothesis, which by proposition \ref{prop:inner_prod_kernels_alt} would be mathematically justified if the temperatures on a given day were stationary processes over Earth (a `sphere'). If isometry between $\Zc$ and $\mathcal{M}$ were to hold up to a scaling factor, then the $k$-nearest neighbours of $z_i$ amongst $\{z_j;j\neq i\}$ would correspond to the $k$-nearest neighbours of $\phi(z_i)$ amongst $\{\phi(z_j);j\neq i\}$, with respect to $d_{\mathcal{M}}^{\mathrm{geo}}$. In order to see if the data are consistent with the hypothesis of isometry, we therefore compute the proportion of edges in common between the embedding $k$-nn graph $\mathcal{G}$ (as per section \ref{subsec:nn_graph}), and the geographic $k$-nn graph defined by the known locations $z_1,\ldots,z_n$. Figure \ref{fig:temps_three_plots}c) shows this proportion as a function of $k$. As a baseline to help interpret these results, we sampled $n$ points uniformly from the $\hat{r}$-dimensional unit hypersphere, derived the $k$-nn graph from these points, then computed the proportion of edges in common with the geographic $k$-nn graph. This was repeated independently $100$ times, and the resulting minimum, mean and maximum proportions of edges in common for each $k$ are shown in red and black in figure \ref{fig:temps_three_plots}c). The correspondence between $\mathcal{G}$ and the geographic $k$-nn graph is much better than under this uniform model. However, we see that as $k$ increases up to $50$, the embedding $k$-nn graph has about $70\%$ of edges in common with the geographic $k$-nn graph, but increasing $k$ further up to about $k=130$ does not increase this percentage further. This plateauing suggests isometry does not hold. 

% To explore the relationship between geometry of the data and geographic locations, we consider the hypothesis that $\Zc$ is a geographic region and $Z_i$ specifies the geographic location, in terms of latitude and longitude, of the $i$th town or city.
% We can instantiate this hypothesis in the form of assumption \ref{ass:diffeo_Rd_alt}. Now, consider any one of the towns or cities, and let $z_i\in\Zc$ be its known geographic location.

We now consider the second hypothesis. The plateauing leaves open the possibility that there may be coincidence between the embedding and geographic $k$-nn graphs in some localised areas of $\mathcal{Z}$ but not in others. Indeed isometry as in \eqref{eq:isometry-definition} or \eqref{eq:euc_isometry_alt} requires equality of shortest path lengths \emph{for all} $z,z^\prime\in\mathcal{Z}$. Figure \ref{fig:temps_knn} shows the locations of the towns and cities, and the edges in the embedding $k$-nn graph $\mathcal{G}$, with $k=5$ chosen so that according to figure \ref{fig:temps_three_plots}c) the embedding and geographic $k$-nn graphs have around $50\%$ of edges in common. We see from figure \ref{fig:temps_knn} that in some regions, especially in central Europe, edges in the embedding $k$-nn graph generally correspond to geographic proximity, but elsewhere this correspondence does not hold. For example there are edges connecting Edinburgh, U.K., to cities in Norway which are not amongst its geographically nearest neighbours. Similarly, there are edges connecting Novorossiysk, Russia, to cities on the opposite shore of the Black Sea which are not amongst its geographically nearest neighbours. Conversely, geographic proximity does not always imply presence of an edge. For example, there are no edges between Baia Mare, Romania, and two geographically close cities directly to the east, on the other side of the Carpathian mountains.

\begin{figure}[h!]
\centering
\includegraphics[trim={0 0.3cm 0 0}, width=0.95\columnwidth]{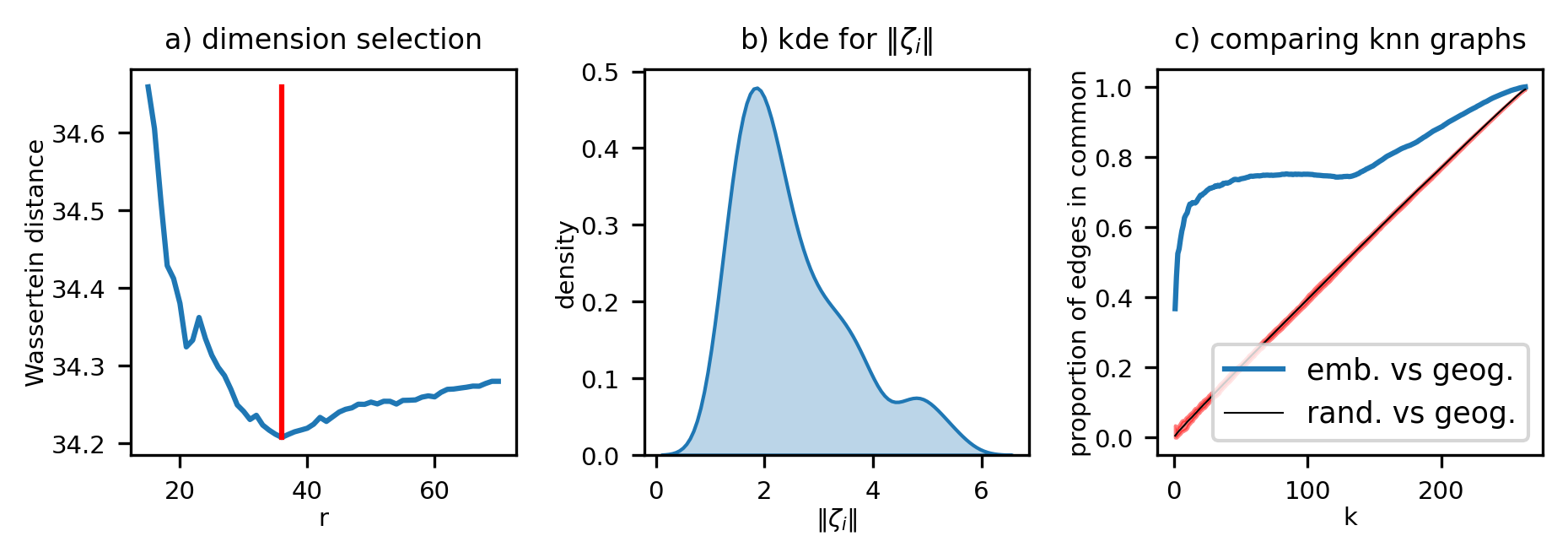}
\caption{\label{fig:temps_three_plots}Temperatures example. a) Wasserstein dimension selection; red line indicates minimum at $\hat{r}=36$. b) Kernel density estimate of the probability density of PC score magnitudes. c) The blue curve shows proportion of edges in common between embedding $k$-nn graph and geographic $k$-nn graph. The black line shows the mean proportion in common between the $k$-nn graph of a $100$ uniformly random embeddings and the geographic $k$-nn graph. The red band indicates the range between maximum and minimum proportions across these $100$ random embeddings.}
\end{figure}

Plotting the $k$-nn graph $\mathcal{G}$ in this way shows presence or absence of edges, but it doesn't convey the weight of these edges in the $k$-nn graph $\mathcal{G}$, which as per section \ref{subsec:nn_graph}, can approximate distances $d_{\mathcal{M}}^{\mathrm{geo}}$. Since the embedding is of dimension $\hat{r}=36$, it is challenging or perhaps impossible to construct a two-dimensional visualisation which faithfully conveys all aspects of its geometry.  However, the visualisation task is much simpler if we choose some town or city, and then visualise the shortest paths in the embedding $k$-nn graph from that city to all other cities --- the graph consisting of the union of all such paths is sometimes called a \emph{shortest path tree}. 

Figure \ref{fig:temps_paths} shows the shortest paths in $\mathcal{G}$ from Tallinn, Estonia, to all other towns and cities. Each such path is a sequence of towns or cities, and is visualised as a spline with knot points given by the locations of these towns and cities, with colour indicating length. Tallinn was chosen because of the different relationships between these shortest paths and geographic shortest paths which can be seen in different regions: the shortest paths in $\mathcal{G}$ which terminate at some towns and cities in central Europe, to the south-west of Tallinn, resemble geographic shortest paths, indicating a geometric relationship not far removed from isometry. By contrast, the red dots in figure \ref{fig:temps_paths} highlight the shortest path in $\mathcal{G}$ from Tallinn to Tripoli, Libya. This path passes through Sweden, Norway, the U.K., Ireland, France, Spain, back to France and then Italy. Clearly, this is not the geographically shortest path from Tallinn to Tripoli, indicating a strong deviation from isometry in these regions. Recalling from section \ref{sec:isometry} the relationship between weak stationarity and isometry, this deviation from isometry implies a pronounced lack of stationarity (with respect to geographic location) in these regions. This prompts us to consider what factors might disrupt temperature correlations along the geographically shortest path: for the case of Talinn to Tripoli, it seems reasonable to conjecture that changes in altitude, e.g. the Alps, and the Adriatic and Mediterranean seas might be such factors. We stress that our analysis by no means formally assesses the evidence that such factors are at play. Rather, it is through the process of inspecting figure \ref{fig:temps_paths} that we are led simply to consider such ideas as a step in data exploration. 
 
\begin{figure}[h!]
\includegraphics[width=1\columnwidth]{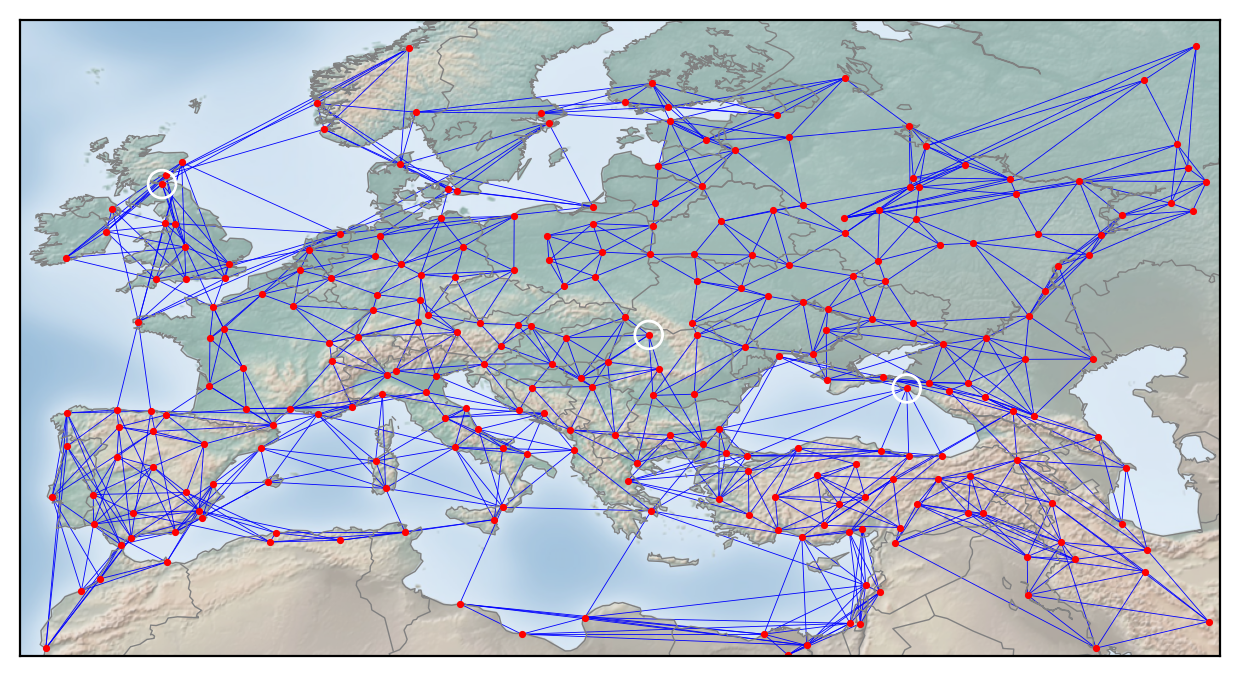}
\caption{\label{fig:temps_knn}Temperatures example. Locations of towns and cities are shown in red. The blue lines correspond to edges in the embedding $k$-nn graph $\mathcal{G}$, with $k=5$. The white circles highlight, from west to east: Edinburgh, U.K.; Baia Mare, Romania; and Novorossiysk, Russia. }
\end{figure}
\begin{figure}[h!]
\includegraphics[width=1\columnwidth]{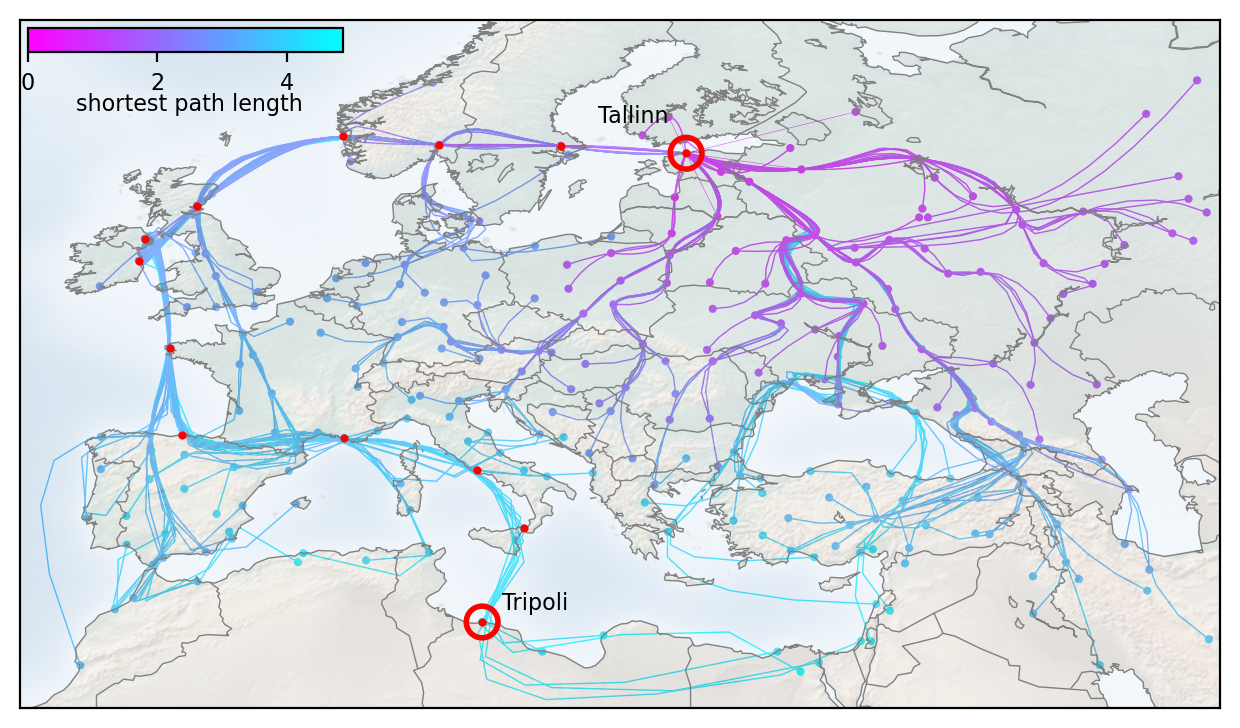}
\caption{\label{fig:temps_paths}Temperatures example. Shortest paths in the embedding $k$-nn graph $\mathcal{G}$ from Tallinn, Estonia, to all other towns and cities. Each shortest path is visualized as a spline, with knot points given by the geographic locations of its constituent towns and cities. The red dots highlight the shortest path from Tallinn to Tripoli, Libya.}
\end{figure}

\section{Connections and conclusions}\label{sec:discussion}

\paragraph{Geometric representation of high dimensional data.} In their seminal JRSSB paper, \citet{hall2005geometric} introduced the perspective that, if  $\Y_i$ and $\Y_j$ are i.i.d. random vectors whose elements satisfy suitable weak dependence and moment conditions, then $p^{-1}\|\Y_i-\Y_j\|^2$ converges to a constant as $p\to\infty$. This leads to a conclusion that i.i.d. high-dimensional data vectors tend to lie deterministically at the vertices of a simplex.  The LMM deviates from the assumption of i.i.d. data vectors; under the LMM, the ``noise-free'' data vectors $\Y_i-\sigma \E_i\equiv[X_1(Z_i) \cdots X_p(Z_i)]^{\top}$, $i=1,\ldots,n$, are exchangeable but not independent. This dependence, combined with the latent variables $Z_1,...,Z_n$, is key to the emergence of manifold structure in high dimensional data under the LMM. In this way, the LMM and our theoretical results extend and enrich the geometric perspective of \citet{hall2005geometric}, broadening the scope of high dimension low sample size (HDLSS) asymptotics \citep{shen2016general,aoshima2018survey}. 

\paragraph{Reflecting on the Manifold Hypothesis.}
Conventional interpretation of the Manifold Hypothesis as per the quote from \citep{cayton2005algorithms} in section \ref{sec:Introduction} is that data vectors $\Y_1,\ldots,\Y_n\in \mathbb{R}^p$ are samples from some distribution supported on a manifold embedded in $\mathbb{R}^p$, perhaps subject to noise disturbances. Our analysis of the LMM in section \ref{sec:connecting} provides a more nuanced perspective: $\Y_1,\ldots,\Y_n$ are noisy, random projections of samples on a manifold $\mathcal{M}$; the manifold itself is a high-dimensional distortion of some latent domain $\Zc$ and arises due to correlation over $\Zc$. Under appropriate assumptions, $\mathcal{M}$ is homeomorphic or isometric to $\Zc$. 
% Conventional interpretation of the Manifold Hypothesis as per the quote from \citep{cayton2005algorithms} in section \ref{sec:Introduction} is that data vectors $\Y_1,\ldots,\Y_n\in \mathbb{R}^p$ are samples from some distribution supported on a manifold embedded in $\mathbb{R}^p$, perhaps subject to noise disturbances. Our analysis of the LMM in section \ref{sec:connecting} provides a more nuanced perspective: $\Y_1,\ldots,\Y_n$ are noisy projections of corresponding points $\phi(Z_1),\ldots,\phi(Z_n) \in \mathcal{M}\subset\ell_2$, where $\mathcal{M}$ is defined implicitly via the mean correlation kernel of the LMM, and $Z_1,\ldots,Z_n$ are samples in the latent metric space $\Zc$. Under appropriate assumptions, $\mathcal{M}$ is homeomorphic or isometric to $\Zc$.  We have seen that by using PCA to estimate $\phi(Z_1),\ldots,\phi(Z_n)$ from data and then constructing a nearest neighbour graph, we can discover or explore hypothesis about the topological or geometric structure of $\Zc$. 

\paragraph{Infinite exchangeable arrays.} In how much generality is this perspective applicable? Inspired by remarks of \cite{udell2019big} in the context of latent variable models of low-rank matrices, we note the basic structure of the LMM,
\begin{equation}\label{eq:lmm_conclusion}
\Y_{ij} = X_{j}(Z_i)+\sigma\E_{ij},
\end{equation}
resembles a representation formula for exchangeable arrays due to \citet{aldous1981representations}: if $\Y$ is \emph{any} infinite two-dimensional array of random variables such that permutations of its rows or columns do not alter the distribution of $\Y$, then there exists a function $h$ such that the following equality in distribution holds
\begin{equation}\label{eq:aldous}
\Y_{ij}\stackrel{d}{=} h(\xi,Z_i,X_j,E_{ij})
\end{equation}
where $\xi$ and the $Z_i$'s, $X_j$'s and  $E_{ij}$'s are i.i.d. $\mathcal{U}[0,1]$-distributed random variables. Putting aside the fact that in the LMM the rows of $\Y$ are exchangeable but the columns need not be, the resemblance between  \eqref{eq:lmm_conclusion} and \eqref{eq:aldous} indicates that the LMM is rather general, albeit constrained to an additive form of error. The ability of PCA to extinguish noise, as characterised in theorem \ref{thm:consistency_summary}, seems closely tied to this additive structure.

\paragraph{PCA in high dimensions.}
The behaviour of PCA and principal component scores in high-dimensions has been the subject of intensive theoretical study, e.g.,  \citep{paul2007asymptotics, johnstone2009consistency, jung2009pca, yata2009pca, lee2010convergence,yata2012effective,jung2012boundary,shen2013surprising,shen2016general,hellton2017and}. A central theme in these works is analysis of the eigenvectors of the sample covariance matrix $n^{-1}\Y^\top\Y\in\mathbb{R}^{p\times p}$, which make up the columns of the matrix $\mathbf{V}_{\Y}$ appearing in \eqref{eq:pc_scores_defn}. It is usually assumed that the data follow a spiked covariance model (a special case of the LMM in which $\Zc$ is Euclidean, $f$ is linear and $X_1,\ldots,X_p$ are deterministic - see appendix \ref{sec:special_cases_of_LMM} for details), with consideration given to various scaling relationships involving $p,n$ and population covariance eigenvalues. In some situations with growing dimension $p\to\infty$, the sample covariance eigenvectors are inconsistent estimators of their population counterparts, such as when the eigenvalues are constant, $n$ is growing and $p/n\to c>0$ \citep{johnstone2009consistency}, or when $n$ is fixed and the eigenvalues grow sublinearly as $p\to\infty$ \citep{jung2009pca, jung2012boundary}.

Theorem \ref{thm:consistency_summary} 
 addresses high-dimensionality, but its proof does not entail establishing consistency of $\mathbf{V}_{\Y}$. Instead it starts with an elementary linear algebra argument (lemma \ref{lem:PC_identity} in section \ref{sec:Proof-and-supporting_consistency}) which shows that  $p^{-1/2}\mathbf{Y}\mathbf{V}_{\mathbf{Y}}=\mathbf{U}_{\mathbf{Y}}\boldsymbol{\Lambda}_{\mathbf{Y}}^{1/2}$,
where the columns
of $\mathbf{U}_{\mathbf{Y}}\in\mathbb{R}^{n\times r}$ are orthonormal
eigenvectors of $p^{-1}\Y\Y^{\top}\in\mathbb{R}^{n\times n}$ with associated eigenvalues on the diagonal of $\boldsymbol{\Lambda}_{\Y}$. The $\sqrt{n/p}$ term in \eqref{eq:consistency_summary} relates to the concentration behavior of the $n\times n$ matrix $p^{-1}\Y \Y^\top$ about its conditional expectation: $\mathbb{E}\left[(p^{-1}\Y\Y^\top)_{ij}|Z_i,Z_j\right]=p^{-1}\mathbb{E}\left[\langle\Y_i,\Y_j\rangle|Z_i,Z_j\right]$, c.f. \eqref{eq:Y_inner_prod_cond_exp}. The $1/\sqrt{n}$ term in \eqref{eq:consistency_summary} concerns approximations to certain integrals with respect to $\mu$, based on the samples $Z_1,\ldots,Z_n$, which arise when relating the rows of $\mathbf{U}_{\mathbf{Y}}\boldsymbol{\Lambda}_{\mathbf{Y}}^{1/2}$ to $\phi(Z_1),\ldots,\phi(Z_n)$.

The proof of theorem \ref{thm:uniform_consistency}, from which theorem \ref{thm:consistency_summary}  is derived, relies heavily on matrix decomposition techniques used by \citet{lyzinski2016community} in the study of spectral embedding of random graphs under a random dot product model.  The uniform (in $i=1,\ldots,n$) nature of theorem \ref{thm:consistency_summary} is directly inspired by the uniform consistency result of \citet{lyzinski2016community}[Thm. 15], which is an instance of convergence with respect to the $2\to\infty$ matrix norm, studied in detail by \citet{cape2019two}. We note more generally that singular vector estimation under low-rank assumptions is an active area of research. As a recent example, \citet{agterberg2022entrywise} obtained finite sample bounds and a Berry-Esseen type theorem for singular vectors under a model in which the signal is a deterministic low-rank matrix and heteroskedasticity and dependence is allowed in additive sub-Gaussian noise. 

Truncated spectral embedding of graphs under a model with an infinite rank kernel was studied by \cite{tang2013universally}, but their results concern the Frobenius norm, hence a weaker, non-uniform measure of error than the $\max_{i=1,\ldots,n}$ error in theorem 1. It remains to be seen if a uniform consistency result similar to theorem 1 can be obtained for the LMM with an infinite rank kernel.

\paragraph{Gaussian Process Latent Variable Models.}
  
The Gaussian
Process Latent Variable Model (GPLVM)  was devised by \citet{lawrence2003gaussian, lawrence2005probabilistic} in order to formulate  dimension reduction as a statistical inference problem. The GPLVM can be viewed as a special case of the LMM in which $\Zc$ is a subset of $\mathbb{R}^d$, the random functions $X_1,\ldots,X_p$ are i.i.d., zero-mean Gaussian processes, and the elements of the noise matrix $\E$ are Gaussian.   Under these assumptions the formula \eqref{eq:Phi_W_expansion} from proposition \ref{prop:Phi_W_expansion}, written in matrix form, is: 
$$
\Y \stackrel{m.s.} {=}p^{1/2}\Phib\W^\top +\sigma \E,
$$ 
with $\Phib\coloneqq[\phi(Z_1)|\cdots|\phi(Z_n)]^\top\in\mathbb{R}^{n\times r}$ and the elements of $\E$ and $p^{1/2}\W\in\mathbb{R}^{r\times p}$ are i.i.d. $\mathcal{N}(0,1)$. The latter property of $\mathbf{W}$ is equivalent to $X_1,\ldots,X_p$ being i.i.d., zero-mean Gaussian processes. After integrating out $\W$ and $\E$ analytically, the $p$ columns of $\Y$ are i.i.d. with common distribution $\mathcal{N}(\mathbf{0}_n,\Phib\Phib^\top +\sigma^2 \mathbf{I}_n)$. \citet[App. B]{lawrence2005probabilistic} considered maximum likelihood estimation of $\Phib$ when $r<\infty$, showing that that for the GPLVM and any orthogonal matrix $\mathbf{Q}$,
\begin{equation}\label{eq:mle}
\widehat{\Phib}^{\mathrm{MLE}}\coloneqq \mathbf{U}_\Y (\boldsymbol{\Lambda}_\Y - \sigma^2\mathbf{I}_n)^{1/2}\mathbf{Q}^\top
\end{equation}
is a maximum likelihood estimator of $\Phib$. Here $\boldsymbol{\Lambda}_\Y$ and $\mathbf{U}_\Y$ are as in the above discussion of PCA in high dimensions. To make the connection to the PCA embedding $\zeta_1,\ldots,\zeta_n$ recall that 
\begin{equation}\label{eq:pca_recap}
p^{-1/2}[\zeta_1|\cdots|\zeta_n]^\top \equiv p^{-1/2}\Y \mathbf{V}_{\mathbf{Y}} \equiv \mathbf{U}_\Y \boldsymbol{\Lambda}_\Y^{1/2},
\end{equation}
where the columns of $\mathbf{V}_{\mathbf{Y}}$ are orthonormal eigenvectors of $\Y^\top\Y$.
Comparing \eqref{eq:mle} to \eqref{eq:pca_recap}, we can interpret the rows of $p^{1/2}\widehat{\Phib}^{\mathrm{MLE}}$ as a modified PCA embedding derived from eigen-decomposition of $p^{-1}\Y\Y^\top - \sigma^2\mathbf{I}_n$ instead of $p^{-1}\Y\Y^\top$.  Computing $\widehat{\Phib}^{\mathrm{MLE}}$ in practice clearly requires $\sigma^2$ to be assumed known and $\boldsymbol{\Lambda}_\Y\succeq \mathbf{I}_r\sigma^2$. It may be possible to use some of our theoretical results to study consistency of $\widehat{\Phib}^{\mathrm{MLE}}$, but we leave this for future research.

Under the assumption that the kernel $f$ belongs to a given parametric family \citet{lawrence2005probabilistic}
proposed maximum a-posteriori estimation of $Z_{1},\ldots,Z_{n}$ as a form of nonlinear dimension reduction. Further developments of the GPLVM include hierarchical structures \citep{lawrence2007hierarchical}, 
variational inference \citep{titsias2010bayesian,JMLR:v17:damianou16a},  connections to Locally Linear Embedding \citep{lawrence2012unifying, roweis2000nonlinear}, 
pseudo-marginal MCMC techniques \citep{gadd2021pseudo}, and handling computational scaleability and missing data \citep{lalchand2022generalised}.

\paragraph{Linear then nonlinear dimension reduction.}
Nonlinear dimension reduction techniques are designed to extract low-dimensional
structure from data for purposes of exploration and visualisation.
These methods were pioneered by \citet{tenenbaum2000global} and \citet{roweis2000nonlinear},
who devised Isomap and Local Linear Embedding, respectively; subsequent
contributions include Semi-definite Embedding \citep{weinberger2004learning};
latent variable-based methods, \citep{lawrence2003gaussian,saul2020tractable};
Diffusion Maps \citep{coifman2005geometric}, Laplacian and Hessian
Eigenmaps \citep{belkin2003laplacian,donoho2003hessian}; Stochastic
Neighbour Embedding (SNE and $t$-SNE) \citep{hinton2002stochastic,van2008visualizing}
and Uniform Manifold Approximation and Projection (U-MAP). Several
such methods are easily accessible through the massively popular Python
package \verb|scikit-learn| \citep{scikit-learn}, and their impact
is exemplified by the fact that, at the time of writing, the $t$-SNE
paper of \citet{van2008visualizing} has over $48,000$  citations
according to Google Scholar. Each of these techniques work on different
principles, but in broad terms, they take as input a set of points
in high-dimensional Euclidean space, and output a set of points in
low-dimensional Euclidean space in a way which is designed to minimise some measure of distortion of pairwise distances or inner-products.

It has been advocated in the literature to reduce data to tens or hundreds of dimensions using PCA as a preprocessing step, before applying nonlinear
dimension reduction to obtain a two or three-dimensional representation. For example, in the context of $t$-SNE, 
\citet{van2008visualizing} state \emph{``This speeds up the computation
of pairwise distances between the data points and suppresses some noise
without severely distorting the interpoint distances''}. Similar
recommendations are given in  \citep{van2014accelerating, kobak2019art, saul2020tractable}.
Up until now however, there has been no detailed or rigorous statistical justification for this pre-processing. 

Our model and theory explain why a) the manifold might be concentrated within a low-dimensional subspace (e.g. section~\ref{sec:smoothness}); and b) why applying PCA can give us an extremely sharp view (i.e., uniformly consistent, Theorem~\ref{thm:consistency_summary}). This reinforces the message that nonlinear dimensionality reduction techniques need not be viewed as an alternative to PCA, but rather that the combination of the two, i.e., \emph{linear then nonlinear dimensionality reduction}, may be particularly effective. 

Incidentally, we recommend keeping the PCA embedding at hand to retain access to \emph{ambient} properties of the manifold, for example, to estimate reach \citep{aamari2019estimating}, inspect inner products (figure~\ref{fig:plan_paths}), or evaluate the fidelity of low-dimensional visualisations such as t-SNE.

% Our model and theory fills that gap, reinforcing the message that nonlinear dimensionality reduction techniques need not be viewed as an alternative to PCA, but rather
% that the combination of the two, i.e., \emph{linear then nonlinear dimensionality reduction}, may be particularly effective. 

\paragraph{Exploratory Data Analysis.} What are the limitations of the data analysis workflow we have proposed? This workflow is intentionally generic, and suitable for preliminary exploration of high-dimensional data and hypotheses about the data generating mechanism. It could serve as a first step before more detailed confirmatory analysis, in order to quantify uncertainty, perform formal hypothesis testing, or fit a parametric model, and so forth, but it clearly does not include those functionalities. We believe the methodology aligns with Tukey's philosophy of `exploratory data analysis' \citep{edatukey} although of course it is not `model-free'. There is no contradiction there, and indeed several works have warned of the dangers of conflating these ideas \citep{hullman2021designing}.

\bibliographystyle{plainnat}
\bibliography{refs,dimension_selection}
\newpage

\appendix
\section{Supporting results for section \ref{sec:Model}}\label{sec:Supporting-results-for-Model}

The following version of Mercer's theorem can be found
in \citep[Thm 4.49]{steinwart2008support}.
\begin{thm}
[Mercer's theorem] \label{thm:mercer}Let $\mathcal{Z}$ be a compact
metric space and let $f:\mathcal{Z}\times\mathcal{Z}\to\mathbb{R}$,
be a symmetric, positive semi-definite, continuous function. Let $\mu$
be a finite Borel measure supported on $\mathcal{Z}$. Then there
exists a countable collection of nonnegative real numbers $(\lambda_{k}^{f})_{k\geq1}$,
$\lambda_{1}^{f}\geq\lambda_{2}^{f}\geq\ldots$ and $\mathbb{R}$-valued
functions $(u_{k}^{f})_{k\geq1}$ which are orthonormal in $L_{2}(\mu)$,
such that:
\[
f(z,z^{\prime})=\sum_{k=1}^{\infty}\lambda_{k}^{f}u_{k}^{f}(z)u_{k}^{f}(z^{\prime}),\quad z,z^{\prime}\in\mathcal{Z},
\]
where the convergence is absolute and uniform.
\end{thm}

\section{Special cases of the LMM}\label{sec:special_cases_of_LMM}

\subsubsection*{Spiked covariance model}
The spiked covariance model \cite{johnstone2001distribution, paul2007asymptotics} is the de facto standard model under which to study the theoretical properties of PCA, and is derived as follows. Let $\mathbf{X}\in\mathbb{R}^{n\times p}$ be a matrix of random variables such that $\mathbb{E}[\mathbf{X}^\top\mathbf{X}]$ has rank $r$. Consider the eigendecomposition $n^{-1}\mathbb{E}[\mathbf{X}^\top\mathbf{X}]=\mathbf{V}\boldsymbol{\Lambda}\mathbf{V}^\top$, where $\mathbf{V}\in\mathbb{R}^{p\times r}$, and define $\mathbf{Z}\coloneqq\mathbf{X}\mathbf{V}\boldsymbol{\Lambda}^{-1/2}$. We have 
\begin{equation}
\mathbb{E}\left[\mathbf{Z}^{\top}\mathbf{Z}\right]=n\mathbf{I}_{r},\qquad\mathbf{V}^{\top}\mathbf{V}=\mathbf{I}_{r},\label{eq:spiked-cov-conditions}
\end{equation}
and $\mathbf{X}=\mathbf{Z}\boldsymbol{\Lambda}^{1/2}\mathbf{V}^\top $, a.s., where the latter equality can be checked by verifying
$$\mathbb{E}[\|\mathbf{X}-\mathbf{Z}\boldsymbol{\Lambda}^{1/2}\mathbf{V}^\top \|_{\mathrm{F}}^2]=\mathrm{tr}\,\mathbb{E}[(\mathbf{X}-\mathbf{Z}\boldsymbol{\Lambda}^{1/2}\mathbf{V}^\top )^\top(\mathbf{X}-\mathbf{Z}\boldsymbol{\Lambda}^{1/2}\mathbf{V}^\top )]=0.$$

The spiked covariance model takes the form:
$$
\Y =\mathbf{Z}\boldsymbol{\Lambda}^{1/2}\mathbf{V}^\top + \sigma\mathbf{E},
$$
where the elements
of $\mathbf{E}\in\mathbb{R}^{n\times p}$ are usually assumed to be zero-mean,
unit variance and uncorrelated.
The rows of $\mathbf{Z}\in\mathbb{R}^{n\times r}$ are  called individual-specific
random effects, and are usually assumed to be i.i.d. The following proposition shows that a spiked covariance model of precisely this form is a special case of the LMM.

% $\boldsymbol{\Lambda}=\mathrm{diag}(\lambda_{1},\ldots,\lambda_{r})$
% where $\lambda_{1},\ldots,\lambda_{r}$ are positive, the columns
% of $\mathbf{V}\in\mathbb{R}^{p\times r}$ are fixed, and 

% we have:
% \begin{equation}
% \mathbb{E}\left[\frac{1}{n}\mathbf{Y}^{\top}\mathbf{Y}\right]=\mathbf{V}^{\top}\mathbf{\boldsymbol{\Lambda}}\mathbf{V}+\sigma^{2}\mathbf{I}_{p},\qquad\mathbb{E}\left[\left.\frac{1}{p}\mathbf{Y}\mathbf{Y}^{\top}\right|\mathbf{Z}\right]=\mathbf{Z}\mathbf{\boldsymbol{\Lambda}}\mathbf{Z}^{\top}+\sigma^{2}\mathbf{I}_{n}.\label{eq:spiked-cov-identities}
% \end{equation}
% We see from the first
% identity in (\ref{eq:spiked-cov-identities}) that $\mathbf{V}^{\top}\mathbf{\boldsymbol{\Lambda}}\mathbf{V}\in\mathbb{R}^{p\times p}$
% can be interpreted as the population covariance matrix associated
% with the $n$ ``noise-free'' data vectors making up the rows of $\Y-\sigma\mathbf{E}=\mathbf{Z}\boldsymbol{\Lambda}^{1/2}\mathbf{V}^{\top}$. The second identity in (\ref{eq:spiked-cov-identities}) can be viewed as an instance of \eqref{eq:Y_inner_prod_cond_exp} once we establish the spiked covariance model is a special case of the LMM, as per the following proposition.
\begin{prop}
\label{prop:spiked_cov_as_lmm}For any $r<\infty$, let the rows of $\mathbf{Z}\in\mathbb{R}^{n\times r}$ be i.i.d. random vectors such that the
first equality in (\ref{eq:spiked-cov-conditions}) holds, let $\boldsymbol{\Lambda}=\mathrm{diag}(\lambda_{1},\ldots,\lambda_{r})$
where $\lambda_{1},\ldots,\lambda_{r}$ are any strictly positive
real numbers and let $\mathbf{V}=[v_{1}|\cdots|v_{p}]^{\top}\in\mathbb{R}^{p\times r}$
be any deterministic matrix such that the second equality in (\ref{eq:spiked-cov-conditions})
holds. Then, if $\Y$ follows the Latent Metric Model
specified by: 
\begin{equation}
\mathcal{Z}\subset\mathbb{R}^{r},\qquad[Z_{1}|\cdots|Z_{n}]^{\top}\coloneqq\mathbf{Z}\qquad X_{j}(z)\coloneqq\left\langle v_{j},\boldsymbol{\Lambda}^{1/2}z\right\rangle ,\label{eq:spiked_covariance_as_lmm_Z_and_V}
\end{equation}
the mean correlation kernel associated with this Latent Metric Model
is: 
\[
f(z,z^{\prime})=\frac{1}{p}\left\langle z,\boldsymbol{\Lambda}z^{\prime}\right\rangle ,
\]
which has rank $r$, $\lambda_{k}^{f} = \lambda_{k}/p$,
$[u_{1}^{f}(z)\,\cdots\,u_{r}^{f}(z)]^{\top} =  z\in\mathbb{R}^{r}$ and the following identity holds:
$$\Y = \mathbf{Z}\boldsymbol{\Lambda}^{1/2}\mathbf{V}^\top + \sigma\mathbf{E}.
$$
\end{prop} 
\begin{proof}
The claimed expression for $f(z,z^\prime)$ holds by substituting the definition of $X_j(z)$ in \eqref{eq:spiked_covariance_as_lmm_Z_and_V} into the definition $f(z,z^\prime)\coloneqq{p^{-1}}\sum_{j=1}^p\mathbb{E}[X_j(z)X_j(z^\prime)]$ and using the assumption of the proposition that $\mathbf{V}^\top\mathbf{V}=\mathbf{I}_r$. The eigenfunctions $[u_{1}^{f}(z)\,\cdots\,u_{r}^{f}(z)]^{\top} =  z\in\mathbb{R}^{r}$ are orthonormal due to the assumption $\mathbb{E}\left[\mathbf{Z}^{\top}\mathbf{Z}\right]=n\mathbf{I}_{r}$ and the i.i.d. nature of the rows of $\mathbf{Z}$. The expression for $\mathbf{Y}$ in the statement holds by substituting \eqref{eq:spiked_covariance_as_lmm_Z_and_V} into the definition of $\mathbf{Y}$ under the LMM, i.e., $\mathbf{Y}_{ij}=X_j(Z_i)+\sigma\mathbf{E}_{ij}$.
\end{proof}
The relationship between the spiked covariance model (SCM) and the LMM can thus be summarised as follows:
\begin{itemize}[leftmargin=0.4cm]
\item  the metric space $(\mathcal{Z},d_{\Zc})$ in the LMM generalizes the Euclidean
domain of individual-specific random effects in the spiked covariance model;
\item the eigenfunctions $u_{k}^{f}$, $k\geq1$, in the LMM generalise
the linear dependence on individual-specific random effects in the
SCM;
\item the random functions $X_{j}$, $j=1,\ldots,p$, in the LMM generalise the  deterministic, linear functions $\left\langle v_{j},\boldsymbol{\Lambda}^{1/2}z\right\rangle $,
$j=1,\ldots,p$, which in light of (\ref{eq:spiked_covariance_as_lmm_Z_and_V})
are implicit in the SCM;
\item the LMM allows for possibly infinite rank, generalising the finite-rank nature of the SCM. 
\end{itemize}

% \subsubsection*{Gaussian Process Latent Variable model}

% When $\Zc$ is a subset of $\mathbb{R}^d$ and $X_1,\ldots,X_p$ are independent and identically distributed Gaussian processes, the LMM reduces to the Gaussian Process Latent Variable model of \citet{lawrence2003gaussian, lawrence2005probabilistic}. In this case the elements of the matrix $\W$ are independent and identically distributed $\mathcal{N}(0,1)$ and the aforementioned authors derive a likelihood function with these variables out. Assuming $f$ belongs to a given parametric family, e.g., a radial
% basis function kernel, \citet{lawrence2005probabilistic}
% proposed maximum a-posteriori estimation of $Z_{1},\ldots,Z_{n}$, parameters of the kernel and $\sigma^{2}$ using a gradient method. \citet{titsias2010bayesian} proposed alternative variational methods with enable model assessment.  \citet{lawrence2012unifying} derived a Gaussian Markov random
% field model related to a GPLVM through which Locally Linear Embedding \cite{roweis2000nonlinear} has a statistical interpretation.

\subsubsection*{Finite mixture model}
Consider the case where $\mathcal{Z}$ has finitely many elements,
say $\mathcal{Z}=\{1,\ldots,m\}$. For the following discussion it
is not important that we take these elements to be the numbers $1,\ldots,m$,
any $m$ distinct abstract elements will do. In this situation the LMM is a form of finite
mixture model with random mixture centres. Indeed we see from:
\[
\mathbf{Y}_{ij}=X_{j}(Z_{i})+\sigma\mathbf{E}_{ij}
\]
that $[X_{1}(z)\,\cdots\,X_{p}(z)]$ can be interpreted as the $p$-dimensional
random centre of a mixture component labeled by $z\in\mathcal{Z}$,
and the latent variable $Z_{i}$ indicates which mixture component
the $i$th row of the data matrix $\mathbf{Y}$ is drawn from. The
simple form of the noise in the LMM constrains
the generality of this mixture model: recall the elements of $\mathbf{E}$
are independent across columns; elements in the same column but distinct
rows are uncorrelated; all elements are unit variance. %We leave more general assumptions on $\mathbf{E}$ for future investigation.

To make $\mathcal{Z}$ into a metric space we  consider the discrete
metric $d_{\mathcal{Z}}(z,z^{\prime})\coloneqq0$ for $z=z^{\prime},$
otherwise $d_{\mathcal{Z}}(z,z^{\prime})\coloneqq1$. The kernel $f$
is specified by the matrix $\mathbf{F}\in\mathbb{R}^{m\times m}$
with entries 
\[
\mathbf{F}_{kl}\coloneqq\frac{1}{p}\sum_{j=1}^{p}\mathbb{E}[X_{j}(k)X_{j}(l)],\quad k,l\in\{1,\ldots,m\}.
\]
In this situation \ref{ass:cont_covar} and \ref{ass:finite rank} hold immediately, and $r\leq m$.

Topological equivalence of $\mathcal{M}$ and $\mathcal{Z}$ in this
situation would mean that $\mathcal{M}$ consists of $m$ distinct
points $\{\phi(1),\ldots,\phi(m)\}$, each associated with exactly
one element of $\mathcal{Z}$. If such topological equivalence were to hold then theorem \ref{thm:consistency_summary}
would tell us that the PCA embedding vectors will be clustered around the
$m$ distinct points $\{\mathbf{Q}^{-1}\phi(1),\ldots,\mathbf{Q}^{-1}\phi(m)\}$,
with specifically $p^{-1/2}\zeta_{i}$ being close to $\mathbf{Q}^{-1}\phi(Z_{i})$.

To verify topological equivalence it remains to check \ref{ass:injective} holds. To this end, suppose that  $r=m$, i.e. $\mathbf{F}$
is full rank. Then it is not possible that any two rows of $\mathbf{F}$
are identical. That is, for $k,l\in\{1,\ldots m\}$ such that $k\neq l$
, there must exist some $\xi\in\{1,\ldots,m\}$ such that $f(k,\xi)=\mathbf{F}_{k\xi}\neq\mathbf{F}_{l\xi}=f(l,\xi)$. Thus assumption \ref{ass:injective} is satisfied and hence $\mathcal{M}$ is topologically equivalent to $\mathcal{Z}$ if $r=m$. 
%Noting that in this setting where $\mathcal{Z}=\{1,\ldots,m\}$ and $d_{\mathcal{Z}}$ is the discrete metric, the continuity of $\phi$ and its inverse as maps between $(\mathcal{Z},d_{\mathcal{Z}})$ and $(\mathcal{M},d_{\mathcal{M}})$ are automatic, we conclude that when $r=m$, $\mathcal{M}$ is topologically equivalent to $\mathcal{Z}$.

In practical terms, we therefore see that in order to organise the $n$ rows of $\mathbf{Y}$ into $m$ clusters, one
can first reduce dimension to $r=m$ by computing the PCA embedding and
then apply some clustering technique to those embedding vectors. This two-step
procedure of PCA followed by clustering, sometimes described as spectral clustering, is very popular in the practice
of high-dimensional data analysis and is exactly what \citet{yata2020geometric} recommend in the conclusion of their study of PCA embedding for mixture models in a regime where the number of samples
is fixed and the dimension tends to infinity. It is already known that PCA, albeit under slightly different variations and assumptions, allows for ``perfect clustering'' in high-dimensional mixture models \citep{loffler2021optimality,agterberg2022entrywise}.

% Unlike our theorem
% \ref{thm:uniform_consistency}, the results of \citet{yata2020geometric}
%  are not uniform across PC scores.

\begin{figure}[ht]
\includegraphics[width=1\columnwidth]{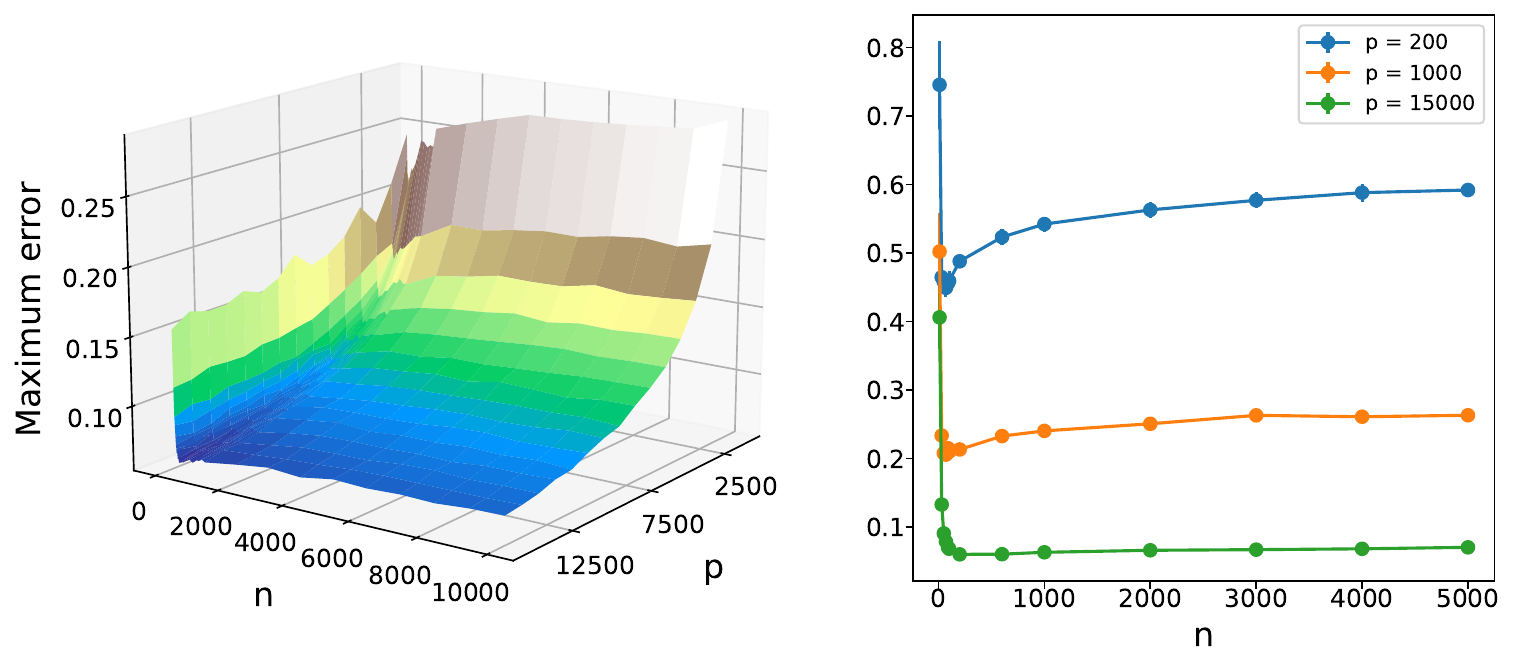}\caption{Mixture model example. Left: maximum error $\max_{i\protect\neq j}\left|p^{-1/2}\|\zeta_{i}-\zeta_{j}\|_{2}-\|\phi(Z_{i})-\phi(Z_{j})\|_{\ell_2}\right|$,
averaged over $50$ independent realisations from the model, as a
function of $n$ and $p$. Right: the same error for $p=200,1000,15000$,
as a function of $n$. \label{fig:Mixture-model-example.}}
\end{figure}

To illustrate the behaviour of the LMM and PCA embedding in this context, we consider a case in which $\mathcal{Z}=\{1,2,3\}$
and $\mu$ is the uniform distribution on $\mathcal{Z}$; for each
$j=1,\ldots,p$, $[X_{j}(1)\;X_{j}(2)\;X_{j}(3)]^{\top}\sim\mathcal{N}(\mathbf{0},\boldsymbol{\Sigma})$
where $\boldsymbol{\Sigma}$ is full-rank; and the elements of $\mathbf{E}$
are independent and identically distributed $\mathcal{N}(0,1)$ with
$\sigma=1$. Figure \ref{fig:Mixture-model-example.} shows the error 
$$\max_{i\protect\neq j}\left|p^{-1/2}\|\zeta_{i}-\zeta_{j}\|_{2}-\|\phi(Z_{i})-\phi(Z_{j})\|_{\ell_2}\right|,
$$
averaged over $50$ independent realisations from the model. The plot
on the left of the figure indicates that over the ranges considered,
for fixed $n$ the error decreases as $p$ increases. Theorem \ref{thm:consistency_summary}
is not informative about the converse situation, when $p$ is fixed
and $n$ increases: in this regime, the condition of theorem \ref{thm:uniform_consistency} involving
a lower bound on $n$ will eventually be satisfied, but the condition
involving a lower bound on $p/n$ will eventually be violated. We
examine this in the right plot of figure \ref{fig:Mixture-model-example.}.
We see that for fixed $p$, as $n$ increases the error initially
quickly decreases, but then the appears to very slowly increase $n\gg p$. We conjecture the former and is related to the $1/\sqrt{n}$ term in \eqref{eq:consistency_summary}.

\begin{figure}[ht]
\includegraphics[width=1\columnwidth]{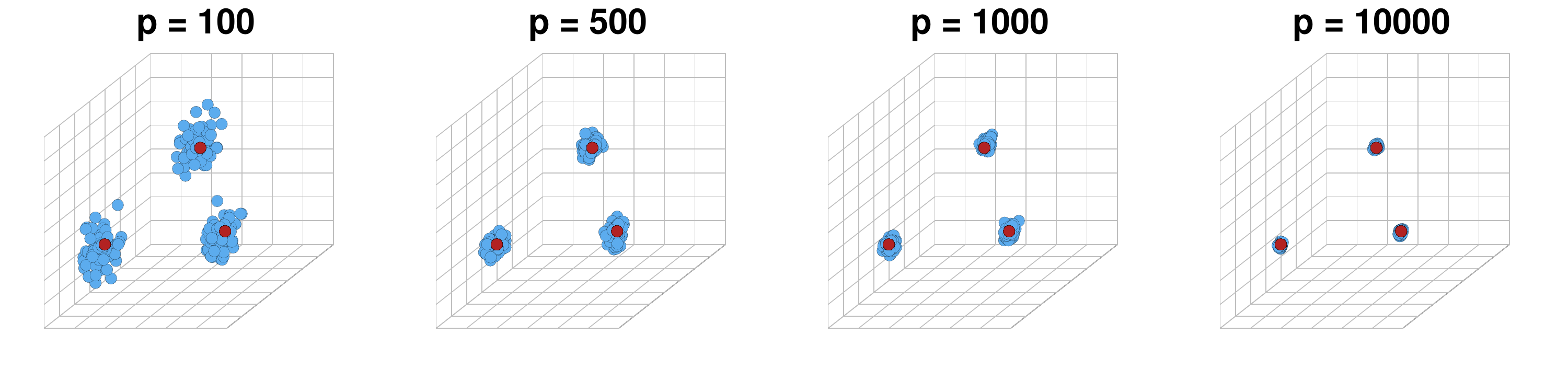}\\
\includegraphics[width=1\columnwidth]{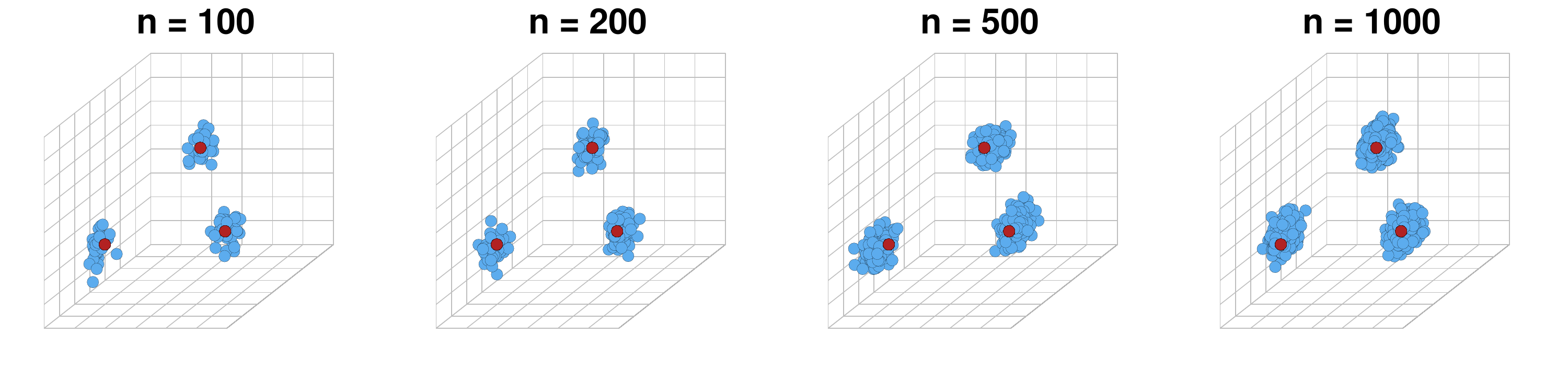}\caption{Mixture model example. PCA embedding $\{p^{-1/2}\zeta_{1},\ldots,p^{-1/2}\zeta_{n}\}$
(blue dots) and $\phi(1),\phi(2),\phi(3)$ (red dots). Top row: $n$ fixed to $200$ and $p$
varying. Bottom row $p$
fixed to $200$ and $n$ varying. \label{fig:Mixture-model-example_varying_n_and_p}}
\end{figure}

Figure \ref{fig:Mixture-model-example_varying_n_and_p} illustrates
how this error performance relates to the clustering of the PCA embedding vectors.
When $n$ is fixed, we see that as $p$ increases the embedding vectors are
increasingly tightly clustered around $\phi(1),\phi(2),\phi(3)$, in keeping with theorem \ref{thm:consistency_summary}.
When $p$ is fixed, we see that three clusters of embedding vectors are clearly
discernible, but the clusters appear not to shrink as $n$ grows.

Overall we conclude that, whilst theorem \ref{thm:consistency_summary}
shows that both $n$ and $p/n$ being large is sufficient to drive
the error to zero, our numerical results suggest that for fixed $p$ the error does not explode as
$n$ grows, and even when $n\gg p$ it may be that the PCA embedding still conveys
the topological or geometric structure of $\mathcal{M}$ and hence
$\mathcal{Z}$.

\section{Proofs and supporting material for section \ref{sec:connecting}}\label{sec:Proofs-and-supporting_for_manifold}

\begin{proof}[Proof of Proposition \ref{prop:Phi_W_expansion}]
Define 
\begin{equation}
\widetilde{\mathbf{W}}_{jk}\coloneqq\int_{\mathcal{Z}}X_{j}(z)u_{k}^{f}(z)\mu(\mathrm{d}z),\label{eq:W_tilde_defn}
\end{equation}
and note that
\begin{equation}
\widetilde{\mathbf{W}}_{jk}=p^{1/2}(\lambda_{k}^{f})^{1/2}\mathbf{W}_{jk},\label{eq:W_tilde_identity}
\end{equation}
where $\mathbf{W}_{jk}$ is defined in \eqref{eq:W_jk_defn}.

Pick any $r_{0}<  r$ and recall $r\in\{1,2\ldots,\}\cup\{\infty\}$
is the number of nonzero eigenvalues $(\lambda_{k}^{f})_{k\geq1}$.
We claim that, for any $z\in\mathcal{Z}$, the following equality
holds: 
\begin{equation}
\frac{1}{p}\sum_{j=1}^{p}\mathbb{E}\left[\left|X_{j}(z)-\sum_{k=1}^{r_{0}}u_{k}^{f}(z)\widetilde{\mathbf{W}}_{jk}\right|^{2}\right]=f(z,z)-\sum_{k=1}^{r_{0}}\lambda_{k}^{f}|u_{k}^{f}(z)|^{2}.\label{eq:L_2_decomp}
\end{equation}
To verify the equality (\ref{eq:L_2_decomp}), observe:
\begin{align*}
 & \frac{1}{p}\sum_{j=1}^{p}\mathbb{E}\left[\left|X_{j}(z)-\sum_{k=1}^{r_{0}}u_{k}^{f}(z)\widetilde{\mathbf{W}}_{jk}\right|^{2}\right]\\
 & =\frac{1}{p}\sum_{j=1}^{p}\mathbb{E}\left[\left|X_{j}(z)\right|^{2}\right]-\frac{2}{p}\sum_{j=1}^{p}\mathbb{E}\left[X_{j}(z)\sum_{k=1}^{r_{0}}u_{k}^{f}(z)\widetilde{\mathbf{W}}_{jk}\right]\\
 & \quad+\frac{1}{p}\sum_{j=1}^{p}\sum_{k=1}^{r_{0}}\sum_{\ell=1}^{r_{0}}\mathbb{E}\left[\widetilde{\mathbf{W}}_{jk}\widetilde{\mathbf{W}}_{j\ell}\right]u_{k}^{f}(z)u_{\ell}^{f}(z)\\
 & =f(z,z)-2\sum_{k=1}^{r_{0}}u_{k}^{f}(z)\int_{\mathcal{Z}}f(z,z^{\prime})u_{k}^{f}(z^{\prime})\mu(\mathrm{d}z^{\prime})\\
 & \quad+\sum_{k=1}^{r_{0}}\sum_{\ell=1}^{r_{0}}u_{k}^{f}(z)u_{\ell}^{f}(z)\int_{\mathcal{Z}}\int_{\mathcal{Z}}f(z^{\prime},z^{\prime\prime})u_{k}^{f}(z^{\prime})u_{\ell}^{f}(z^{\prime\prime})\mu(\mathrm{d}z^{\prime})\mu(\mathrm{d}z^{\prime\prime})\\
 & =f(z,z)-2\sum_{k=1}^{r_{0}}\lambda_{k}|u_{k}^{f}(z)|^{2}+\sum_{k=1}^{r_{0}}\lambda_{k}^{f}|u_{k}^{f}(z)|^{2}\\
 & =f(z,z)-\sum_{k=1}^{r_{0}}\lambda_{k}^{f}|u_{k}^{f}(z)|^{2},
\end{align*}
where the second equality uses (\ref{eq:W_tilde_defn}) and $f(z,z^{\prime})=p^{-1}\sum_{j=1}^{p}\mathbb{E}[X_{j}(z)X_{j}(z^{\prime})]$,
and the third equality uses the fact that $(u_{k}^{f},\lambda_{k}^{f})_{k\geq1}$,
by definition, are $L_{2}(\mu)$-orthonormal eigenfunctions and eigenvalues
of the integral operator associated with the kernel $f$ and the measure
$\mu$.

By Mercer's theorem (theorem \ref{thm:mercer}) the r.h.s. of (\ref{eq:L_2_decomp})
converges to zero as $r_{0}\to r$, uniformly in $z$. Each of the
summands on the l.h.s. of (\ref{eq:L_2_decomp}) is nonnegative, so
they must also converge to zero uniformly in $z$. Using this uniform
convergence and the fact that for any $j=1,\ldots,p$ and $i=1,\ldots,n$, the pair of random variables $X_{j}$ and $Z_{i}$ are statistically
independent,  we have:
\begin{align}
\lim_{r_{0}\to r}\mathbb{E}\left[\left|X_{j}(Z_{i})-\sum_{k=1}^{r_{0}}u_{k}^{f}(Z_{i})\widetilde{\mathbf{W}}_{jk}\right|^{2}\right] & =\lim_{r_{0}\to r}\int_{\mathcal{Z}}\mathbb{E}\left[\left|X_{j}(z)-\sum_{k=1}^{r_{0}}u_{k}^{f}(z)\widetilde{\mathbf{W}}_{jk}\right|^{2}\right]\mu(\mathrm{d}z)\nonumber \\
 & \leq\lim_{r_{0}\to r}\sup_{z}\mathbb{E}\left[\left|X_{j}(z)-\sum_{k=1}^{r_{0}}u_{k}^{f}(z)\widetilde{\mathbf{W}}_{jk}\right|^{2}\right]=0.\label{eq:X_j_Z_i_con}
\end{align}
Noting the identity (\ref{eq:W_tilde_identity}) and recalling $\phi(z)=[(\lambda_{1}^{f})^{1/2}u_{1}^{f}(z)\;(\lambda_{2}^{f})^{1/2}u_{2}^{f}(z)\;\cdots]^{\top}$,
we find that (\ref{eq:X_j_Z_i_con}) can equivalently be written:
\[
X_{j}(Z_{i})\stackrel{m.s.}{=}p^{1/2}\left\langle \phi(Z_{i}),\mathbf{W}_{j}\right\rangle _{\ell_2},
\]
where $\mathbf{W}_{j}$ is the $j$th row of $\mathbf{W}$. This completes
the proof of the first identity in \eqref{eq:Phi_W_expansion}.

The second identity in \eqref{eq:Phi_W_expansion} follows from the fact that $(u_{k}^{f},\lambda_{k}^{f})_{k\geq1}$
are orthonormal eigenfunctions/values:
\begin{align*}
\sum_{j=1}^{p}\mathbb{E}[\mathbf{W}_{jk}\mathbf{W}_{j\ell}] & =\frac{1}{(\lambda_{k}^{f}\lambda_{\ell}^{f})^{1/2}}\int_{\mathcal{Z}}\int_{\mathcal{Z}}u_{k}^{f}(z)\frac{1}{p}\sum_{j=1}^{p}\mathbb{E}\left[X_{j}(z)X_{j}(z^{\prime})\right]u_{\ell}^{f}(z^{\prime})\mu(\mathrm{d}z^{\prime})\mu(\mathrm{d}z)\\
 & =\frac{\lambda_{\ell}^{f}}{(\lambda_{k}^{f}\lambda_{\ell}^{f})^{1/2}}\int_{\mathcal{Z}}u_{k}^{f}(z)u_{\ell}^{f}(z)\mu(\mathrm{d}z)=\begin{cases}
1, & k=\ell\\
0, & k\neq\ell
\end{cases}.
\end{align*}
\end{proof}

We introduce the following assumption in order to prove proposition \ref{prop:mixing} below.
\begin{assump}\label{ass:mixing}
For mixing coefficients $\varphi$ satisfying $\sum_{k\geq 1}\varphi^{1/2}(k)<\infty$ and all $z,z^\prime\in\Zc$, the sequence $\{(X_j(z),X_j(z^\prime));j\geq 1\}$ is $\varphi$-mixing.
\end{assump}
\noindent Assumption \ref{ass:moments} is stated in section \ref{sec:Proof-and-supporting_consistency}.
\begin{prop}\label{prop:mixing} 
Assume \ref{ass:cont_covar}, \ref{ass:moments} and \ref{ass:mixing}, and let $q\geq 1$ and $\varphi$ be as therein. Then there exists a constant $C(\varphi)$ depending only on $\varphi$ such
that for any $\delta>0$ and any $i,j$, 
$$
\mathbb{P}\left(\left.\left|p^{-1}\left\langle \mathbf{Y}_{i},\mathbf{Y}_{j}\right\rangle -\langle \phi(Z_{i}),\phi(Z_{j})\rangle_{\ell_2}-\sigma^2\mathbf{I}[i=j]\right|\geq\delta\right|Z_i,Z_j\right)
 \leq\frac{1}{\delta^{2q}}\frac{1}{p^{q}}C(\varphi)M(q,\sigma)
$$
where 
\begin{align*}
M(q,\sigma) & \coloneqq\sup_{j\geq1}\sup_{z\in\mathcal{Z}}\mathbb{E}\left[|X_{j}(z)|^{4q}\right]\\
 & \quad+\sigma\sup_{i,j\geq 1}\mathbb{E}\left[\left|\mathbf{E}_{ij}\right|^{2q}\right] \sup_{j\geq1}\sup_{z\in\mathcal{Z}}\mathbb{E}\left[\left.|X_{j}(z)|^{2q}\right|\right]\\
 & \quad+\sigma^2\sup_{i,j\geq 1}\mathbb{E}\left[\left|\mathbf{E}_{ij}\right|^{4q}\right].
\end{align*}
\end{prop}
\begin{proof} Fix any $i,j$ and consider the decomposition:
\[
p^{-1}\left\langle \mathbf{Y}_{i},\mathbf{Y}_{j}\right\rangle -\langle \phi(Z_{i}),\phi(Z_{j})\rangle_{\ell_2}-\sigma^2\mathbf{I}[i=j]=\sum_{k=1}^{4}\Delta_{k}
\]
where 
\begin{align*}
 & \Delta_{1}\coloneqq p^{-1}\left\langle \mathbf{X}(Z_{i}),\mathbf{X}(Z_{j})\right\rangle -f(Z_{i},Z_{j})\\
 & \Delta_{2}\coloneqq p^{-1}\sigma\left\langle \mathbf{X}(Z_{i}),\mathbf{E}_{j}\right\rangle \\
 & \Delta_{3}\coloneqq p^{-1}\sigma\left\langle \mathbf{X}(Z_{j}),\mathbf{E}_{i}\right\rangle \\
 & \Delta_{4}\coloneqq p^{-1}\sigma^2\left\langle \mathbf{E}_{i},\mathbf{E}_{j}\right\rangle -\sigma^2\mathbf{I}[i=j]
\end{align*}
and $\mathbf{X}(z)\coloneqq[X_1(z)\,\ldots\, X_p(z)]^\top$.

Writing $\Delta_{1}$ as 
\[
\Delta_{1}=\frac{1}{p}\sum_{k=1}^{p}\Delta_{1,k},\qquad\Delta_{1,k}\coloneqq X_{k}(Z_{i})X_{k}(Z_{j})-\mathbb{E}\left[\left.X_{k}(Z_{i})X_{k}(Z_{j})\right|Z_{i},Z_j\right].
\]
 we see that $\Delta_{1}$ is an arithmetic mean of $p$ random variables,
each of which is conditionally mean-zero given $Z_{i},Z_{j}$.

For $\Delta_{2}$, we have 
\[
\Delta_{2}= \frac{\sigma}{p}\sum_{k=1}^{p}\Delta_{2,k},\qquad\Delta_{2,k}\coloneqq X_k(Z_i)\mathbf{E}_{jk},
\]
By definition of the LMM, 
the three collections of random variables, $(
Z_{1},\ldots,Z_{n})$,
$(X_1,\ldots,X_p)$ and $(\mathbf{E}_{1},\ldots,\mathbf{E}_{n})$
are mutually independent, and the elements of each vector $\mathbf{E}_{j}\in\mathbb{R}^{p}$
are mean-zero and independent. Therefore given $Z_i,Z_j$
and $X_1,\ldots,X_p$, $\Delta_{2}$ is an arithmetic mean of
  conditionally independent and conditionally mean-zero random
variables. The same decomposition holds for $\Delta_3$, with $i,j$ interchanged.

For $\Delta_{4}$, we have 
\[
\Delta_{4} =\frac{\sigma^2}{p}\sum_{1\leq k \leq p}\Delta_{4,k },\qquad\Delta_{4,k }\coloneqq\mathbf{E}_{ik}\mathbf{E}_{jk}-\mathbf{I}[i=j].
\]
Recalling from the definition of the LMM that the elements of $\mathbf{E}$ are mean zero, unit-variance,  uncorrelated across rows,  and independent across columns, we see that $\Delta_{4}$ is a sum of $p$ mean-zero and mutually independent random variables.

\sloppy The proof proceeds by using a moment inequality for mixing random variables \citep{xuejun2010complete}[Lemma 1.7] to bound $\mathbb{E}[|\Delta_1|^{2q}|Z_i,Z_j]$, and using the Marcinkiewicz--Zygmund inequality to bound  $\mathbb{E}[|\Delta_2|^{2q}|Z_i,Z_j,X_1,\ldots,X_p]$, $\mathbb{E}[|\Delta_3|^{2q}|Z_i,Z_j,X_1,\ldots,X_p]$ and $\mathbb{E}[|\Delta_4|^{2q}]$. These bounds are then combined and Markov's inequality applied. The details are similar to \cite{gray2023hierarchical}[Proof of proposition 2], so are omitted. 

\end{proof}

\begin{proof}[Proof of Proposition \ref{prop:homeomorphism}]
As explained above the statement of proposition \ref{prop:homeomorphism}, we only need to show that \ref{ass:injective} holds if and only if $\phi$ is one-to-one. For any $z,z^\prime\in\mathcal{Z}$ consider the identities:
\begin{align*}
\|\phi(z)-\phi(z^\prime)\|_{\ell_2}^2 &=\|\phi(z)\|_{\ell_2}^2 +\|\phi(z^\prime)\|_{\ell_2}^2 -2\langle\phi(z),\phi(z^\prime)\rangle_{\ell_2}\\
& = f(z,z)+f(z^\prime,z^\prime)-2f(z,z^\prime)\\
& = \frac{1}{p}\sum_{j=1}^p\mathbb{E}[|X_j(z)|^2]+\frac{1}{p}\sum_{j=1}^p\mathbb{E}[|X_j(z^\prime)|^2]-\frac{2}{p}\sum_{j=1}^p\mathbb{E}[X_j(z)X_j(z^\prime)]\\
&= \frac{1}{p}\sum_{j=1}^p\mathbb{E}[|X_j(z)-X_j(z^\prime)|^2].
\end{align*}
Hence $\phi(z)\neq \phi(z^\prime)$ if and only if $\sum_{j=1}^p\mathbb{E}[|X_j(z)-X_j(z^\prime)|^2]>0$, showing that \ref{ass:injective} is equivalent to $\phi$ being one-to-one, as required to complete the proof of the proposition.

To go further let us also show that $\phi$ being one-to-one is equivalent to the condition: 
\begin{equation}
\text{for any } z,z^\prime\in\mathcal{Z} \text{ such that } z\neq z^\prime, \text{ there exists } \xi \text{ such that } f(z,\xi)\neq f(z^\prime,\xi).\label{eq:inj_alt}
\end{equation}
We first show that \eqref{eq:inj_alt} implies $\phi$ is one-to-one. We
prove the contrapositive to this statement. So suppose that $\phi$ is not
one-to-one. Then there must exist $z\neq z^{\prime}\in\mathcal{Z}$
such that $\phi(z)=\phi(z^{\prime})$. This implies that for any $\xi$
in $\mathcal{Z}$, $f(z,\xi)=\left\langle \phi(z),\phi(\xi)\right\rangle _{\ell_2}=\left\langle \phi(z^{\prime}),\phi(\xi)\right\rangle _{\ell_2}=f(z^{\prime},\xi)$,
which is the converse of \eqref{eq:inj_alt}. 

In the other direction, suppose the converse of \eqref{eq:inj_alt} holds, i.e., the exists $z\neq z^\prime$ such that $f(z,\xi)=f(z^\prime,\xi)$ for all $\xi$. By considering the cases $\xi=z$ and $\xi=z^\prime$ we find $f(z,z)=f(z,z^\prime)=f(z^\prime,z^\prime)$. In turn,
$$
\|\phi(z)-\phi(z^\prime)\|^2_{\ell_2} = f(z,z)+f(z^\prime,z^\prime)-2f(z,z^\prime)=0,
$$
i.e., $\phi(z) =\phi(z^\prime)$, and hence $\phi$ is not one-to-one. 
\end{proof}

% \subsection{Definitions and terminology for paths in \texorpdfstring{$\mathbb{R}^d$}{Rd}}\label{subsec:Manifold-definitions}

% A path $\eta$ in $\Zc$ is called a piecewise-$C^{1}$ path if there exists $k<\infty$ and $0=\tau_0<\tau_1<\cdots<\tau_k=1$ such that for each $i=0\ldots,k-1$, $t\mapsto \eta_t$ is continuously differentiable on each open interval $(\tau_i,\tau_{i+1})$ and such that its derivative $\dot{\eta}_t$ converges to some limit as  $t\to \tau_i$, and converges to some limit as $t\to \tau_{i+1}$.

% \begin{defn}[Piecewise differentiably smooth subset of $\mathbb{R}^d$]\label{defn:piecewise_smooth} $\Zc$ is called a piecewise differentiably smooth subset of $\mathbb{R}^d$ if for all $z,z^\prime\in\Zc$ there exists a path in $\Zc$ with end-points $z,z^\prime$ which achieves the infimum over paths $\inf_{\eta:\eta_0=z,\eta_1=z^\prime}L(\eta)$, and any such length-minimising path is a piecewise $C^1$ path.
% \end{defn}

\subsection{Proofs and supporting material for section \ref{sec:isometry}} \label{subsec:isometry_proofs}
The purpose of this section is to state some definitions and intermediate results, building towards the proofs of propositions \ref{prop:sq_euc_kernels_alt} and \ref{prop:inner_prod_kernels_alt}. Recall that the term ``continuous path'' was used in \ref{ass:diffeo_Rd_alt}. From henceforth we just say ``path'' for short.

The following definitions are standard in metric geometry \cite{burago2001course}. For $x,x^\prime\in\mathcal{M},$ a \emph{path} in $\mathcal{M}$ with end-points
$x,x^\prime$ is a continuous function $\gamma:[0,1]\to\mathcal{M}$ such
that $\gamma_{0}=x$ and $\gamma_{1}=x^\prime$, where  $\mathcal{M}$ is equipped with the distance $\|\cdot-\cdot\|_{\ell_2}$. With $n\geq1$, a non-decreasing
sequence $t_{0},\ldots,t_{n}$ such that $t_{0}=0$ and $t_{n}=1$,
is called a \emph{partition.} Given a path $\gamma$ and a partition
$\mathcal{P}=(t_{0},\ldots,t_{n})$, define $\chi(\gamma,\mathcal{P})\coloneqq\sum_{k=1}^{n}\|\gamma_{t_{k}}-\gamma_{t_{k-1}}\|_{\ell_2}$.
The \emph{length} of $\gamma$  is $L(\gamma)\coloneqq\sup_{\mathcal{P}}\chi(\gamma,\mathcal{P})$,
where the supremum is over all possible partitions.

% A path $\gamma$ in $\mathcal{M}$ is called a piecewise $C^{1}$ path if there exists $k<\infty$ and $0=\tau_0<\tau_1<\cdots<\tau_k=1$ such that for each $i=0,\ldots,k-1$ and $t\in(\tau_i,\tau_{i+1})$,
% there exists $\dot{\gamma}_{t}\in\ell_{2}$ such that
% \[
% \lim_{h\to 0}\left\Vert \frac{\gamma_{t+h}-\gamma_{t}}{h}-\dot{\gamma}_{t}\right\Vert _{\ell_2}=0,
% \]
% and $t\mapsto\dot{\gamma}_{t}$ is continuous in $\ell_{2}$ on $(\tau_i,\tau_{i+1})$, converges to some limit in $\ell_2$ as  $t\to \tau_i$, and converges to some limit in $\ell_2$ as $t\to \tau_{i+1}$.

When $\Zc$ is a subset of $\mathbb{R}^d$, a \emph{path} $\eta$ in $\mathcal{Z}$ with end-points $z,z^\prime$
is a continuous function $\eta:[0,1]\to\mathcal{Z}$ such that $\eta_{0}=z,\eta_{1}=z^\prime$,
and with $\chi(\eta,\mathcal{P})\coloneqq\sum_{k=1}^{n}\|\eta_{t_{k}}-\eta_{t_{k-1}}\|_{\mathbb{R}^d}$
the length of $\eta$ is $L(\eta)\coloneqq\sup_{\mathcal{P}}\chi(\eta,\mathcal{P})$.

The \emph{shortest path lengths}, also known as \emph{geodesic distances}, in $\mathcal{M}$ and $\Zc$ are:
\begin{equation}\label{eq:d_M_and_d_Z_defn}
d_{\mathcal{M}}^{\mathrm{geo}}(x,x^\prime)\coloneqq \inf_{\gamma:\gamma_0=x,\gamma_1=x^\prime} L(\gamma)\qquad d_{\mathcal{Z}}^{\mathrm{geo}}(z,z^\prime)\coloneqq \inf_{\eta:\eta_0=z,\eta_1=z^\prime} L(\eta),
\end{equation}
where the infima are over all paths in respectively $\mathcal{M}$ and $\Zc$ with the indicated end-points.

% The main intermediate result is theorem \ref{thm:gamma_dot} below, relying on assumption \ref{ass:diffeo_Rd}. This assumption is verified as part of the proofs of propositions \ref{prop:sq_euc_kernels_alt} and \ref{prop:inner_prod_kernels_alt}, which appear at the end of section \ref{subsec:isometry_proofs}. 

\begin{assump}\label{ass:diffeo_Rd} Assume  \ref{ass:diffeo_Rd_alt} holds and with $d$ as therein, additionally assume there exists a closed ball $\widetilde{\mathcal{Z}}\subset\mathbb{R}^{d}$ centered on the origin such that:
$\mathcal{Z}\subset\widetilde{\mathcal{Z}}$; the definition of $f(z,z^\prime)$
can be extended from $\mathcal{Z}\times\mathcal{Z}$ to $\widetilde{\mathcal{Z}}\times\widetilde{\mathcal{Z}}$; $f$ is $C^2$ on $\widetilde{\mathcal{Z}}\times\widetilde{\mathcal{Z}}$ and the matrix $\mathbf{H}_{\xi}\in\mathbb{R}^{d\times d}$
with elements:
\[
(\mathbf{H}_{\xi})_{ij}\coloneqq\left.\frac{\partial^{2}f}{\partial z_{i}\partial z_{j}^{\prime}}\right|_{(\xi,\xi)}
\]
is positive-definite for all $\xi\in\Zc$.
\end{assump}
The statement of the following theorem, from \cite{whiteley2021matrix}, is paraphrased slightly in order to match the assumptions of interest here.
\begin{thm}[\cite{whiteley2021matrix}, Thm 1.]\label{thm:neurips21}Assume \ref{ass:cont_covar}, \ref{ass:injective}
and \ref{ass:diffeo_Rd}. Then $\phi$ is a bi-Lipschitz homeomorphism between $\mathcal{Z}$
and $\mathcal{M}$. Let $x,x^\prime$ be any two points in $\mathcal{M}$, and let $\gamma$ be any path in $\mathcal{M}$ of finite length, with end-points $x,x^\prime$. Define $\eta:[0,1]\to\mathcal{Z}$
by $\eta_{t}\coloneqq\phi^{-1}(\gamma_{t})$. Then $\eta$ is a path
in $\mathcal{Z}$ with $L(\eta)<\infty$. For any $\epsilon>0$ there
exists a partition $\mathcal{P}_{\epsilon}$ such that for any partition
$\mathcal{P}=(t_{0},\ldots,t_{n})$ satisfying $\mathcal{P}_{\epsilon}\subseteq\mathcal{P}$,
\begin{equation}
\left|L(\gamma)-\sum_{k=1}^{n}\left\langle \eta_{t_{k}}-\eta_{t_{k-1}},\+H_{\eta_{t_{k-1}}}(\eta_{t_{k}}-\eta_{t_{k-1}})\right\rangle ^{1/2}\right|\leq\epsilon.\label{eq:path_length}
\end{equation}
\end{thm}

\begin{proof}[Proof of proposition \ref{prop:sq_euc_kernels_alt}]
Under the assumptions of the proposition, by direct calculation $\mathbf{H}_{\xi}=-2 g^{\prime}(0)\mathbf{I}_{d}$
for all $\xi\in\Zc$ and \ref{ass:diffeo_Rd} holds.

Fix any $z,z^\prime\in\Zc$ and let $\eta$ be any finite length path in $\Zc$ with these end-points. By theorem \ref{thm:neurips21}, $\phi$ is Lipschitz, so $\gamma$ defined by $\gamma_t\coloneqq\phi(\eta_t)$ has finite length. Define $x\coloneqq \phi(z)$, $x^\prime\coloneqq \phi(z^\prime)$. Applying theorem \ref{thm:neurips21}, we have from \eqref{eq:path_length} that for any $\epsilon>0$ there exists a partition  $\mathcal{P}_\epsilon$ such that for any $\mathcal{P}=(t_0,\ldots,t_n)$ satisfying $\mathcal{P}_\epsilon \subseteq \mathcal{P}$, 

\begin{equation}\label{eq:L_1}
\left|L(\gamma)-\sqrt{-2 g^{\prime}(0)}\chi(\eta,\mathcal{P})\right|\leq \epsilon.
\end{equation}
Also, using the definition of path length $L(\eta)$ and the triangle inequality, there exists $\widetilde{\mathcal{P}}_{\epsilon}$ such that for  any partition $\mathcal{P}$ satisfying $\widetilde{\mathcal{P}}_\epsilon \subseteq \mathcal{P}$, we have:
\begin{equation}\label{eq:L_2}
\left|L(\eta)- \chi(\eta,\mathcal{P})\right|\leq\epsilon.
\end{equation}
Choosing $\mathcal{P}$ to be the union of $\mathcal{P}_{\epsilon}$ and $\widetilde{\mathcal{P}}_{\epsilon}$, i.e., if $\tau\in \mathcal{P}_{\epsilon}$ or $\widetilde{\mathcal{P}}_{\epsilon}$, then $\tau\in\mathcal{P}$, we find that \eqref{eq:L_1} and \eqref{eq:L_2} are satisfied simultaneously. Since $\epsilon$ was arbitrarily small,  
we find that $L(\gamma)=\sqrt{-2 g^{\prime}(0)}L(\eta)$.

By theorem \ref{thm:neurips21}, $\phi$ is a bi-Lipschitz homeomorphism, so $\tilde{\gamma}_t=\phi(\tilde{\eta}_t)$ defines a bijection between the set of finite-length paths $\tilde{\gamma}$ in $\mathcal{Z}$ with end-points $\phi(z),\phi(z^\prime)$ and the set of finite length paths $\tilde{\eta}$ in $\Zc$ with end-points $z,z^\prime$.
Therefore by taking the infimum over $\eta$ on both sides of $L(\gamma)=\sqrt{-2 g^{\prime}(0)}L(\eta)$ where $\gamma$ is defined by $\gamma_t=\phi(\eta_t)$ as above, we find that 
\begin{equation}
d_{\mathcal{M}}^{\mathrm{geo}}(\phi(z),\phi(z^\prime))=\sqrt{-2 g^{\prime}(0)}d_{\mathcal{Z}}^{\mathrm{geo}}(z,z^\prime)\label{eq:isometry_distances_equal}
\end{equation}
as required.
\end{proof}

% \begin{lem}[Inner-product kernels on the sphere]\label{lem:inner_prod_kernels}
% If $\mathcal{Z}=\{z\in\mathbb{R}^{d};\|z\|_{\mathbb{R}^{d}}=1\}$ and  $f(z,z^{\prime})=g(\left\langle z,z^{\prime}\right\rangle _{\mathbb{R}^{d}})$
% in an open neighbourhood of $\mathcal{D}$ then $\mathbf{H}_z = g^{\prime}(1) \mathbf{I}_{d}$ for all $z\in\Zc$. If $g^{\prime}(1)>0$, 
% \begin{equation}\label{eq:sphere_isometry}
% d_{\mathcal{M}}^{\mathrm{geo}}(\phi(z),\phi(z^\prime)) = \sqrt{g^{\prime}(1)} d_{\Zc}^{\mathrm{geo}}(z,z^\prime).
% \end{equation}
% \end{lem}

\begin{proof}[Proof of proposition \ref{prop:inner_prod_kernels_alt}]
For the $\Zc$ in question, we have $g(\langle z,z^\prime\rangle_{\mathbb{R}^d}) =g(1-\|z-z^\prime\|_{\mathbb{R}^d}^2/2)$. The proof is completed by applying proposition \ref{prop:sq_euc_kernels_alt} and using the chain rule of differentiation.
\end{proof}

\section{Proof and supporting results for theorem \ref{thm:consistency_summary}}\label{sec:Proof-and-supporting_consistency}

Theorem \ref{thm:consistency_summary} is a corollary to theorem \ref{thm:uniform_consistency}. The proofs of both these theorems are in section \ref{subsec:Proof-of-main-theorem}. Section \ref{subsec:defns_and_prelims} contains definitions and notation used throughout section \ref{sec:Proof-and-supporting_consistency}. Various intermediate results used in the proof of theorem \ref{thm:uniform_consistency} are given in sections \ref{subsec:matrix_est}-\ref{subsec:matrix_conc}.

The following assumption is a more detailed version of \ref{ass:moments_q_equals_1}.
\begin{assump}\label{ass:moments}
For some $q\geq1$, $\sup_{j\geq 1}\sup_{z\in\mathcal{Z}}\mathbb{E}[|X_{j}(z)|^{4q}]<\infty$
and $\sup_{j \geq 1}\sup_{i\geq 1}\mathbb{E}[|\mathbf{E}_{ij}|^{4q}]<\infty$.
\end{assump}

\begin{thm}
\label{thm:uniform_consistency}Assume \ref{ass:cont_covar}, \ref{ass:independence}, \ref{ass:finite rank} and \ref{ass:moments}, and let 
$q\geq1$ and $r<\infty$ be as therein. For $\min(p,n)\geq r$, let $\mathbf{Y}\in\mathbb{R}^{n\times p}$ follow the LMM from section \ref{sec:Model} and let $\zeta_{1},\ldots,\zeta_{n}$
be the dimension-$r$ PCA embedding.
Then there exists a random orthogonal matrix $\mathbf{Q\in}\mathbb{R}^{r\times r}$ depending on $n$ and $p$ such
that for any $\delta\in(0,1)$ and $\epsilon\in(0,1]$, if 
\[
n\geq c_{1}\sigma^{2}r^{1/2}\left(1\vee\frac{\sigma^{2}r^{1/2}}{\epsilon^{2}}\right)\vee\log\left(\frac{r}{\delta}\right)\quad\text{and}\quad\frac{p}{n}\geq c_{2}(q)\frac{r}{\delta^{1/q}\epsilon^{2}},
\]
then 
\[
\max_{i=1,\ldots,n}\left\Vert p^{-1/2}\mathbf{Q}\zeta_{i}-\phi(Z_{i})\right\Vert_{2}\leq\epsilon
\]
with probability at least $1-\delta$. Here $c_{1}$ and $c_{2}(q)$ are constants depending on the suprema in \ref{ass:moments} and the quantity $\inf_{p\geq 1} \lambda_r^f$ which is strictly positive under \ref{ass:finite rank}; and $\|\cdot\|_{2}$ is the Euclidean norm.
\end{thm}

\subsection{Definitions and preliminaries}\label{subsec:defns_and_prelims}

Throughout section \ref{sec:Proof-and-supporting_consistency} the
probability measure $\mu$ in the LMM is considered fixed, $(\lambda_{k}^{f},u_{k}^{f})_{k\geq1}$
are as in section \ref{sec:Model}, and assumption \ref{ass:finite rank} is taken to hold, so that the rank of $f$ is finite, i.e., $r<\infty$.

\subsubsection{Notation concerning vectors and matrices in general}

We notationally index the eigenvalues of a generic symmetric matrix
$\mathbf{A}$ in a non-increasing but otherwise arbitrary order $\lambda_{1}(\mathbf{A})\geq\lambda_{2}(\mathbf{A})\geq\cdots$.
For a vector $x$ with elements $x_{i}$, $\|x\|_{\infty}\coloneqq\max_{i}|x_{i}|$
and $\|x\|_{2}\coloneqq\sqrt{\sum_{i}|x_{i}|^{2}},$ and the spectral
norm and Frobenius norm of matrices are denoted $\|\cdot\|_{2}$ and
$\|\cdot\|_{F}$.

\subsubsection{Some matrices of interest}

Let the matrix $\boldsymbol{\Phi}\in\mathbb{R}^{n\times r}$ be defined
by

\[
\mathbf{\boldsymbol{\Phi}}\coloneqq[\phi(Z_{1})|\cdots|\phi(Z_{n})]^{\top},
\]

Let $\boldsymbol{\Lambda}_{\mathbf{Y}}\in\mathbb{R}^{r\times r}$
be the diagonal matrix with diagonal elements the eigenvalues $\lambda_{1}(p^{-1}\mathbf{Y}\mathbf{Y}^{\top})$ $\ldots,\lambda_{r}(p^{-1}\mathbf{Y}\mathbf{Y}^{\top})$,
and let $\mathbf{U}_{\mathbf{Y}}\in\mathbb{R}^{n\times r}$ have as
its columns orthonormal eigenvectors associated with these eigenvalues.
Since $\boldsymbol{\Phi}\in\mathbb{R}^{n\times r}$ and $r\leq\text{\ensuremath{\min}}(p,n)$,
the matrix $\boldsymbol{\Phi}\boldsymbol{\Phi}^{\top}$ has rank at
most $r$. Let $\boldsymbol{\Lambda}_{\boldsymbol{\Phi}}\in\mathbb{R}^{r\times r}$
be the diagonal matrix with diagonal elements which are the eigenvalues
$\lambda_{1}(\boldsymbol{\Phi}\boldsymbol{\Phi}^{\top}),\ldots,\lambda_{r}(\boldsymbol{\Phi}\boldsymbol{\Phi}^{\top})$,
and let $\mathbf{U}_{\boldsymbol{\Phi}}\in\mathbb{R}^{n\times r}$
have as its columns orthonormal eigenvectors associated with these
eigenvalues. Let $\mathbf{F}_{1}\boldsymbol{\Sigma}\mathbf{F}_{2}^{\top}$
denote the full singular value decomposition of $\mathbf{U}_{\boldsymbol{\Phi}}^{\top}\mathbf{U}_{\mathbf{Y}}$
and define the random orthogonal matrix $\mathbf{F}_{\star}\coloneqq\mathbf{F}_{1}\mathbf{F}_{2}^{\top}$.

\subsubsection{Some events of interest}

With $U_{j}$ denoting the $j$th column of $\mathbf{U}_{\boldsymbol{\Phi}}$,
define:

\begin{align*}
 & A_{1}(\epsilon)\coloneqq\left\{ \|p^{-1}\mathbf{Y}\mathbf{Y}^{\top}-\boldsymbol{\Phi}\boldsymbol{\Phi}^{\top}-\sigma^{2}\mathbf{I}_{n}\|_{2}\leq\epsilon n\right\} \\
 & A_{2}(\epsilon)\coloneqq\bigcap_{i=1}^{n}B_{\mathbf{Y},i}(\epsilon)\cap\bigcap_{i=1}^{r}B_{\boldsymbol{\Phi},i}(\epsilon)\\
 & A_{3}(\epsilon)\coloneqq\left\{ \max_{j=1,\ldots,r}\|(p^{-1}\mathbf{Y}\mathbf{Y}^{\top}-\boldsymbol{\Phi}\boldsymbol{\Phi}^{\top}-\sigma^{2}\mathbf{I}_{n})U_{j}\|_{\infty}\leq\epsilon n^{1/2}\right\} \\
 & A_{\mathrm{rank}}\coloneqq\left\{ \mathrm{rank}(\mathbf{Y}\mathbf{Y}^{\top})\geq r\right\} \cap\left\{ \mathrm{rank}(\boldsymbol{\Phi}\boldsymbol{\Phi}^{\top})=r\right\} \\
 & B_{\mathbf{Y},i}(\epsilon)\coloneqq\begin{cases}
\left\{ \lambda_{i}^{f}(1-\epsilon)\leq\frac{1}{n}\lambda_{i}(p^{-1}\mathbf{Y}\mathbf{Y}^{\top})\leq\lambda_{i}^{f}(1+\epsilon)\right\} , & 1\leq i\leq r,\\
\left\{ \frac{1}{n}\lambda_{i}(p^{-1}\mathbf{Y}\mathbf{Y}^{\top})\leq\epsilon\lambda_{r}^{f}\right\} , & r+1\leq i\leq n.
\end{cases}\\
 & B_{\boldsymbol{\Phi},i}(\epsilon)\coloneqq\left\{ (1-\epsilon)\lambda_{i}^{f}\leq\frac{1}{n}\lambda_{i}(\boldsymbol{\Phi}\boldsymbol{\Phi}^{\top})\leq(1+\epsilon)\lambda_{i}^{f}\right\} ,\qquad1\leq i\leq r.
\end{align*}

\subsection{Proofs of theorems \ref{thm:uniform_consistency} and  \ref{thm:consistency_summary}} \label{subsec:Proof-of-main-theorem}

\begin{proof}
[Proof of theorem \ref{thm:uniform_consistency}] Let $\mathbf{F}_{1}\boldsymbol{\Sigma}\mathbf{F}_{2}^{\top}$
be the full singular value decomposition of $\mathbf{U}_{\boldsymbol{\Phi}}^{\top}\mathbf{U}_{\mathbf{Y}}$
and define the random orthogonal matrix $\mathbf{F}_{\star}\coloneqq\mathbf{F}_{1}\mathbf{F}_{2}^{\top}$.
On the event $A_{\text{rank}}$ we have $\mathbf{U}_{\boldsymbol{\Phi}}\boldsymbol{\Lambda}_{\boldsymbol{\Phi}}\mathbf{U}_{\boldsymbol{\Phi}}^{\top}=\boldsymbol{\Phi}\boldsymbol{\Phi}^{\top}$,
and applying lemma \ref{lem:orthogonal_transformation} we find there
exists a random orthogonal matrix $\widehat{\mathbf{Q}}$ such that
$\mathbf{U}_{\boldsymbol{\Phi}}\boldsymbol{\Lambda}_{\boldsymbol{\Phi}}^{1/2}=\boldsymbol{\Phi}\widehat{\mathbf{Q}}$,
hence $[\mathbf{U}_{\boldsymbol{\Phi}}\boldsymbol{\Lambda}_{\boldsymbol{\Phi}}^{1/2}\mathbf{F}_{\star}]_{i}=\phi(Z_{i})^{\top}\mathbf{Q}$
for all $i=1,\ldots n$, where $\mathbf{Q}\coloneqq\widehat{\mathbf{Q}}\mathbf{F}_{\star}$
is orthogonal and $[\cdot]_{i}$ denotes the $i$th row of a matrix.
Lemma \ref{lem:PC_identity} shows that $[\mathbf{U}_{\mathbf{Y}}\boldsymbol{\Lambda}_{\mathbf{Y}}^{1/2}]_{i}=p^{-1/2}\zeta_{i}$.
Combining these observations we have shown that on the event $A_{\mathrm{rank}}$,
\begin{equation}
\|p^{-1/2}\mathbf{Q}\zeta_{i}-\phi(Z_{i})\|_{2}=\|[\mathbf{U}_{\mathbf{Y}}\boldsymbol{\Lambda}_{\mathbf{Y}}^{1/2}-\mathbf{U}_{\boldsymbol{\Phi}}\boldsymbol{\Lambda}_{\boldsymbol{\Phi}}^{1/2}\mathbf{F}_{\star}]_{i}\|_{2},\qquad i=1,\ldots,n.\label{eq:consistency_first_step}
\end{equation}

Now fix any $\epsilon_{1}>0$, $\epsilon_{2}\in(0,1/2)$ and $\epsilon_{3}>0$.
Note that the event $A_{\mathrm{rank}}$ is a superset of $A_{2}(\epsilon_{2})$
and thus $A_{1}(\epsilon_{1})\cap A_{2}(\epsilon_{2})\cap A_{3}(\epsilon_{3})\subseteq A_{\mathrm{rank}}$.
Throughout the remainder of the proof of theorem \ref{thm:uniform_consistency}
we shall establish various identities and inequalities involving random
variables, random matrices, etc; all such identifies and inequalities
to be understood as holding on the event $A_{1}(\epsilon_{1})\cap A_{2}(\epsilon_{2})\cap A_{3}(\epsilon_{3})$,
although we shall avoid making this explicit in our notation in order
to avoid repetition. For example, for two random matrices say $\mathbf{A}$
and $\mathbf{B}$, we write ``$\mathbf{A}=\mathbf{B}$'' as shorthand
for ``$\mathbf{A}(\omega)=\mathbf{B}(\omega)$ for all $\omega\in A_{1}(\epsilon_{1})\cap A_{2}(\epsilon_{2})\cap A_{3}(\epsilon_{3})$''
and similarly for two random variables say $X,Y$, we write ``$X\leq Y$''
as shorthand for ``$X(\omega)\leq Y(\omega)$ for all $\omega\in A_{1}(\epsilon_{1})\cap A_{2}(\epsilon_{2})\cap A_{3}(\epsilon_{3})$''.

Noting that on the event $A_{\mathrm{rank}}$, the matrices $\boldsymbol{\Lambda}_{\mathbf{Y}}^{-1/2}$
and $\boldsymbol{\Lambda}_{\boldsymbol{\Phi}}^{-1/2}$ are well-defined,
let us introduce:
\begin{align*}
\mathbf{C}_{1} & \coloneqq\mathbf{F}_{\star}\boldsymbol{\Lambda}_{\mathbf{Y}}^{1/2}-\boldsymbol{\Lambda}_{\boldsymbol{\Phi}}^{1/2}\mathbf{F}_{\star}\\
\mathbf{C}_{2} & \coloneqq(\mathbf{U}_{\boldsymbol{\Phi}}^{\top}\mathbf{U}_{\mathbf{Y}}-\mathbf{F}_{\star})\boldsymbol{\Lambda}_{\mathbf{Y}}^{1/2}\\
\mathbf{C}_{3} & \coloneqq\mathbf{U}_{\mathbf{Y}}-\mathbf{U}_{\boldsymbol{\Phi}}\mathbf{F}_{\star}=\mathbf{U}_{\mathbf{Y}}-\mathbf{U}_{\boldsymbol{\Phi}}\mathbf{U}_{\boldsymbol{\Phi}}^{\top}\mathbf{U}_{\mathbf{Y}}+\mathbf{U}_{\boldsymbol{\Phi}}(\mathbf{U}_{\boldsymbol{\Phi}}^{\top}\mathbf{U}_{\mathbf{Y}}-\mathbf{F}_{\star})\\
\mathbf{D}_{1} & \coloneqq\mathbf{U}_{\boldsymbol{\Phi}}\mathbf{C}_{1}\\
\mathbf{D}_{2} & \coloneqq\mathbf{U}_{\boldsymbol{\Phi}}\mathbf{C}_{2}\\
\mathbf{D}_{3} & \coloneqq(\mathbf{I}-\mathbf{U}_{\boldsymbol{\Phi}}\mathbf{U}_{\boldsymbol{\Phi}}^{\top})(p^{-1}\mathbf{Y}\mathbf{Y}^{\top}-\boldsymbol{\Phi}\boldsymbol{\Phi}^{\top})\mathbf{C}_{3}\boldsymbol{\Lambda}_{\mathbf{Y}}^{-1/2}\\
\mathbf{D}_{4} & \coloneqq-\mathbf{U}_{\boldsymbol{\Phi}}\mathbf{U}_{\boldsymbol{\Phi}}^{\top}(p^{-1}\mathbf{Y}\mathbf{Y}^{\top}-\boldsymbol{\Phi}\boldsymbol{\Phi}^{\top})\mathbf{U}_{\boldsymbol{\Phi}}\mathbf{F}_{\star}\boldsymbol{\Lambda}_{\mathbf{Y}}^{-1/2}\\
\mathbf{D}_{5} & \coloneqq(p^{-1}\mathbf{Y}\mathbf{Y}^{\top}-\boldsymbol{\Phi}\boldsymbol{\Phi}^{\top})\mathbf{U}_{\boldsymbol{\Phi}}(\mathbf{F}_{\star}\boldsymbol{\Lambda}_{\mathbf{Y}}^{-1/2}-\boldsymbol{\Lambda}_{\boldsymbol{\Phi}}^{-1/2}\mathbf{F}_{\star})
\end{align*}
We now claim that:
\begin{equation}
\mathbf{U}_{\mathbf{Y}}\boldsymbol{\Lambda}_{\mathbf{Y}}^{1/2}-\mathbf{U}_{\boldsymbol{\Phi}}\boldsymbol{\Lambda}_{\boldsymbol{\Phi}}^{1/2}\mathbf{F}_{\star}=(p^{-1}\mathbf{Y}\mathbf{Y}^{\top}-\boldsymbol{\Phi}\boldsymbol{\Phi}^{\top})\mathbf{U}_{\boldsymbol{\Phi}}\boldsymbol{\Lambda}_{\boldsymbol{\Phi}}^{-1/2}\mathbf{F}_{\star}+\sum_{i=1}^{5}\mathbf{D}_{i},\label{eq:U_Lam_decomp}
\end{equation}
which up to some notational differences, is the same decomposition
used by \citet[Proof of Thm 18.]{lyzinski2016community} in the analysis
of spectral methods for community detection in graphs. To verify the
decomposition (\ref{eq:U_Lam_decomp}), observe:

\begin{align}
\mathbf{U}_{\mathbf{Y}}\boldsymbol{\Lambda}_{\mathbf{Y}}^{1/2}-\mathbf{U}_{\boldsymbol{\Phi}}\boldsymbol{\Lambda}_{\boldsymbol{\Phi}}^{1/2}\mathbf{F}_{\star} & =\mathbf{U}_{\mathbf{Y}}\boldsymbol{\Lambda}_{\mathbf{Y}}^{1/2}-\mathbf{U}_{\boldsymbol{\Phi}}\mathbf{F}_{\star}\boldsymbol{\Lambda}_{\mathbf{Y}}^{1/2}\nonumber \\
 & \quad+\mathbf{U}_{\boldsymbol{\Phi}}\mathbf{C}_{1}\nonumber \\
 & =(\mathbf{I}_{n}-\mathbf{U}_{\boldsymbol{\Phi}}\mathbf{U}_{\boldsymbol{\Phi}}^{\top})\mathbf{U}_{\mathbf{Y}}\boldsymbol{\Lambda}_{\mathbf{Y}}^{1/2}\nonumber \\
 & \quad+\mathbf{U}_{\boldsymbol{\Phi}}\mathbf{C}_{2}\nonumber \\
 & \quad+\mathbf{U}_{\boldsymbol{\Phi}}\mathbf{C}_{1}\nonumber \\
 & =(\mathbf{I}_{n}-\mathbf{U}_{\boldsymbol{\Phi}}\mathbf{U}_{\boldsymbol{\Phi}}^{\top})(p^{-1}\mathbf{Y}\mathbf{Y}^{\top}-\boldsymbol{\Phi}\boldsymbol{\Phi}^{\top})\mathbf{U}_{\mathbf{Y}}\boldsymbol{\Lambda}_{\mathbf{Y}}^{-1/2}\label{eq:add_zero}\\
 & \quad+\mathbf{U}_{\boldsymbol{\Phi}}\mathbf{C}_{2}\nonumber \\
 & \quad+\mathbf{U}_{\boldsymbol{\Phi}}\mathbf{C}_{1}\nonumber \\
 & =(p^{-1}\mathbf{Y}\mathbf{Y}^{\top}-\boldsymbol{\Phi}^{\top}\boldsymbol{\Phi})\mathbf{U}_{\boldsymbol{\Phi}}\mathbf{F}_{\star}\boldsymbol{\Lambda}_{\mathbf{Y}}^{-1/2}\nonumber \\
 & \quad-\mathbf{U}_{\boldsymbol{\Phi}}\mathbf{U}_{\boldsymbol{\Phi}}^{\top}(p^{-1}\mathbf{Y}\mathbf{Y}^{\top}-\boldsymbol{\Phi}\boldsymbol{\Phi}^{\top})\mathbf{U}_{\boldsymbol{\Phi}}\mathbf{F}_{\star}\boldsymbol{\Lambda}_{\mathbf{Y}}^{-1/2}\nonumber \\
 & \quad+(\mathbf{I}_n-\mathbf{U}_{\boldsymbol{\Phi}}\mathbf{U}_{\boldsymbol{\Phi}}^{\top})(p^{-1}\mathbf{Y}\mathbf{Y}^{\top}-\boldsymbol{\Phi}\boldsymbol{\Phi}^{\top})\mathbf{C}_{3}\boldsymbol{\Lambda}_{\mathbf{Y}}^{-1/2}\nonumber \\
 & \quad+\mathbf{U}_{\boldsymbol{\Phi}}\mathbf{C}_{2}\nonumber \\
 & \quad+\mathbf{U}_{\boldsymbol{\Phi}}\mathbf{C}_{1}\nonumber \\
 & =(p^{-1}\mathbf{Y}\mathbf{Y}^{\top}-\boldsymbol{\Phi}\boldsymbol{\Phi}^{\top})\mathbf{U}_{\boldsymbol{\Phi}}\boldsymbol{\Lambda}_{\boldsymbol{\Phi}}^{-1/2}\mathbf{F}_{\star}\nonumber \\
 & \quad+(p^{-1}\mathbf{Y}\mathbf{Y}^{\top}-\boldsymbol{\Phi}\boldsymbol{\Phi}^{\top})\mathbf{U}_{\boldsymbol{\Phi}}(\mathbf{F}_{\star}\boldsymbol{\Lambda}_{\mathbf{Y}}^{-1/2}-\boldsymbol{\Lambda}_{\boldsymbol{\Phi}}^{-1/2}\mathbf{F}_{\star})\nonumber \\
 & \quad-\mathbf{U}_{\boldsymbol{\Phi}}\mathbf{U}_{\boldsymbol{\Phi}}^{\top}(p^{-1}\mathbf{Y}\mathbf{Y}^{\top}-\boldsymbol{\Phi}\boldsymbol{\Phi}^{\top})\mathbf{U}_{\boldsymbol{\Phi}}\mathbf{F}_{\star}\boldsymbol{\Lambda}_{\mathbf{Y}}^{-1/2}\nonumber \\
 & \quad+(\mathbf{I}_n-\mathbf{U}_{\boldsymbol{\Phi}}\mathbf{U}_{\boldsymbol{\Phi}}^{\top})(p^{-1}\mathbf{Y}\mathbf{Y}^{\top}-\boldsymbol{\Phi}\boldsymbol{\Phi}^{\top})\mathbf{C}_{3}\boldsymbol{\Lambda}_{\mathbf{Y}}^{-1/2}\nonumber \\
 & \quad+\mathbf{U}_{\boldsymbol{\Phi}}\mathbf{C}_{2}\nonumber \\
 & \quad+\mathbf{U}_{\boldsymbol{\Phi}}\mathbf{C}_{1}\nonumber \\
 & =(p^{-1}\mathbf{Y}\mathbf{Y}^{\top}-\boldsymbol{\Phi}\boldsymbol{\Phi}^{\top})\mathbf{U}_{\boldsymbol{\Phi}}\boldsymbol{\Lambda}_{\boldsymbol{\Phi}}^{-1/2}\mathbf{F}_{\star}+\mathbf{D}_{5}+\mathbf{D}_{4}+\mathbf{D}_{3}+\mathbf{D}_{2}+\mathbf{D}_{1}
\end{align}
where (\ref{eq:add_zero}) holds because $\mathbf{U}_{\mathbf{Y}}\boldsymbol{\Lambda}_{\mathbf{Y}}^{1/2}=p^{-1}\mathbf{Y}^{\top}\mathbf{Y}\mathbf{U}_{\mathbf{Y}}\boldsymbol{\Lambda}_{\mathbf{Y}}^{-1/2}$
and $\mathbf{U}_{\boldsymbol{\Phi}}\mathbf{U}_{\boldsymbol{\Phi}}^{\top}\boldsymbol{\Phi}\boldsymbol{\Phi}^{\top}=\mathbf{\boldsymbol{\Phi}}\boldsymbol{\Phi}^{\top}$
.

The proof proceeds by bounding the Frobenius norm of each matrix $\mathbf{D}_{i}$,
$i=1,\ldots,5$.. Using lemma \ref{lem:order_swtich},
\begin{align}
\|\mathbf{D}_{1}\|_{F} & =\|\mathbf{C}_{1}\|_{F}\nonumber \\
 & \leq\frac{r^{1/2}}{2n^{1/2}(1-\epsilon_{2})^{1/2}(\lambda_{r}^{f})^{1/2}}\left[n\frac{(\epsilon_{1}+n^{-1}\sigma^{2})^{2}}{\lambda_{r}^{f}(1-2\epsilon_{2})}\left(1+2\frac{\lambda_{1}^{f}}{\lambda_{r}^{f}}\left(\frac{1+\epsilon_{2}}{1-2\epsilon_{2}}\right)\right)+n\epsilon_{1}+\sigma^{2}\right]\nonumber \\
 & =\frac{r^{1/2}n^{1/2}(\epsilon_{1}+n^{-1}\sigma^{2})}{2(1-\epsilon_{2})^{1/2}(\lambda_{r}^{f})^{1/2}}\left[\frac{(\epsilon_{1}+n^{-1}\sigma^{2})}{\lambda_{r}^{f}(1-2\epsilon_{2})}\left(1+2\frac{\lambda_{1}^{f}}{\lambda_{r}^{f}}\left(\frac{1+\epsilon_{2}}{1-2\epsilon_{2}}\right)\right)+1\right].\label{eq:D_1_Frob}
\end{align}
Using lemma \ref{lem:davis_kahan},
\begin{align}
\|\mathbf{D}_{2}\|_{F} & \leq r^{1/2}\|\mathbf{C}_{2}\|_{2}\nonumber \\
 & =r^{1/2}n^{1/2}[\lambda_{1}^{f}(1+\epsilon_{2})]^{1/2}\left[\frac{\epsilon_{1}+n^{-1}\sigma^{2}}{\lambda_{r}^{f}(1-2\epsilon_{2})}\right]^{2}.\label{eq:D_2_Frob}
\end{align}
Again using lemma \ref{lem:davis_kahan} and the fact that $\mathbf{U}_{\mathbf{Y}}-\mathbf{U}_{\boldsymbol{\Phi}}\mathbf{U}_{\boldsymbol{\Phi}}^{\top}\mathbf{U}_{\mathbf{Y}}=(\mathbf{U}_{\mathbf{Y}}\mathbf{U}_{\mathbf{Y}}^{\top}-\mathbf{U}_{\boldsymbol{\Phi}}\mathbf{U}_{\boldsymbol{\Phi}}^{\top})\mathbf{U}_{\mathbf{Y}}$,
\begin{align}
\|\mathbf{D}_{3}\|_{F} & \leq 2r^{1/2}\|p^{-1}\mathbf{Y}\mathbf{Y}^{\top}-\boldsymbol{\Phi}\boldsymbol{\Phi}^{\top}\|_{2}\|\mathbf{C}_{3}\|_{2}\|\boldsymbol{\Lambda}_{\mathbf{Y}}^{-1/2}\|_{2}\nonumber \\
 & \leq 2r^{1/2}\frac{(\epsilon_{1}n+\sigma^{2})}{n^{1/2}\left[\lambda_{r}^{f}(1-\epsilon_{2})\right]^{1/2}}\left(\|\mathbf{U}_{\mathbf{Y}}\mathbf{U}_{\mathbf{Y}}^{\top}-\mathbf{U}_{\boldsymbol{\Phi}}\mathbf{U}_{\boldsymbol{\Phi}}^{\top}\|_{2}+\|\mathbf{U}_{\boldsymbol{\Phi}}^{\top}\mathbf{U}_{\mathbf{Y}}-\mathbf{F}_{\star}\|_{2}\right)\nonumber \\
 & \leq2r^{1/2}n^{1/2}\frac{(\epsilon_{1}+n^{-1}\sigma^{2})^{2}}{\left[\lambda_{r}^{f}(1-\epsilon_{2})\right]^{3/2}}\left(1+\frac{\epsilon_{1}+n^{-1}\sigma^{2}}{\lambda_{r}^{f}(1-\epsilon_{2})}\right)\label{eq:D_3_Frob}
\end{align}
Directly:
\begin{align}
\|\mathbf{D}_{4}\|_{F} & \leq r^{1/2}\|\mathbf{D}_{4}\|_{2}\nonumber \\
 & \leq r^{1/2}\|p^{-1}\mathbf{Y}\mathbf{Y}^{\top}-\boldsymbol{\Phi}\boldsymbol{\Phi}^{\top}\|_{2}\|\boldsymbol{\Lambda}_{\mathbf{Y}}^{-1/2}\|_{2}\nonumber \\
 & \leq r^{1/2}\frac{(\epsilon_{1}n+\sigma^{2})}{n^{1/2}\left[\lambda_{r}^{f}(1-\epsilon_{2})\right]^{1/2}}\nonumber \\
 & =r^{1/2}n^{1/2}\frac{(\epsilon_{1}+n^{-1}\sigma^{2})}{\left[\lambda_{r}^{f}(1-\epsilon_{2})\right]^{1/2}}\label{eq:D_4_Frob}
\end{align}
Using lemma \ref{lem:order_swtich},

\begin{align}
\|\mathbf{D}_{5}\|_{F} & =\|(p^{-1}\mathbf{Y}\mathbf{Y}^{\top}-\boldsymbol{\Phi}\boldsymbol{\Phi}^{\top})\mathbf{U}_{\boldsymbol{\Phi}}(\mathbf{F}_{\star}\boldsymbol{\Lambda}_{\mathbf{Y}}^{-1/2}-\boldsymbol{\Lambda}_{\boldsymbol{\Phi}}^{-1/2}\mathbf{F}_{\star})\|_{F}\nonumber \\
 & \leq r^{1/2}\|p^{-1}\mathbf{Y}\mathbf{Y}^{\top}-\boldsymbol{\Phi}\boldsymbol{\Phi}^{\top}\|_{2}\|\mathbf{F}_{\star}\boldsymbol{\Lambda}_{\mathbf{Y}}^{-1/2}-\boldsymbol{\Lambda}_{\boldsymbol{\Phi}}^{-1/2}\mathbf{F}_{\star}\|_{F}\nonumber \\
 & \leq r^{1/2}(\epsilon_{1}n+\sigma^{2})\frac{\|\mathbf{F}_{\star}\boldsymbol{\Lambda}_{\mathbf{Y}}-\boldsymbol{\Lambda}_{\boldsymbol{\Phi}}\mathbf{F}_{\star}\|_{F}}{2n^{3/2}(\lambda_{r}^{f})^{3/2}(1-\epsilon_{2})^{3/2}}\nonumber \\
 & \leq\frac{rn^{2}(\epsilon_{1}+n^{-1}\sigma^{2})}{2n^{3/2}(\lambda_{r}^{f})^{3/2}(1-\epsilon_{2})^{3/2}}\left[\frac{(\epsilon_{1}+n^{-1}\sigma^{2})^{2}}{\lambda_{r}^{f}(1-2\epsilon_{2})}\left(1+2\frac{\lambda_{1}^{f}}{\lambda_{r}^{f}}\left(\frac{1+\epsilon_{2}}{1-2\epsilon_{2}}\right)\right)+\epsilon_{1}+\frac{\sigma^{2}}{n}\right]\nonumber \\
 & =\frac{rn^{1/2}(\epsilon_{1}+n^{-1}\sigma^{2})^{2}}{2(\lambda_{r}^{f})^{3/2}(1-\epsilon_{2})^{3/2}}\left[\frac{(\epsilon_{1}+n^{-1}\sigma^{2})}{\lambda_{r}^{f}(1-2\epsilon_{2})}\left(1+2\frac{\lambda_{1}^{f}}{\lambda_{r}^{f}}\left(\frac{1+\epsilon_{2}}{1-2\epsilon_{2}}\right)\right)+1\right]\label{eq:D_5_Frob}
\end{align}
Having obtained the above bounds on $\|\mathbf{D}_{i}\|_{F}$, for
$i=1,\ldots,5$, we turn to the first term on the r.h.s. of (\ref{eq:U_Lam_decomp}).
Writing $[\cdot]_{i}$ to indicate the $i$th row of a matrix,
\begin{align}
&\max_{i=1,\ldots,n}\|[(p^{-1}\mathbf{Y}\mathbf{Y}^{\top}-\boldsymbol{\Phi}\boldsymbol{\Phi}^{\top})\mathbf{U}_{\boldsymbol{\Phi}}\boldsymbol{\Lambda}_{\boldsymbol{\Phi}}^{-1/2}\mathbf{F}_{\star}]_{i}\|_{2}\\ & =\max_{i=1,\ldots,n}\|[(p^{-1}\mathbf{Y}\mathbf{Y}^{\top}-\boldsymbol{\Phi}\boldsymbol{\Phi}^{\top})\mathbf{U}_{\boldsymbol{\Phi}}\boldsymbol{\Lambda}_{\boldsymbol{\Phi}}^{-1/2}]_{i}\|_{2}\nonumber \\
 & \leq\frac{1}{n^{1/2}(\lambda_{r}^{f})^{1/2}(1-\epsilon_{2})^{1/2}}\max_{i=1,\ldots,n}\|[(p^{-1}\mathbf{Y}\mathbf{Y}^{\top}-\boldsymbol{\Phi}\boldsymbol{\Phi}^{\top})\mathbf{U}_{\boldsymbol{\Phi}}]_{i}\|_{2}\nonumber \\
 & \leq\frac{r^{1/2}}{n^{1/2}(\lambda_{r}^{f})^{1/2}(1-\epsilon_{2})^{1/2}}\max_{j=1,\ldots,r}\|(p^{-1}\mathbf{Y}\mathbf{Y}^{\top}-\boldsymbol{\Phi}\boldsymbol{\Phi}^{\top})U_{j}\|_{\infty}\nonumber \\
 & \leq\frac{r^{1/2}\epsilon_{3}}{(\lambda_{r}^{f})^{1/2}(1-\epsilon_{2})^{1/2}}.\label{eq:Y-Phi_2_to_inf}
\end{align}
 where $U_{j}$ is the $j$th column of $\mathbf{U}_{\boldsymbol{\Phi}}$.

Recall that at the start of the proof we fixed arbitrary values $\epsilon_{1}>0$,
$\epsilon_{2}\in(0,1/2)$ and $\epsilon_{3}>0$. We now need to work
with a specific numerical value for $\epsilon_{2}$, so let us take
it to be $1/4$. Elementary manipulations of the bounds (\ref{eq:D_1_Frob})-(\ref{eq:D_5_Frob})
then show that there exists $\tilde{c}_{0}$ depending only on the constants $c_\lambda^{\mathrm{max}},c_\lambda^{\mathrm{min}}$ in lemma \ref{lem:c_lambda_constants}
such that
\begin{align*}
\|\mathbf{D}_{1}\|_{F} & \leq\tilde{c}_{0}r^{1/2}n^{1/2}\left(\epsilon_{1}+\frac{\sigma^{2}}{n}\right)\left(\epsilon_{1}+\frac{\sigma^{2}}{n}+1\right)\\
\|\mathbf{D}_{2}\|_{F} & \leq\tilde{c}_{0}r^{1/2}n^{1/2}\left(\epsilon_{1}+\frac{\sigma^{2}}{n}\right)^{2}\\
\|\mathbf{D}_{3}\|_{F} & \leq\tilde{c}_{0}r^{1/2}n^{1/2}\left(\epsilon_{1}+\frac{\sigma^{2}}{n}\right)^{2}\left(\epsilon_{1}+\frac{\sigma^{2}}{n}+1\right)\\
\|\mathbf{D}_{4}\|_{F} & \leq\tilde{c}_{0}r^{1/2}n^{1/2}\left(\epsilon_{1}+\frac{\sigma^{2}}{n}\right)\\
\|\mathbf{D}_{5}\|_{F} & \leq\tilde{c}_{0}rn^{1/2}\left(\epsilon_{1}+\frac{\sigma^{2}}{n}\right)^{2}\left(\epsilon_{1}+\frac{\sigma^{2}}{n}+1\right).
\end{align*}
 Now assuming 
\begin{equation}
n\geq2\sigma^{2}r^{1/2}\label{eq:n_lb_proof}
\end{equation}
 i.e, $n^{-1}\sigma^{2}r^{1/2}\leq1/2$, and assuming
\begin{equation}
\epsilon_{1}r^{1/2}\leq1/2\label{eq:epsilon_1_ub}
\end{equation}
 we have
\[
\left(\epsilon_{1}+\frac{\sigma^{2}}{n}\right)r^{1/2}\leq1.
\]
Applying this inequality in the above bound on $\|\mathbf{D}_{5}\|_{F}$
and allowing $\tilde{c}_{0}$ to increase where necessary we obtain:
\[
\max_{i=1,\ldots,5}\|\mathbf{D}_{i}\|_{F}\leq\tilde{c}_{0}r^{1/2}n^{1/2}\left(\epsilon_{1}+\frac{\sigma^{2}}{n}\right)
\]
Combining this estimate with (\ref{eq:Y-Phi_2_to_inf}) and again
allowing $\tilde{c}_{0}$ to increase as needed, 
\begin{align}
\max_{i=1,\ldots,n}\|[\mathbf{U}_{\mathbf{Y}}\boldsymbol{\Lambda}_{\mathbf{Y}}^{1/2}-\mathbf{U}_{\boldsymbol{\Phi}}\boldsymbol{\Lambda}_{\boldsymbol{\Phi}}^{1/2}\mathbf{F}_{\star}]_{i}\|_{2} & \leq\max_{i=1,\ldots,n}\|[(p^{-1}\mathbf{Y}\mathbf{Y}^{\top}-\boldsymbol{\Phi}\boldsymbol{\Phi}^{\top})\mathbf{U}_{\boldsymbol{\Phi}}\boldsymbol{\Lambda}_{\boldsymbol{\Phi}}^{-1/2}\mathbf{F}_{\star}]_{i}\|_{2}+\sum_{i=1}^{5}\|\mathbf{D}_{i}\|_{F}\nonumber \\
 & \leq r^{1/2}\tilde{c}_{0}n^{1/2}\left(\epsilon_{1}+\frac{\sigma^{2}}{n}\right)+r^{1/2}\tilde{c}_{0}\epsilon_{3}.\label{eq:UY-Uphi_est}
\end{align}
Now fix any $\epsilon\in(0,1]$ and let us strengthen (\ref{eq:n_lb_proof})
to 
\begin{equation}
n\geq\left(2\sigma^{2}r^{1/2}\right)\vee\left(\frac{9}{\epsilon^{2}}\tilde{c}_{0}^{2}r\sigma^{4}\right)\label{eq:n_lb_proof_2}
\end{equation}
so that $r^{1/2}\tilde{c}_{0}n^{-1/2}\sigma^{2}\leq\epsilon/3$. Then
setting $\epsilon_{1}\coloneqq\epsilon/(3n^{1/2}r^{1/2}\tilde{c}_{0})$
(which satisfies (\ref{eq:epsilon_1_ub}) since $\tilde{c}_{0}\geq1$),
$\epsilon_{3}\coloneqq\epsilon/(3r^{1/2}\tilde{c}_{0})$ and recalling
that we have already chosen $\epsilon_{2}\coloneqq1/4$ we have as
a consequence of (\ref{eq:UY-Uphi_est}),
\begin{multline*}    
\mathbb{P}\left(\max_{i=1,\ldots,n}\|[\mathbf{U}_{\mathbf{Y}}\boldsymbol{\Lambda}_{\mathbf{Y}}^{1/2}-\mathbf{U}_{\boldsymbol{\Phi}}\boldsymbol{\Lambda}_{\boldsymbol{\Phi}}^{1/2}\mathbf{F}_{\star}]_{i}\|_{2}\leq\epsilon\right)\\
\geq1-\mathbb{P}(A_{1}(\epsilon/[3n^{1/2}r^{1/2}\tilde{c}_{0}])^{c})-\mathbb{P}(A_{2}(1/4)^{c})-\mathbb{P}(A_{3}(\epsilon/[3r^{1/2}\tilde{c}_{0}])^{c}).
\end{multline*}
Now fix any $\delta\in(0,1)$. By lemma \ref{lem:Y-phi-I_conc}, proposition
\ref{prop:eigenvalue_perturbation} and lemma \ref{lem:two_to_infty},
there exists constants $\tilde{c}_{1}(q)$, $\tilde{c}_{2}$ and $\tilde{c}_{3}(q)$
(depending only on the constants $c_{\lambda}^{\mathrm{max}}$, $c_{\lambda}^{\mathrm{min}}$ from lemma \ref{lem:c_lambda_constants} and the constants $c_X(2q)$, $c_E(2q)$ from lemma \ref{lem:Y-phi-I_conc}) such that
\begin{align*}
 & \frac{p}{n}\geq\tilde{c}_{1}(q)^{1/q}\frac{r}{\delta^{1/q}\epsilon^{2}}\quad\Rightarrow\quad\mathbb{P}(A_{1}(\epsilon/[3n^{1/2}r^{1/2}\tilde{c}_{0}])^{c})\leq\frac{\delta}{3}.\\
 & n\geq\tilde{c}_{2}\left[\sigma^{2}\vee\log\left(\frac{r}{\delta}\right)\right]\text{ and }p\geq\frac{\tilde{c}_{2}}{\delta^{1/q}}\quad\Rightarrow\quad\mathbb{P}(A_{2}(1/4)^{c})\leq\frac{\delta}{3}.\\
 & \frac{p}{n^{1/q}}\geq\tilde{c}_{3}(q)^{1/q}\frac{r^{1+1/q}}{\delta^{1/q}\epsilon^{2}}\quad\Rightarrow\quad\mathbb{P}(A_{3}(\epsilon/[3r^{1/2}\tilde{c}_{0}])^{c})\leq\frac{\delta}{3}.
\end{align*}
Combining these conditions with (\ref{eq:n_lb_proof_2}) and appropriately
defining $c_{1}$ and $c_{2}$ gives the conditions in the statement
of the theorem. Recalling (\ref{eq:consistency_first_step}), the
proof is complete.
\end{proof}

\begin{proof}[Proof of theorem \ref{thm:consistency_summary}]
If \ref{ass:moments_q_equals_1} holds, then \ref{ass:moments} holds with $q=1$. We may then apply theorem \ref{thm:uniform_consistency} in the case $q=1$, and in order for the lower bound conditions on $n$ and $p/n$ in the statement of theorem \ref{thm:uniform_consistency} to be satisfied for some given $\delta$ and $\epsilon$, it is sufficient that:
\begin{equation}\label{eq:n_and_p/n_conds_proof}
n\geq \frac{-\check{c}_1\log \delta}{\epsilon^2} \qquad\text{and}\qquad \frac{p}{n}\geq \frac{\check{c}_2}{\epsilon^2\delta},
\end{equation}
for suitable constants $\check{c}_1>0$ and $\check{c}_2>0$ depending on $\sigma,c_1,c_2(q)$ and $\sup_{p\geq 1} r$, noting the latter supremum is finite under \ref{ass:finite rank}.

To complete the proof we need to show that for any $\delta \in(0,1)$ there exists $\epsilon_0>0$ and $M>0$ such that if $(1/\sqrt{n}+\sqrt{n/p})^{-1}>M$, then: 
\begin{equation}\label{eq:big_o_P_explicit}
\mathbb{P}\left[\max_{i=1,\ldots,n}\|p^{-1/2}\mathbf{Q}\zeta_i - \phi(Z_i)\|_2 > \epsilon_0\left(\frac{1}{\sqrt{n}}+\sqrt{\frac{n}{p}}\right)\right] < \delta.
\end{equation}
So to proceed, fix any $\delta\in(0,1)$, define $\epsilon_0\coloneqq \sqrt{-\check{c}_1\log\delta} \vee \sqrt{\check{c}_2/\delta}$, $M\coloneqq  \epsilon_0$ and $\epsilon\coloneqq \epsilon_0(1/\sqrt{n}+\sqrt{n/p})$. 

Assume that $ (1/\sqrt{n}+\sqrt{n/p} )^{-1} \geq M$ and notice that in this situation $\epsilon \in(0,1]$, which is a requirement of theorem \ref{thm:uniform_consistency}. It follows from the definition of $\epsilon_0$ that:
\begin{align*}
\epsilon_0^2 & \geq -\check{c}_1\log\delta \geq \frac{-\check{c}_1\log\delta}{\left(1+n/\sqrt{p}\right)^2}=\frac{-\check{c}_1\log\delta}{n\left(1/\sqrt{n}+\sqrt{n/p}\right)^2},
\end{align*}
and rearranging then using the above definition of $\epsilon$ gives:
$$
n\geq \frac{-\check{c}_1\log\delta}{\epsilon^2},
$$
i.e., the first inequality in \eqref{eq:n_and_p/n_conds_proof} holds. Similarly
$$
\epsilon_0^2 \geq  \frac{\check{c}_2}{\delta}
\geq \frac{\check{c}_2}{\left(\sqrt{p}/n+1\right)^2\delta} =\frac{\check{c}_2}{\frac{p}{n}\left(1/\sqrt{n}+\sqrt{n/p}\right)^2\delta}  
$$
hence
$$
\frac{p}{n}\geq \frac{\check{c}_2}{\epsilon^2\delta},
$$
i.e., the second inequality in \eqref{eq:n_and_p/n_conds_proof} holds. Thus by theorem \ref{thm:uniform_consistency},
$$
\mathbb{P}\left[\max_{i=1,\ldots,n}\|p^{-1/2}\mathbf{Q}\zeta_i-\phi(Z_i)\|_2 > \left(\frac{1}{\sqrt{n}}+\sqrt{\frac{n}{p}}\right)\epsilon_0\right] < \delta.
$$
which is \eqref{eq:big_o_P_explicit}.

\end{proof}

\subsection{Matrix estimates}\label{subsec:matrix_est}
\begin{lem}
\label{lem:davis_kahan} Assume \ref{ass:cont_covar} and \ref{ass:finite rank}.
Then for any $\epsilon_{1}>0$ and $\epsilon_{2}\in(0,1/2)$, on the
event
\[
A_{1}(\epsilon_{1})\cap A_{2}(\epsilon_{2})
\]
we have
\[
\|\mathbf{U}_{\mathbf{Y}}\mathbf{U}_{\mathbf{Y}}^{\top}-\mathbf{U}_{\boldsymbol{\Phi}}^{\top}\mathbf{U}_{\boldsymbol{\Phi}}\|_{2}\leq\frac{\epsilon_{1}+n^{-1}\sigma^{2}}{\lambda_{r}^{f}(1-2\epsilon_{2})}
\]
and 
\[
\|\mathbf{U}_{\boldsymbol{\Phi}}^{\top}\mathbf{U}_{\mathbf{Y}}-\mathbf{F}_{\star}\|_{2}\leq\left[\frac{\epsilon_{1}+n^{-1}\sigma^{2}}{\lambda_{r}^{f}(1-2\epsilon_{2})}\right]^{2}.
\]
\end{lem}

\begin{proof}
In outline, the proof follows \citet[Proof of Prop. 16]{lyzinski2016community},
although we work with the spectral rather than Frobenius norm. On
the event in the statement we have: 
\[
|\lambda_{r}(\boldsymbol{\Phi}\boldsymbol{\Phi}^{\top})-\lambda_{r+1}(p^{-1}\mathbf{Y}\mathbf{Y}^{\top})|\geq n\lambda_{r}^{f}(1-2\epsilon_{2})>0
\]
 and with $\sigma_{i}$ denoting the $i$th singular value of $U_{\boldsymbol{\Phi}}^{\top}U_{\mathbf{Y}}$
and $\sigma_{i}=\cos(\theta_{i})$, the Davis-Kahan $\sin(\theta)$
theorem gives:
\begin{align}
\|\mathbf{U}_{\mathbf{Y}}\mathbf{U}_{\mathbf{Y}}^{\top}-\mathbf{U}_{\boldsymbol{\Phi}}^{\top}\mathbf{U}_{\boldsymbol{\Phi}}\|_{2}=\max_{i}|\sin(\theta_{i})| & \leq\frac{\|p^{-1}\mathbf{Y}\mathbf{Y}^{\top}-\boldsymbol{\Phi}\boldsymbol{\Phi}^{\top}\|_{2}}{|\lambda_{r}(\boldsymbol{\Phi}\boldsymbol{\Phi}^{\top})-\lambda_{r+1}(p^{-1}\mathbf{Y}\mathbf{Y}^{\top})|}\nonumber \\
 & \leq\frac{\epsilon_{1}+n^{-1}\sigma^{2}}{\lambda_{r}^{f}(1-2\epsilon_{2})}.\label{eq:DAvis_Kahan_app}
\end{align}
Therefore
\begin{align*}
\|\mathbf{U}_{\boldsymbol{\Phi}}^{\top}\mathbf{U}_{\mathbf{Y}}-\mathbf{F}_{\star}\|_{2} & =\|\mathbf{F}_{1}\boldsymbol{\Sigma}\mathbf{F}_{2}^{\top}-\mathbf{F}_{1}\mathbf{F}_{2}^{\top}\|_{2}\\
 & =\|\mathbf{F}_{1}(\boldsymbol{\Sigma}-\mathbf{I}_{r})\mathbf{F}_{2}^{\top}\|_{2}\\
 & =\|\boldsymbol{\Sigma}-\mathbf{I}_{r}\|_{2}\\
 & =\max_{i=1,\ldots,r}|1-\sigma_{i}|\\
 & \leq\max_{i=1,\ldots,r}|1-\sigma_{i}^{2}|=\max_{i=1,\ldots,r}|\sin(\theta_{i})|^{2}\\
 & \le\left[\frac{\epsilon_{1}+n^{-1}\sigma^{2}}{\lambda_{r}^{f}(1-2\epsilon_{2})}\right]^{2}
\end{align*}
where for the first inequality uses $\|\mathbf{U}_{\boldsymbol{\Phi}}^{\top}\mathbf{U}_{\mathbf{Y}}\|_{2}\leq1$
and the second inequality is from (\ref{eq:DAvis_Kahan_app}).
\end{proof}
\begin{lem}
\label{lem:order_swtich}Assume \ref{ass:cont_covar} and  \ref{ass:finite rank}.
For any $\epsilon_{1}>0$, $\epsilon_{2}\in(0,1/2)$, on the event
\[
A_{1}(\epsilon_{1})\cap A_{2}(\epsilon_{2})
\]
we have
\begin{align*}
\|\mathbf{F}_{\star}\boldsymbol{\Lambda}_{\mathbf{Y}}-\boldsymbol{\Lambda}_{\boldsymbol{\Phi}}\mathbf{F}_{\star}\|_{F} & \leq r^{1/2}\left[n\frac{(\epsilon_{1}+n^{-1}\sigma^{2})^{2}}{\lambda_{r}^{f}(1-2\epsilon_{2})}\left(1+2\frac{\lambda_{1}^{f}}{\lambda_{r}^{f}}\left(\frac{1+\epsilon_{2}}{1-2\epsilon_{2}}\right)\right)+n\epsilon_{1}+\sigma^{2}\right],\\
\|\mathbf{F}_{\star}\boldsymbol{\Lambda}_{\mathbf{Y}}^{1/2}-\boldsymbol{\Lambda}_{\boldsymbol{\Phi}}^{1/2}\mathbf{F}_{\star}\|_{F} & \leq\frac{\|\mathbf{F}_{\star}\boldsymbol{\Lambda}_{\mathbf{Y}}-\boldsymbol{\Lambda}_{\boldsymbol{\Phi}}\mathbf{F}_{\star}\|_{F}}{2n^{1/2}(1-\epsilon_{2})^{1/2}(\lambda_{r}^{f})^{1/2}},\\
\|\mathbf{F}_{\star}\boldsymbol{\Lambda}_{\mathbf{Y}}^{-1/2}-\boldsymbol{\Lambda}_{\boldsymbol{\Phi}}^{-1/2}\mathbf{F}_{\star}\|_{F} & \leq\frac{\|\mathbf{F}_{\star}\boldsymbol{\Lambda}_{\mathbf{Y}}^{1/2}-\boldsymbol{\Lambda}_{\boldsymbol{\Phi}}^{1/2}\mathbf{F}_{\star}\|_{F}}{n(1-\epsilon_{2})\lambda_{r}^{f}}.
\end{align*}
\end{lem}

\begin{proof}
Using a decomposition idea from \citep[proof of lemma 17]{lyzinski2016community},
with 
\[
\mathbf{R}\coloneqq\mathbf{U}_{\mathbf{Y}}-\mathbf{U}_{\boldsymbol{\Phi}}\mathbf{U}_{\boldsymbol{\Phi}}^{\top}\mathbf{U}_{\mathbf{Y}},
\]
we have
\begin{align*}
\mathbf{F}_{\star}\boldsymbol{\Lambda}_{\mathbf{Y}}-\boldsymbol{\Lambda}_{\boldsymbol{\Phi}}\mathbf{F}_{\star} & =(\mathbf{F}_{\star}-\mathbf{U}_{\boldsymbol{\Phi}}^{\top}\mathbf{U}_{\mathbf{Y}})\boldsymbol{\Lambda}_{\mathbf{Y}}+\mathbf{U}_{\boldsymbol{\Phi}}^{\top}(p^{-1}\mathbf{Y}\mathbf{Y}^{\top}-\boldsymbol{\Phi}\boldsymbol{\Phi}^{\top})\mathbf{R}\\
 & +\mathbf{U}_{\boldsymbol{\Phi}}^{\top}(p^{-1}\mathbf{Y}\mathbf{Y}^{\top}-\boldsymbol{\Phi}\boldsymbol{\Phi}^{\top})\mathbf{U}_{\boldsymbol{\Phi}}\mathbf{U}_{\boldsymbol{\Phi}}^{\top}\mathbf{U}_{\mathbf{Y}}\\
 & +\boldsymbol{\Lambda}_{\boldsymbol{\Phi}}(\mathbf{U}_{\boldsymbol{\Phi}}^{\top}\mathbf{U}_{\mathbf{Y}}-\mathbf{F}_{\star})
\end{align*}
hence
\begin{align}
\|\mathbf{F}_{\star}\boldsymbol{\Lambda}_{\mathbf{Y}}-\boldsymbol{\Lambda}_{\boldsymbol{\Phi}}\mathbf{F}_{\star}\|_{2} & \leq\|\mathbf{U}_{\boldsymbol{\Phi}}^{\top}\mathbf{U}_{\mathbf{Y}}-\mathbf{F}_{\star}\|_{2}(\|\boldsymbol{\Lambda}_{\mathbf{Y}}\|_{2}+\|\boldsymbol{\Lambda}_{\Phi}\|_{2})\label{eq:F_lam-lam_F_1}\\
 & +\|\mathbf{U}_{\boldsymbol{\Phi}}^{\top}(p^{-1}\mathbf{Y}\mathbf{Y}^{\top}-\boldsymbol{\Phi}\boldsymbol{\Phi}^{\top})\mathbf{R}\|_{2}\label{eq:F_lam-lam_F_2}\\
 & +\|\mathbf{U}_{\boldsymbol{\Phi}}^{\top}(p^{-1}\mathbf{Y}\mathbf{Y}^{\top}-\boldsymbol{\Phi}\boldsymbol{\Phi}^{\top})\mathbf{U}_{\boldsymbol{\Phi}}\mathbf{U}_{\boldsymbol{\Phi}}^{\top}\mathbf{U}_{\mathbf{Y}}\|_{2}\label{eq:F_lam-lam_F_3}
\end{align}
For the term on the r.h.s. of (\ref{eq:F_lam-lam_F_1}), on the event
in the statement of the present lemma and using lemma \ref{lem:davis_kahan}
we have:
\[
\|\mathbf{U}_{\boldsymbol{\Phi}}^{\top}\mathbf{U}_{\mathbf{Y}}-\mathbf{F}_{\star}\|_{2}(\|\boldsymbol{\Lambda}_{\mathbf{Y}}\|_{2}+\|\boldsymbol{\Lambda}_{\Phi}\|_{2})\leq\left[\frac{\epsilon_{1}+n^{-1}\sigma^{2}}{\lambda_{r}^{f}(1-2\epsilon_{2})}\right]^{2}2n\lambda_{1}^{f}(1+\epsilon_{2}).
\]
For the term in (\ref{eq:F_lam-lam_F_2}), using $\mathbf{R}=(\mathbf{U}_{\mathbf{Y}}\mathbf{U}_{\mathbf{Y}}^{\top}-\mathbf{U}_{\boldsymbol{\Phi}}^{\top}\mathbf{U}_{\boldsymbol{\Phi}})\mathbf{U}_{\mathbf{Y}},$
we have again on the event in the statement of the present lemma and
using lemma \ref{lem:davis_kahan},
\begin{align*}
\|\mathbf{U}_{\boldsymbol{\Phi}}^{\top}(p^{-1}\mathbf{Y}\mathbf{Y}^{\top}-\boldsymbol{\Phi}\boldsymbol{\Phi}^{\top})\mathbf{R}\|_{2} & \leq\|p^{-1}\mathbf{Y}\mathbf{Y}^{\top}-\boldsymbol{\Phi}\boldsymbol{\Phi}^{\top}\|_{2}\|\mathbf{R}\|_{2}\\
 & \leq(\|p^{-1}\mathbf{Y}\mathbf{Y}^{\top}-\boldsymbol{\Phi}\boldsymbol{\Phi}^{\top}-\sigma^{2}\mathbf{I}_{n}\|_{2}+\sigma^{2})\|\mathbf{U}_{\mathbf{Y}}\mathbf{U}_{\mathbf{Y}}^{\top}-\mathbf{U}_{\boldsymbol{\Phi}}^{\top}\mathbf{U}_{\boldsymbol{\Phi}}\|_{2}\\
 & \leq(\epsilon_{1}n+\sigma^{2})\left(\frac{\epsilon_{1}+n^{-1}\sigma^{2}}{\lambda_{r}^{f}(1-2\epsilon_{2})}\right)=n\frac{(\epsilon_{1}+n^{-1}\sigma^{2})^{2}}{\lambda_{r}^{f}(1-2\epsilon_{2})}.
\end{align*}
For the term in (\ref{eq:F_lam-lam_F_3}),
\begin{align*}
\|\mathbf{U}_{\boldsymbol{\Phi}}^{\top}(p^{-1}\mathbf{Y}\mathbf{Y}^{\top}-\boldsymbol{\Phi}\boldsymbol{\Phi}^{\top})\mathbf{U}_{\boldsymbol{\Phi}}\mathbf{U}_{\boldsymbol{\Phi}}^{\top}\mathbf{U}_{\mathbf{Y}}\|_{2} & \leq\left(\|(p^{-1}\mathbf{Y}\mathbf{Y}^{\top}-\boldsymbol{\Phi}\boldsymbol{\Phi}^{\top}-\sigma^{2}\mathbf{I}_{n})\|_{2}+\sigma^{2}\right)\|\mathbf{U}_{\boldsymbol{\Phi}}^{\top}\mathbf{U}_{\mathbf{Y}}\|_{2}\\
 & \leq n\epsilon_{1}+\sigma^{2}.
\end{align*}
The bound on $\|\mathbf{F}_{\star}\boldsymbol{\Lambda}_{\mathbf{Y}}-\boldsymbol{\Lambda}_{\boldsymbol{\Phi}}\mathbf{F}_{\star}\|_{F}$
given in the statement holds by combining the above spectral norm
bounds.

For the bound on $\|\mathbf{F}_{\star}\boldsymbol{\Lambda}_{\mathbf{Y}}^{1/2}-\boldsymbol{\Lambda}_{\boldsymbol{\Phi}}^{1/2}\mathbf{F}_{\star}\|_{F}$
we use the fact that the elements of $\mathbf{F}_{\star}\boldsymbol{\Lambda}_{\mathbf{Y}}^{1/2}-\boldsymbol{\Lambda}_{\boldsymbol{\Phi}}^{1/2}\mathbf{F}_{\star}$
can be written:
\begin{align*}
(\mathbf{F}_{\star}\boldsymbol{\Lambda}_{\mathbf{Y}}^{1/2}-\boldsymbol{\Lambda}_{\boldsymbol{\Phi}}^{1/2}\mathbf{F}_{\star})_{ij} & =(\mathbf{F}_{\star})_{ij}\lambda_{j}(p^{-1}\mathbf{Y}\mathbf{Y}^{\top})^{1/2}-\lambda_{i}(\boldsymbol{\Phi}\boldsymbol{\Phi}^{\top})^{1/2}(\mathbf{F}_{\star})_{ij}\\
 & =(\mathbf{F}_{\star})_{ij}\frac{[\lambda_{j}(p^{-1}\mathbf{Y}\mathbf{Y}^{\top})-\lambda_{i}(\boldsymbol{\Phi}\boldsymbol{\Phi}^{\top})]}{\lambda_{j}(p^{-1}\mathbf{Y}\mathbf{Y}^{\top})^{1/2}+\lambda_{i}(\boldsymbol{\Phi}\boldsymbol{\Phi}^{\top})^{1/2}}
\end{align*}
hence
\[
|(\mathbf{F}_{\star}\boldsymbol{\Lambda}_{\mathbf{Y}}^{1/2}-\boldsymbol{\Lambda}_{\boldsymbol{\Phi}}^{1/2}\mathbf{F}_{\star})_{ij}|\leq\frac{|(\mathbf{F}_{\star}\boldsymbol{\Lambda}_{\mathbf{Y}}-\boldsymbol{\Lambda}_{\Phi}\mathbf{F}_{\star})_{ij}|}{2n^{1/2}(1-\epsilon_{2})^{1/2}(\lambda_{r}^{f})^{1/2}},
\]
and so 
\[
\|\mathbf{F}_{\star}\boldsymbol{\Lambda}_{\mathbf{Y}}^{1/2}-\boldsymbol{\Lambda}_{\boldsymbol{\Phi}}^{1/2}\mathbf{F}_{\star}\|_{F}\leq\frac{\|\mathbf{F}_{\star}\boldsymbol{\Lambda}_{\mathbf{Y}}-\boldsymbol{\Lambda}_{\Phi}\mathbf{F}_{\star}\|_{F}}{2n^{1/2}(1-\epsilon_{2})^{1/2}(\lambda_{r}^{f})^{1/2}}.
\]
The bound on $\|\mathbf{F}_{\star}\boldsymbol{\Lambda}_{\mathbf{Y}}^{-1/2}-\boldsymbol{\Lambda}_{\boldsymbol{\Phi}}^{-1/2}\mathbf{F}_{\star}\|_{F}$
in the statement is obtained in a similar manner using the fact that
for any $a,b>0$, $a^{-1/2}-b^{-1/2}=(b^{1/2}-a^{1/2})/(a^{1/2}b^{1/2})$
.
\end{proof}

\subsection{Some linear algebra}
\begin{lem}
\label{lem:AA^T}For any $m_{1},m_{2}\geq1$, $\mathbf{A}\in\mathbb{R}^{m_{2}\times m_{1}}$,
$q\leq\min\{m_{1},m_{2}\}$ and strictly positive real numbers $\lambda_{1},\ldots,\lambda_{q}$,
\\
a) there exists $\mathbf{U}\in\mathbb{R}^{m_{2}\times q}$ such that
$\mathbf{U}^{\top}\mathbf{U}=\mathbf{I}_{q}$ and $\mathbf{A}\mathbf{A}^{\top}\mathbf{U}=\mathbf{U}\boldsymbol{\Lambda}$,
if and only if there exists $\mathbf{V}\in\mathbb{R}^{m_{1}\times q}$
such that $\mathbf{V}^{\top}\mathbf{V}=\mathbf{I}_{q}$ and $\mathbf{A}^{\top}\mathbf{A}\mathbf{V}=\mathbf{V}\boldsymbol{\Lambda}$,
where $\boldsymbol{\Lambda}\coloneqq\mathrm{diag}(\lambda_{1},\ldots,\lambda_{q})$;
\\
b) when $\mathbf{V}$ with the properties stated in part a) exists,
a choice of $\mathbf{U}$ which has the properties stated in part
a) is $\mathbf{U}=\mathbf{A}\mathbf{V}\boldsymbol{\Lambda}^{-1/2}$;\\
c) $\lambda_{i}(\mathbf{A}^{\top}\mathbf{A})=\lambda_{i}(\mathbf{A}\mathbf{A}^{\top})$,
for $i=1,\ldots,\min\{m_{1},m_{2}\}$.\\
d) the rank of $\mathbf{A}^{\top}\mathbf{A}$ is equal to that of
$\mathbf{A}\mathbf{A}^{\top}$;
\end{lem}

\begin{proof}
Assume the existence of $\mathbf{V}$ with the properties stated in
part a). Taking $\mathbf{U}\coloneqq\mathbf{A}\mathbf{V}\boldsymbol{\Lambda}^{-1/2}$
we have
\begin{align*}
\mathbf{U}^{\top}\mathbf{U} & \coloneqq\boldsymbol{\Lambda}^{-1/2}\mathbf{V}^{\top}\mathbf{A}^{\top}\mathbf{A}\mathbf{V}\boldsymbol{\Lambda}^{-1/2}\\
 & =\boldsymbol{\Lambda}^{-1/2}\mathbf{V}^{\top}\mathbf{V}\boldsymbol{\Lambda}\boldsymbol{\Lambda}^{-1/2}\\
 & =\boldsymbol{\Lambda}^{-1/2}\boldsymbol{\Lambda}\boldsymbol{\Lambda}^{-1/2}=\mathbf{I}_{q}
\end{align*}
and 
\begin{align*}
\mathbf{A}\mathbf{A}^{\top}\mathbf{U} & =\mathbf{A}\mathbf{A}^{\top}\mathbf{A}\mathbf{V}\boldsymbol{\Lambda}^{-1/2}\\
 & =\mathbf{A}\mathbf{V}\boldsymbol{\Lambda}\boldsymbol{\Lambda}^{-1/2}\\
 & =\mathbf{U}\boldsymbol{\Lambda}.
\end{align*}
The implication in the other direction for part a) holds by interchanging
$\mathbf{A}^{\top}$ and $\mathbf{U}$ with respectively $\mathbf{A}$
and $\mathbf{V}$. We have thus proved parts a) and b) of the lemma.
Part a) implies that the non-zero eigenvalues of $\mathbf{A}^{\top}\mathbf{A}$
are equal to those of $\mathbf{A}\mathbf{A}^{\top}$, which establishes
the claim of part c). Part d) follows from part c).
\end{proof}

\begin{lem}
\label{lem:orthogonal_transformation}For any $m_{1}\leq m_{2}$ and
$\mathbf{A}\in\mathbb{R}^{m_{2}\times m_{1}}$ such that $\mathbf{A}$
has rank $m_{1}$, there exists an orthogonal matrix $\mathbf{Q}\in\mathbb{R}^{m_{1}\times m_{1}}$
such that $\mathbf{U}\boldsymbol{\Lambda}^{1/2}=\mathbf{A}\mathbf{Q}$,
where $\boldsymbol{\Lambda}=\mathrm{diag}\{\lambda_{1}(\mathbf{A}\mathbf{A}^{\top}),\cdots,\lambda_{m_{1}}(\mathbf{A}\mathbf{A}^{\top})\}$
and the columns of $\mathbf{U}\in\mathbb{R}^{m_{2}\times m_{1}}$
are orthonormal eigenvectors of $\mathbf{A}\mathbf{A}^{\top}$ with eigenvalues $\lambda_{1}(\mathbf{A}\mathbf{A}^{\top}),$ $\ldots,\lambda_{m_{1}}(\mathbf{A}\mathbf{A}^{\top})$.
\end{lem}

\begin{proof}
We have $\mathbf{A}\mathbf{A}^{\top}=\mathbf{U}\boldsymbol{\Lambda}\mathbf{U}^{\top}$,
hence $\mathbf{U}\boldsymbol{\Lambda}^{1/2}=\mathbf{A}\mathbf{A}^{\top}\mathbf{U}\boldsymbol{\Lambda}^{-1/2}$.
Take $\mathbf{Q}\coloneqq\mathbf{A}^{\top}\mathbf{U}\boldsymbol{\Lambda}^{-1/2}\in\mathbb{R}^{m_{1}\times m_{1}}$.
We then find:

\[
\mathbf{Q}^{\top}\mathbf{Q}=\boldsymbol{\Lambda}^{-1/2}\mathbf{U}^{\top}\mathbf{A}\mathbf{A}^{\top}\mathbf{U}\boldsymbol{\Lambda}^{-1/2}=\boldsymbol{\Lambda}^{-1/2}\mathbf{U}^{\top}\mathbf{U}\boldsymbol{\Lambda}\mathbf{U}^{\top}\mathbf{U}\boldsymbol{\Lambda}^{-1/2}=\mathbf{I}_{m_{1}}
\]
and 
\begin{equation}
\mathbf{Q}\mathbf{Q}^{\top}=\mathbf{A}^{\top}\mathbf{U}\boldsymbol{\Lambda}^{-1}\mathbf{U}^{\top}\mathbf{A}.\label{eq:QQ^T}
\end{equation}
Consider the reduced  singular value decomposition $\mathbf{A}=\mathbf{U}\boldsymbol{\Lambda}^{1/2}\mathbf{V}^{\top}$
where $\mathbf{V}\in\mathbb{R}^{m_{1}\times m_{1}}$ has orthonormal
columns. Substituting into the r.h.s. of (\ref{eq:QQ^T}),
\[
\mathbf{Q}\mathbf{Q}^{\top}=\mathbf{V}\boldsymbol{\Lambda}^{1/2}\mathbf{U}^{\top}\mathbf{U}\boldsymbol{\Lambda}^{-1}\mathbf{U}^{\top}\mathbf{U}\boldsymbol{\Lambda}^{1/2}\mathbf{V}^{\top}=\mathbf{V}\mathbf{V}^{\top}=\mathbf{I}_{m_{1}}.
\]
\end{proof}

\subsection{Some properties of the LMM}
\begin{lem}
\label{lem:PC_identity}On the event that the rank of $\mathbf{Y}^{\top}\mathbf{Y}$
is at least $r$, $p^{-1/2}\mathbf{Y}\mathbf{V}_{\mathbf{Y}}=\mathbf{U}_{\mathbf{Y}}\boldsymbol{\Lambda}_{\mathbf{Y}}^{1/2}$,
where the columns
of $\mathbf{U}_{\mathbf{Y}}\in\mathbb{R}^{n\times r}$ are orthonormal
eigenvectors of $p^{-1}\mathbf{Y}\mathbf{Y}^{\top}$  with associated eigenvalues on the diagonal of the diagonal matrix $\boldsymbol{\Lambda}_{\mathbf{Y}}\in\mathbb{R}^{r\times r}$.
\end{lem}

\begin{proof}
Apply lemma \ref{lem:AA^T}, part b).
\end{proof}
\noindent Thus by computing the PCA embedding $\zeta_1,\ldots,\zeta_n$ and rescaling by $p^{-1/2}$, we are,
in effect, computing the $n$ rows of $\mathbf{U}_{\mathbf{Y}}\boldsymbol{\Lambda}_{\mathbf{Y}}^{1/2}$,
where $\mathbf{U}_{\mathbf{Y}}\boldsymbol{\Lambda}_{\mathbf{Y}}^{1/2}(\mathbf{U}_{\mathbf{Y}}\boldsymbol{\Lambda}_{\mathbf{Y}}^{1/2})^{\top}=\mathbf{U}_{\mathbf{Y}}\boldsymbol{\Lambda}_{\mathbf{Y}}\mathbf{U}_{\mathbf{Y}}^{\top}$
is a rank$-r$ approximation to $p^{-1}\mathbf{Y}\mathbf{Y}^{\top}$.

\begin{lem}
\label{lem:conditiona_expectation_of_Y^TY}Assume \ref{ass:cont_covar} and \ref{ass:finite rank}. Then $p^{-1}\mathbb{E}[\mathbf{Y}\mathbf{Y}^{\top}|Z_{1},\ldots,Z_{n}]=\boldsymbol{\Phi}\boldsymbol{\Phi}^{\top}+\sigma^{2}\mathbf{I}_{n}$.
\end{lem}

\begin{proof}
Let $\mathbf{X}\in\mathbb{R}^{n\times p}$
be the matrix with entries $\mathbf{X}_{ij}\coloneqq X_{j}(Z_{i})$.
According to the model specification in section \ref{sec:Model},
$\mathbf{X}$ and $\mathbb{\mathbf{E}}$ are independent, and $\mathbb{E}[\mathbf{E}\mathbf{E}^{\top}]=p\mathbf{I}_{n}$.
Thus: 
\begin{align*}
\mathbb{E}[\mathbf{Y}\mathbf{Y}^{\top}|Z_{1},\ldots,Z_{n}] & =\mathbb{E}[\mathbf{X}\mathbf{X}^{\top}|Z_{1},\ldots,Z_{n}]+\sigma\mathbb{E}[\mathbf{\mathbf{X}}\mathbf{E}^{\top}|Z_{1},\ldots,Z_{n}]\\
 & \quad+\sigma\mathbb{E}[\mathbf{E}\mathbf{X}^{\top}|Z_{1},\ldots,Z_{n}]+\sigma^{2}\mathbb{E}[\mathbf{E}\mathbf{E}^{\top}|Z_{1},\ldots,Z_{n}]\\
 & =p\boldsymbol{\Phi}\boldsymbol{\Phi}^{\top}+p\sigma^{2}\mathbf{I}_{n}.
\end{align*}
\end{proof}

\begin{lem}\label{lem:c_lambda_constants}
Assume \ref{ass:cont_covar}, \ref{ass:finite rank} and \ref{ass:moments}. Then there exists a constant  $c_{\lambda}^{\mathrm{max}}<\infty$ depending only on the first supremum in \ref{ass:moments},  and a constant $c_{\lambda}^{\mathrm{min}}>0$ such that
$$
\sup_{p\geq 1}\left\{\sup_z f(z,z)+\lambda_1^f\right\}\leq c_{\lambda}^{\mathrm{max}},\qquad\inf_{p\geq 1} \lambda_r^f \geq c_{\lambda}^{\mathrm{min}}.
$$
\end{lem}
\begin{proof}
The existence of $c_{\lambda}^{\mathrm{min}}$ as required is an immediate consequence of \ref{ass:finite rank}. Using \ref{ass:moments} and Jensen's inequality gives:
$$
\sup_z f(z,z) = \sup_z \frac{1}{p}\sum_{j=1}^p\mathbb{E}[|X_j(z)|^2]\leq \sup_z \frac{1}{p}\sum_{j=1}^p\mathbb{E}[|X_j(z)|^{4q}]^{2/4q}<\infty.  
$$
The existence of $c_{\lambda}^{\mathrm{max}}$ as required follows from the above inequalities combined with:
$$
\lambda_1^f\leq \sum_{k=1}^\infty \lambda_k^f = \sum_{k=1}^\infty \lambda_k^f \mathbb{E}[|u_k^f(Z_1)|^2] = \mathbb{E}[f(Z_1,Z_1)]\leq \sup_z f(z,z). 
$$
\end{proof}

\subsection{Matrix concentration results}\label{subsec:matrix_conc}

The following matrix-valued version of the Bernstein inequality can
be found in, e.g., \citep[Thm 1.6.2]{tropp2015introduction}
\begin{thm}
[Matrix Bernstein inequality]\label{thm:matrix_bern}Let $\mathbf{M}_{1},\ldots,\mathbf{M}_{n}$
be independent random matrices with common dimensions $m_{1}\times m_{2}$
satisfying $\mathbb{E}[\mathbf{M}_{i}]=0$ and $\|\mathbf{M}_{i}\|_{2}\leq L$
for each $1\leq i\leq n$ and some constant $L$. Let $\mathbf{M}\coloneqq\sum_{i=1}^{n}\mathbf{M}_{i}$
and $v(\mathbf{M})=\max\left\{ \|\mathbb{E}[\mathbf{M}\mathbf{M}^{\top}]\|_{2},\|\mathbb{E}[\mathbf{M}^{\top}\mathbf{M}]\|_{2}\right\} $.
Then for all $t\geq0$, 
\[
\mathbb{P}\left(\|\mathbf{M}\|_{2}\geq t\right)\leq(m_{1}+m_{2})\exp\left(\frac{-t^{2}/2}{v(\mathbf{M})+Lt/3}\right).
\]
\end{thm}

\begin{lem}
\label{lem:phiphi^T_conc}Assume \ref{ass:cont_covar} and \ref{ass:finite rank}.
For any $t\geq0$, 
\[
\mathbb{P}\left(\|n^{-1}\boldsymbol{\Phi}^{\top}\boldsymbol{\Phi}-n^{-1}\mathbb{E}[\boldsymbol{\Phi}^{\top}\boldsymbol{\Phi}]\|_{2}\geq t\right)\leq2r\exp\left(\frac{-t^{2}n/2}{(c_{\lambda}^{\mathrm{max}})^{2}+c_{\lambda}^{\mathrm{max}}t/3}\right),
\]
where $c_{\lambda}^{\mathrm{max}}$ is as in lemma \ref{lem:c_lambda_constants}.
\end{lem}

\begin{proof}
Apply theorem \ref{thm:matrix_bern} with $\mathbf{M}_{i}=\frac{1}{n}\phi(Z_{i})\phi(Z_{i})^{\top}-\mathbb{E}[\frac{1}{n}\phi(Z_{i})\phi(Z_{i})^{\top}],$
\begin{align*}
\|\mathbf{M}_{i}\|_{2} & \leq\frac{1}{n}\|\phi(Z_{i})\phi(Z_{i})^{\top}\|_{2}+\frac{1}{n}\|\mathbb{E}[\phi(Z_{i})\phi(Z_{i})^{\top}]\|_{2}\\
 & =\frac{1}{n}\|\phi(Z_{i})\|_{2}^{2}+\frac{1}{n}\lambda_{1}^{f}\\
 & =\frac{1}{n}f(Z_{i},Z_{i})+\frac{1}{n}\lambda_{1}^{f}\\
 & \leq\frac{1}{n}c_{\lambda}^{\mathrm{max}}=:L
\end{align*}
and
\begin{align*}
v(\mathbf{M}) & =\left\Vert \mathbb{E}\left[\left(\sum_{i}\mathbf{M}_{i}\right)\left(\sum_{i}\mathbf{M}_{i}\right)\right]\right\Vert _{2}\\
 & =\left\Vert \mathbb{E}\left[\sum_{i}\mathbf{M}_{i}\mathbf{M}_{i}\right]\right\Vert _{2}\\
 & \leq\frac{1}{n}\left\Vert \mathbb{E}\left[\phi(Z_{1})\phi(Z_{1})^{\top}\phi(Z_{1})\phi(Z_{1})^{\top}\right]\right\Vert _{2}+\frac{1}{n}\|\mathbb{E}[\phi(Z_{1})\phi(Z_{1})^{\top}]^{2}\|_{2}\\
 & \leq\frac{1}{n}\mathbb{E}\left[\left\Vert \phi(Z_{1})\phi(Z_{1})^{\top}\phi(Z_{1})\phi(Z_{1})^{\top}\right\Vert _{2}\right]+\frac{1}{n}\|\mathbb{E}[\phi(Z_{1})\phi(Z_{1})^{\top}]^{2}\|_{2}\\
 & =\frac{1}{n}\mathbb{E}\left[\left\Vert \phi(Z_{1})\phi(Z_{1})^{\top}\right\Vert _{2}^{2}\right]+\frac{1}{n}(\lambda_{1}^{f})^{2}\\
 & =\frac{1}{n}\mathbb{E}\left[\|\phi(Z_{1})\|_{2}^{4}\right]+\frac{1}{n}(\lambda_{1}^{f})^{2}\leq\frac{1}{n}(c_{\lambda}^{\mathrm{max}})^{2}.
\end{align*}
\end{proof}
\begin{lem}
\label{lem:Y-phi-I_conc}Assume  \ref{ass:cont_covar}, \ref{ass:finite rank}
and \ref{ass:moments} with some $q\geq1$. Then for any $t>0$, 
\[
\mathbb{P}\left(\|p^{-1}\mathbf{Y}\mathbf{Y}^{\top}-\boldsymbol{\Phi}\boldsymbol{\Phi}^{\top}-\sigma^{2}\mathbf{I}_{n}\|_{2}\geq t\right)\leq(16)^{q}(2q-1)^{q}\frac{n^{2q}}{t^{2q}}\frac{1}{p^{q}}\left(c_{X}(2q)^{1/2q}+\sigma^{2}c_{E}(2q)^{1/2q}\right)^{2q}
\]
 where 
\[
c_{X}(q)\coloneqq\sup_{j\geq1}\sup_{z\in\mathcal{Z}}\mathbb{E}\left[\left|X_{j}(z)\right|^{2q}\right],\qquad c_{E}(q)\coloneqq\sup_{j\geq1} \sup_{i\geq1}\mathbb{E}\left[\left|\mathbf{E}_{ij}\right|^{2q}\right].
\]
\end{lem}

\begin{proof}
Let us write the matrix $\mathbf{Y}$ in terms of its columns $\mathbf{Y}\equiv[Y_{1}|\cdots|Y_{p}]$
so that:

\begin{equation}
\mathbf{Y}\mathbf{Y}^{\top}=\sum_{j=1}^{p}Y_{j}Y_{j}^{\top}.\label{eq:Y^TY}
\end{equation}
Observe that under the model of section \ref{sec:Model}, conditional
on $(Z_{1},\ldots,Z_{n})$ the summands in (\ref{eq:Y^TY}) are independent
and as per lemma \ref{lem:conditiona_expectation_of_Y^TY}, the conditional
expectation of $\mathbf{Y}\mathbf{Y}^{\top}$ given $Z_{1},\ldots,Z_{n}$
is: $p\boldsymbol{\Phi}\boldsymbol{\Phi}^{\top}+p\sigma^{2}\mathbf{I}_{n}$.

The main tool we use from hereon is a direct combination of the matrix
Chebyshev inequality \citep[Prop. 3.1]{paulin2016efron} and the matrix
polynomial Effron-Stein inequality \citep[Thm 4.2]{paulin2016efron},
applied under the regular conditional distribution of $(Y_{1},\ldots,Y_{p})$
given $(Z_{1},\ldots,Z_{n})$. These inequalities taken together tell
us that, for any $q\geq1$, the following holds almost surely:
\begin{align*}
&\mathbb{P}\left(\left|\|p^{-1}\mathbf{Y}\mathbf{Y}^{\top}-\boldsymbol{\Phi}\boldsymbol{\Phi}^{\top}-\sigma^{2}\mathbf{I}_{n}\|_{2}\geq t\right|Z_{1},\ldots,Z_{n}\right) \\
& \leq\frac{1}{t^{2q}}\mathbb{E}\left[\left.\|p^{-1}\mathbf{Y}\mathbf{Y}^{\top}-\boldsymbol{\Phi}\boldsymbol{\Phi}^{\top}-\sigma^{2}\mathbf{I}_{n}\|_{S_{2q}}^{2q}\right|Z_{1},\ldots,Z_{n}\right]\\
 & \leq\frac{2^{q}(2q-1)^{q}}{t^{2q}}\mathbb{E}\left[\left.\|\boldsymbol{\Sigma}\|_{S_{q}}^{q}\right|Z_{1},\ldots,Z_{n}\right].
\end{align*}
Here $\|\cdot\|_{S_{q}}$ is the Schatten $q$-norm and $\boldsymbol{\Sigma}\in\mathbb{R}^{n\times n}$
is the variance proxy:
\begin{equation}
\boldsymbol{\Sigma}\coloneqq\frac{1}{2p^{2}}\sum_{j=1}^{p}\mathbb{E}\left[\left.\left(Y_{j}Y_{j}^{\top}-\tilde{Y}_{j}\tilde{Y}_{j}^{\top}\right)^{2}\right|Y_{j},Z_{1},\ldots,Z_{n}\right],\label{eq:variance_proxy_effron_Stein-1}
\end{equation}
where, conditional on $Z_{1},\ldots,Z_{n}$, $\tilde{Y}_{j}$ is an
independent copy of $Y_{j}$. For brevity in the remainder of the
proof we shall write $Z\equiv(Z_{1},\ldots,Z_{n})$, and to avoid
repetitive statements of ``almost surely'', every inequality involving
conditional expectations is to be understood as holding in the almost
sure sense.

We estimate:

\begin{align}
\mathbb{E}\left[\left.\|\boldsymbol{\Sigma}\|_{S_{q}}^{q}\right|Z\right]^{1/q} & =\frac{1}{2p^{2}}\mathbb{E}\left[\left.\left\Vert \sum_{j=1}^{p}\mathbb{E}\left[\left.\left(Y_{j}Y_{j}^{\top}-\tilde{Y}_{j}\tilde{Y}_{j}^{\top}\right)^{2}\right|Y_{j},Z\right]\right\Vert _{S_{q}}^{q}\right|Z\right]^{1/q}\nonumber \\
 & \leq\frac{1}{2p^{2}}\sum_{j=1}^{p}\mathbb{E}\left[\left.\left\Vert \mathbb{E}\left[\left.\left(Y_{j}Y_{j}^{\top}-\tilde{Y}_{j}\tilde{Y}_{j}^{\top}\right)^{2}\right|Y_{j},Z\right]\right\Vert _{S_{q}}^{q}\right|Z\right]^{1/q}\label{eq:apply_matrix_mink-1}\\
 & \leq\frac{1}{2p^{2}}\sum_{j=1}^{p}\mathbb{E}\left[\left.\left\Vert \left(Y_{j}Y_{j}^{\top}-\tilde{Y}_{j}\tilde{Y}_{j}^{\top}\right)^{2}\right\Vert _{S_{q}}^{q}\right|Z\right]^{1/q}\label{eq:apply_norm_convexity-1}\\
 & =\frac{1}{2p^{2}}\sum_{j=1}^{p}\mathbb{E}\left[\left.\left\Vert Y_{j}Y_{j}^{\top}-\tilde{Y}_{j}\tilde{Y}_{j}^{\top}\right\Vert _{S_{2q}}^{2q}\right|Z\right]^{1/q}\nonumber \\
 & \leq\frac{1}{2p^{2}}\sum_{j=1}^{p}\left(2\mathbb{E}\left[\left.\left\Vert Y_{j}Y_{j}^{\top}\right\Vert _{S_{2q}}^{2q}\right|Z\right]^{1/2q}\right)^{2}\label{eq:apply_matrix_mink_again-1}\\
 & =\frac{2}{p^{2}}\sum_{j=1}^{p}\mathbb{E}\left[\left.\left\Vert Y_{j}Y_{j}^{\top}\right\Vert _{S_{2q}}^{2q}\right|Z\right]^{1/q}\nonumber 
\end{align}
Here (\ref{eq:apply_matrix_mink-1}) holds by the second claim of
lemma \ref{lem:matrix_minkowski}; \ref{eq:apply_norm_convexity-1}
holds by first claim of lemma \ref{lem:matrix_minkowski} combined
with the fact that $x\mapsto x^{q}$ is convex for $x\geq0$ (recall
$q\geq1$); (\ref{eq:apply_matrix_mink_again-1}) holds by lemma \ref{lem:matrix_minkowski}
and the fact that $\tilde{Y}_{j}$ and $Y_{j}$ are equal in distribution.

By definition of the Schatten-$q$ norm, $\left\Vert Y_{j}Y_{j}^{\top}\right\Vert _{S_{2q}}^{2q}=\sum_{k=1}^{n}\lambda_{k}^{2q}\left(Y_{j}Y_{j}^{\top}\right)$,
where $\lambda_{1}\left(Y_{j}Y_{j}^{\top}\right)=\|Y_{j}\|_{2}^{2}$
and $\lambda_{k}\left(Y_{1}Y_{1}^{\top}\right)=0$ for $k=2,\ldots,n$.
Thus: 
\begin{equation}
\left\Vert Y_{j}Y_{j}^{\top}\right\Vert _{S_{2q}}^{2q}=\|Y_{j}\|_{2}^{4q}=\left|\sum_{i=1}^{n}\left(X_{j}(Z_{i})+\sigma\mathbf{E}_{ij}\right)^{2}\right|^{2q}.\label{eq:schatten_q_W-1}
\end{equation}
By two applications of Minkowski's inequality,
\begin{align*}
\mathbb{E}\left[\left.\left\Vert Y_{j}Y_{j}^{\top}\right\Vert _{S_{2q}}^{2q}\right|Z\right]^{1/2q} & \leq\sum_{i=1}^{n}\mathbb{E}\left[\left.\left|X_{j}(Z_{i})+\sigma\mathbf{E}_{ij}\right|^{4q}\right|Z\right]^{1/2q}\\
 & \leq2\sum_{i=1}^{n}\mathbb{E}\left(\left[\left.\left|X_{j}(Z_{i})\right|^{4q}\right|Z\right]^{1/2q}+\mathbb{E}\left[\left.\left|\sigma\mathbf{E}_{ij}\right|^{4q}\right|Z\right]^{1/2q}\right)\\
 & \leq2n\left(\sup_{l\geq1}\sup_{z\in\mathcal{Z}}\mathbb{E}\left[\left|X_{l}(z)\right|^{4q}\right]^{1/2q}+\sigma^{2}\sup_{i\geq1,l\geq1}\mathbb{E}\left[\left|\mathbf{E}_{il}\right|^{4q}\right]^{1/2q}\right),
\end{align*}
where the final inequality uses the facts that $X_{j}$, $Z$ and
$\mathbf{E}$ are independent.

Combining the above estimates we find:
\begin{align*}
 & \mathbb{P}\left(\left|\|p^{-1}\mathbf{Y}\mathbf{Y}^{\top}-\boldsymbol{\Phi}\boldsymbol{\Phi}^{\top}-\sigma^{2}\mathbf{I}_{n}\|_{2}\geq t\right|Z_{1},\ldots,Z_{n}\right)\\
 & \leq\frac{2^{q}(2q-1)^{q}}{t^{2q}}\left(\frac{2}{p}\right)^{q}4^{q}n^{2q}\left(\sup_{j\geq 1}\sup_{z\in\mathcal{Z}}\mathbb{E}\left[\left|X_{j}(z)\right|^{4q}\right]^{1/2q}+\sigma^{2}\sup_{i\geq1,j\geq 1}\mathbb{E}\left[\left|\mathbf{E}_{ij}\right|^{4q}\right]^{1/2q}\right)^{2q}\\
 & =(16)^{q}(2q-1)^{q}\frac{n^{2q}}{t^{2q}}\frac{1}{p^{q}}\left(\sup_{j\geq1}\sup_{z\in\mathcal{Z}}\mathbb{E}\left[\left|X_{j}(z)\right|^{4q}\right]^{1/2q}+\sigma^{2}\sup_{i\geq1,j\geq 1}\mathbb{E}\left[\left|\mathbf{E}_{ij}\right|^{4q}\right]^{1/2q}\right)^{2q},
\end{align*}
from which the result follows by the tower property of conditional
expectation.
\end{proof}
\begin{prop}
\label{prop:eigenvalue_perturbation}Assume 
\ref{ass:cont_covar}, \ref{ass:finite rank} and \ref{ass:moments} with some $q\geq1$.
For any $\delta,\epsilon\in(0,1)$, if 
\[
n\geq\frac{3\sigma^{2}}{\epsilon c_{\lambda}^{\mathrm{min}}}\vee\left[\log\left(\frac{1}{\delta}\right)+\log(4r)\right]\frac{1}{\epsilon^{2}}\frac{2((c_{\lambda}^{\mathrm{max}})^{2}+\epsilon c_{\lambda}^{\mathrm{max}}c_{\lambda}^{\mathrm{min}}/9)}{(c_{\lambda}^{\mathrm{min}})^{2}/9},
\]
and 
\[
p\geq\frac{1}{\delta^{1/q}\epsilon^{2}}2^{1/q}16(2q-1)\frac{9}{(c_{\lambda}^{\mathrm{min}})^{2}}\left(c_{X}(2q)^{1/2q}+\sigma^{2}c_{E}(2q)^{1/2q}\right)^{2}
\]
where $c_{X}$, $c_{E}$ are as in lemma \ref{lem:Y-phi-I_conc} and $c_{\lambda}^{\mathrm{max}}$, $c_{\lambda}^{\mathrm{min}}$ are as in lemma \ref{lem:c_lambda_constants},
then 
\[
\mathbb{P}\left(\bigcap_{i=1}^{n}B_{\mathbf{Y},i}(\epsilon)\cap\bigcap_{i=1}^{r}B_{\boldsymbol{\Phi},i}(\epsilon)\right)\geq1-\delta.
\]
\end{prop}

\begin{proof}
Throughout the proof we shall adopt the convention $\lambda_{i}^{f}\coloneqq0$
for all $r+1\leq i\leq n$ and, in several places, we shall use the
fact that $\lambda_{i}(\boldsymbol{\Phi}\boldsymbol{\Phi}^{\top})=0$
for $r+1\leq i\leq n$ which holds since $\boldsymbol{\Phi}\in\mathbb{R}^{n\times r}$.

Consider the following decomposition for any $1\leq i\leq n$:
\begin{align*}
\left|\frac{1}{n}\lambda_{i}(p^{-1}\mathbf{Y}\mathbf{Y}^{\top})-\lambda_{i}^{f}\right| & \leq\left|\frac{1}{n}\lambda_{i}(p^{-1}\mathbf{Y}\mathbf{Y}^{\top})-\frac{1}{n}\lambda_{i}(\boldsymbol{\Phi}\boldsymbol{\Phi}^{\top}+\sigma^{2}\mathbf{I}_{n})\right|\\
 & +\left|\frac{1}{n}\lambda_{i}(\boldsymbol{\Phi}\boldsymbol{\Phi}^{\top}+\sigma^{2}\mathbf{I}_{n})-\frac{1}{n}\lambda_{i}(\boldsymbol{\Phi}\boldsymbol{\Phi}^{\top})\right|\\
 & +\left|\frac{1}{n}\lambda_{i}(\boldsymbol{\Phi}\boldsymbol{\Phi}^{\top})-\lambda_{i}^{f}\right|.
\end{align*}
Combining this decomposition with Weyl's inequality; the facts that
for $1\leq i\leq r$, $\mathbb{E}[\boldsymbol{\Phi}^{\top}\boldsymbol{\Phi}]_{ii}=n\lambda_{i}^{f}$
and $\mathbb{E}[\boldsymbol{\Phi}^{\top}\boldsymbol{\Phi}]_{ij}=0$
for $j\neq i$, hence $\lambda_{i}^{f}=\lambda_{i}(n^{-1}\mathbb{E}[\boldsymbol{\Phi}^{\top}\boldsymbol{\Phi}])$;
and by lemma \ref{lem:AA^T}, $\lambda_{i}(\boldsymbol{\Phi}\boldsymbol{\Phi}^{\top})=\lambda_{i}(\boldsymbol{\Phi}^{\top}\boldsymbol{\Phi})$;
whilst for $i\geq r+1$, $\lambda_{i}(\boldsymbol{\Phi}\boldsymbol{\Phi}^{\top})=\lambda_{i}^{f}=0$;
we obtain:
\begin{align}
\max_{1\leq i\leq n}\left|\frac{1}{n}\lambda_{i}(p^{-1}\mathbf{Y}\mathbf{Y}^{\top})-\lambda_{i}^{f}\right| & \leq\frac{1}{n}\|p^{-1}\mathbf{Y}\mathbf{Y}^{\top}-\boldsymbol{\Phi}\boldsymbol{\Phi}^{\top}-\sigma^{2}\mathbf{I}_{n}\|_{2}\nonumber \\
 & +\frac{\sigma^{2}}{n}\nonumber \\
 & +\|n^{-1}\boldsymbol{\Phi}^{\top}\boldsymbol{\Phi}-n^{-1}\mathbb{E}[\boldsymbol{\Phi}^{\top}\boldsymbol{\Phi}]\|_{2}\label{eq:weyl_Y_F}
\end{align}
and 
\[
\max_{1\leq i\leq n}\left|\frac{1}{n}\lambda_{i}(\boldsymbol{\Phi}\boldsymbol{\Phi}^{\top})-\lambda_{i}^{f}\right|\leq\|n^{-1}\boldsymbol{\Phi}^{\top}\boldsymbol{\Phi}-n^{-1}\mathbb{E}[\boldsymbol{\Phi}^{\top}\boldsymbol{\Phi}]\|_{2}.
\]
Now fix any $\epsilon\in(0,1)$. We have 
\begin{align*}
 & \mathbb{P}\left(\bigcap_{i=1}^{n}B_{\mathbf{Y},i}(\epsilon)\cap\bigcap_{i=1}^{r}B_{\boldsymbol{\Phi},i}(\epsilon)\right)\\
 & \geq\mathbb{P}\left(\bigcap_{i=1}^{n}\left\{ \left|\frac{1}{n}\lambda_{i}(p^{-1}\mathbf{Y}\mathbf{Y}^{\top})-\lambda_{i}^{f}\right|<\epsilon\lambda_{r}^{f}\right\} \cap\left\{ \left|\frac{1}{n}\lambda_{i}(\boldsymbol{\Phi}\boldsymbol{\Phi}^{\top})-\lambda_{i}^{f}\right|<\epsilon\lambda_{r}^{f}\right\} \right)\\
 & \geq1-\mathbb{P}\left(\frac{1}{n}\|p^{-1}\mathbf{Y}\mathbf{Y}^{\top}-\boldsymbol{\Phi}\boldsymbol{\Phi}^{\top}-\sigma^{2}\mathbf{I}_{n}\|_{2}\geq\epsilon\lambda_{r}^{f}/3\right)-\mathbb{P}\left(\|n^{-1}\boldsymbol{\Phi}^{\top}\boldsymbol{\Phi}-n^{-1}\mathbb{E}[\boldsymbol{\Phi}^{\top}\boldsymbol{\Phi}]\|_{2}\geq\epsilon\lambda_{r}^{f}/3\right)\\
 & \geq1-(16)^{q}(2q-1)^{q}\frac{1}{(\epsilon c_{\lambda}^{\mathrm{min}}/3)^{2q}}\frac{1}{p^{q}}\left(c_{X}(2q)^{1/2q}+\sigma^{2}c_{E}(2q)^{1/2q}\right)^{2q}-2r\exp\left(\frac{-(\epsilon/3)^{2}(c_{\lambda}^{\mathrm{min}})^{2}n/2}{(c_{\lambda}^{\mathrm{max}})^{2}+c_{\lambda}^{\mathrm{max}}\epsilon c_{\lambda}^{\mathrm{min}}/9}\right)
\end{align*}
where the second inequality holds by using $\lambda_{r}^{f}\leq\lambda_{i}^{f}$
for $i=1,\ldots r,$ together with (\ref{eq:weyl_Y_F}) and the condition
of the proposition $n\geq 3\sigma^{2}/(\epsilon\lambda_{r}^{f})$;
and the third inequality holds by applying lemma \ref{lem:phiphi^T_conc}
and lemma \ref{lem:Y-phi-I_conc} and using $\lambda_r^f\geq c_{\lambda}^{\mathrm{min}}$.

The proof is completed by re-arranging each of the two following inequalities:
\[
\delta/2\geq(16)^{q}(2q-1)^{q}\frac{1}{(\epsilon c_{\lambda}^{\mathrm{min}}/3)^{2q}}\frac{1}{p^{q}}\left(c_{X}(2q)^{1/2q}+\sigma^{2}c_{E}(2q)^{1/2q}\right)^{2q},
\]
\[
\frac{\delta}{2}\geq2r\exp\left(\frac{-(\epsilon/3)^{2}(c_{\lambda}^{\mathrm{min}})^{2}n/2}{(c_{\lambda}^{\mathrm{max}})^{2}+c_{\lambda}^{\mathrm{max}}\epsilon c_{\lambda}^{\mathrm{min}}/9}\right).
\]
\end{proof}
\begin{lem}
\label{lem:matrix_minkowski}For any $m_{1},m_{2}\geq1$ and any matrix
norm $\|\cdot\|_{\star}$ on $\mathbb{R}^{m_{1}\times m_{2}}$, $\|\cdot\|_{\star}$
is convex. For any random $\mathbf{A},\mathbf{B}\in\mathbb{R}^{m_{1}\times m_{2}}$
and any $1\leq q<\infty$ such that $\mathbb{E}\left[\|\mathbf{A}\|_{\star}^{q}\right]\vee\mathbb{E}\left[\|\mathbf{B}\|_{\star}^{q}\right]<\infty$,
$\mathbb{E}\left[\|\mathbf{A}+\mathbf{B}\|_{\star}^{q}\right]^{1/q}\leq\mathbb{E}\left[\|\mathbf{A}\|_{\star}^{q}\right]^{1/q}+\mathbb{E}\left[\|\mathbf{B}\|_{\star}^{q}\right]^{1/q}$.
\end{lem}

\begin{proof}
The convexity holds due to the fact that any norm must be absolutely
homogeneous and satisfy the triangle inequality. For the second claim,
since $\mathbb{E}\left[\|\mathbf{A}\|_{\star}^{q}\right]\vee\mathbb{E}\left[\|\mathbf{B}\|_{\star}^{q}\right]<\infty$
we have the preliminary estimate $\mathbb{E}\left[\|\mathbf{A}+\mathbf{B}\|_{\star}^{q}\right]\leq2^{q-1}(\mathbb{E}\left[\|\mathbf{A}\|_{\star}^{q}\right]+\mathbb{E}\left[\|\mathbf{B}\|_{\star}^{q}\right])<\infty$.
If $\mathbb{E}\left[\|\mathbf{A}+\mathbf{B}\|_{\star}^{q}\right]=0$
then the desired inequality is trivial. So suppose $\mathbb{E}\left[\|\mathbf{A}+\mathbf{B}\|_{\star}^{q}\right]>0$.
Using the triangle inequality for the norm and then Holder's inequality
for the expectation,
\begin{align*}
\mathbb{E}\left[\|\mathbf{A}+\mathbf{B}\|_{\star}^{q}\right] & =\mathbb{E}\left[\|\mathbf{A}+\mathbf{B}\|_{\star}\|\mathbf{A}+\mathbf{B}\|_{\star}^{q-1}\right]\\
 & \leq\mathbb{E}\left[\left(\|\mathbf{A}\|_{\star}+\|\mathbf{B}\|_{\star}\right)\|\mathbf{A}+\mathbf{B}\|_{\star}^{q-1}\right]\\
 & =\mathbb{E}\left[\|\mathbf{A}\|_{\star}\|\mathbf{A}+\mathbf{B}\|_{\star}^{q-1}\right]+\mathbb{E}\left[\|\mathbf{B}\|_{\star}\|\mathbf{A}+\mathbf{B}\|_{\star}^{q-1}\right]\\
 & \leq\left(\mathbb{E}\left[\|\mathbf{A}\|_{\star}^{q}\right]^{1/q}+\mathbb{E}\left[\|\mathbf{B}\|_{\star}^{q}\right]^{1/q}\right)\mathbb{E}\left[\|\mathbf{A}+\mathbf{B}\|_{\star}^{(q-1)(\frac{q}{q-1})}\right]^{1-\frac{1}{q}}\\
 & =\left(\mathbb{E}\left[\|\mathbf{A}\|_{\star}^{q}\right]^{1/q}+\mathbb{E}\left[\|\mathbf{B}\|_{\star}^{q}\right]^{1/q}\right)\frac{\mathbb{E}\left[\|\mathbf{A}+\mathbf{B}\|_{\star}^{q}\right]}{\mathbb{E}\left[\|\mathbf{A}+\mathbf{B}\|_{\star}^{q}\right]^{1/q}}.
\end{align*}
The proof is completed by multiplying both sides by $\mathbb{E}\left[\|\mathbf{A}+\mathbf{B}\|_{\star}^{q}\right]^{1/q}/\mathbb{E}\left[\|\mathbf{A}+\mathbf{B}\|_{\star}^{q}\right]$.
\end{proof}
\begin{lem}
\label{lem:two_to_infty}Assume \ref{ass:cont_covar},  \ref{ass:finite rank},
and \ref{ass:moments} with some $q\geq1$. Let $U_{j}$ denote the
$j$th column of $\mathbf{U}_{\boldsymbol{\Phi}}.$ Then there exists
a constant $b(q)$ depending only on $q$ such that for any $t>0$,
\begin{align*}
 & \mathbb{P}\left(\max_{j=1,\ldots,r}\|(p^{-1}\mathbf{Y}\mathbf{Y}^{\top}-\boldsymbol{\Phi}\boldsymbol{\Phi}^{\top}-\sigma^{2}\mathbf{I}_{n})U_{j}\|_{\infty}\leq t\right)\\
 & \geq1-\frac{n^{1+q}r}{t^{2q}p^{q}}b(2q)2^{6q-1}\left(\max_{j=1,\ldots,p}\sup_{z\in\mathcal{Z}}\mathbb{E}\left[|X_{j}(z)|^{4q}\right]+\sigma^{4q}\max_{i=1,\ldots,n,j=1,\ldots,p}\mathbb{E}[|\mathbf{E}_{ij}|^{4q}]\right).
\end{align*}
\end{lem}

\begin{proof}
The $i$th element of $(p^{-1}\mathbf{Y}\mathbf{Y}^{\top}-\boldsymbol{\Phi}\boldsymbol{\Phi}^{\top}-\sigma^{2}\mathbf{I}_{n})U_{j}$
can be written in the form:
\[
p^{-1}\sum_{k=1}^{p}\boldsymbol{\Delta}_{ij}(k)
\]
 where 
\[
\boldsymbol{\Delta}_{ij}(k)\coloneqq Y_{k}^{(i)}Y_{k}^{\top}U_{j}-\mathbb{E}\left[\left.Y_{k}^{(i)}Y_{k}^{\top}U_{j}\right|Z_{1},\ldots,Z_{n}\right]
\]
and for any $i,j$, the random variables $\boldsymbol{\Delta}_{ij}(k)$,
$k=1,\ldots,p$ are conditionally independent and conditionally mean
zero given $Z_{1},\ldots,Z_{n}$.

Applying Markov's inequality, the Marcinkiewicz-Zygmund inequality
and Minkowski's inequality, all conditionally on $Z\equiv(Z_{1},\ldots,Z_{n})$,
we have for any $q\geq1$ the following inequalities hold almost surely,
\begin{align}
\mathbb{P}\left(\left.\left|p^{-1}\sum_{k=1}^{p}\boldsymbol{\Delta}_{ij}(k)\right|\geq t\right|Z\right) & \leq\frac{1}{t^{2q}}\mathbb{E}\left[\left.\left|p^{-1}\sum_{k=1}^{p}\boldsymbol{\Delta}_{ij}(k)\right|^{2q}\right|Z\right]\nonumber \\
 & \leq\frac{b(2q)}{t^{2q}p^{2q}}\mathbb{E}\left[\left.\left|\sqrt{\sum_{k=1}^{p}|\boldsymbol{\Delta}_{ij}(k)|^{2}}\right|^{2q}\right|Z\right]\nonumber \\
 & \leq\frac{b(2q)}{t^{2q}p^{2q}}\left(\sum_{k=1}^{p}\mathbb{E}\left[\left.|\boldsymbol{\Delta}_{ij}(k)|^{2q}\right|Z\right]^{1/q}\right)^{q}\nonumber \\
 & =\frac{b(2q)}{t^{2q}p^{q}}\max_{k=1,\ldots,p}\mathbb{E}\left[\left.|\boldsymbol{\Delta}_{ij}(k)|^{2q}\right|Z\right].\label{eq:Delta_deviation}
\end{align}
Re-arranging the expression for $\boldsymbol{\Delta}_{ij}(k)$, applying
the Cauchy-Schwartz inequality and $\|U_{j}\|_{2}=1$, we estimate
\begin{align*}
|\boldsymbol{\Delta}_{ij}(k)| & \leq\left\Vert Y_{k}^{(i)}Y_{k}-\mathbb{E}\left[\left.Y_{k}^{(i)}Y_{k}\right|Z\right]\right\Vert _{2}\|U_{j}\|_{2}\\
 & \leq|Y_{k}^{(i)}|\|Y_{k}\|_{2}+\mathbb{E}\left[\left.|Y_{k}^{(i)}|\|Y_{k}\|_{2}\right|Z\right]
\end{align*}
 and so 
\begin{align}
\mathbb{E}\left[\left.|\boldsymbol{\Delta}_{ij}(k)|^{2q}\right|Z\right] & \leq2^{2q}\mathbb{E}\left[\left.(|Y_{k}^{(i)}|^{2}\|Y_{k}\|_{2}^{2})^{q}\right|Z\right]\nonumber \\
 & =2^{2q}\mathbb{E}\left[\left.\left(\sum_{l=1}^{n}|Y_{k}^{(i)}|^{2}|Y_{k}^{(l)}|^{2}\right)^{q}\right|Z\right]\nonumber \\
 & \leq2^{2q}\left(\sum_{l=1}^{n}\mathbb{E}\left[\left.(|Y_{k}^{(i)}|^{2}|Y_{k}^{(l)}|^{2})^{q}\right|Z\right]^{1/q}\right)^{q}\nonumber \\
 & \leq2^{2q}\left(\sum_{l=1}^{n}\mathbb{E}\left[\left.|Y_{k}^{(i)}|^{4q}\right|Z\right]^{1/(2q)}\mathbb{E}\left[\left.|Y_{k}^{(l)}|^{4q}\right|Z\right]^{1/(2q)}\right)^{q}\nonumber \\
 & =2^{2q}n^{q}\max_{l=1,\ldots,n}\mathbb{E}\left[\left.|Y_{k}^{(l)}|^{4q}\right|Z\right]\nonumber \\
 & =2^{2q}n^{q}\max_{l=1,\ldots,n}\mathbb{E}\left[\left.|X_{k}(Z_{l})+\sigma\mathbf{E}_{kl}|^{4q}\right|Z\right]\nonumber \\
 & \leq2^{6q-1}n^{q}\left(\sup_{l\geq 1}\sup_{z\in\mathcal{Z}}\mathbb{E}\left[|X_{l}(z)|^{4q}\right]+\sigma^{4q}\sup_{l\geq 1,\tilde{l} \geq 1}\mathbb{E}[|\mathbf{E}_{l\tilde{l}}|^{4q}]\right).\label{eq:Delta_moment}
\end{align}
Combining the almost sure upper bounds (\ref{eq:Delta_moment}) and
(\ref{eq:Delta_deviation}), using the tower property of conditional
expectation and then taking a union bound over $i=1,\ldots,n$ and
$j=1,\ldots,r$, we find:
\begin{align*}
 & \mathbb{P}\left(\max_{j=1,\ldots,r}\|(p^{-1}\mathbf{Y}\mathbf{Y}^{\top}-\boldsymbol{\Phi}\boldsymbol{\Phi}^{\top}-\sigma^{2}\mathbf{I}_{n})U_{j}\|_{\infty}\leq t\right)\\
 & \geq1-\frac{n^{1+q}r}{t^{2q}p^{q}}b(2q)2^{6q-1}\left(\sup_{j\geq 1}\sup_{z\in\mathcal{Z}}\mathbb{E}\left[|X_{j}(z)|^{4q}\right]+\sigma^{4q}\sup_{i\geq 1,j\geq 1}\mathbb{E}[|\mathbf{E}_{ij}|^{4q}]\right),
\end{align*}
which completes the proof.
\end{proof}

\section{Supplementary figures for Section~\ref{subsec:Choosing-r}}
\begin{figure}%[tbhp]
    \centering
    \includegraphics[width=\textwidth]{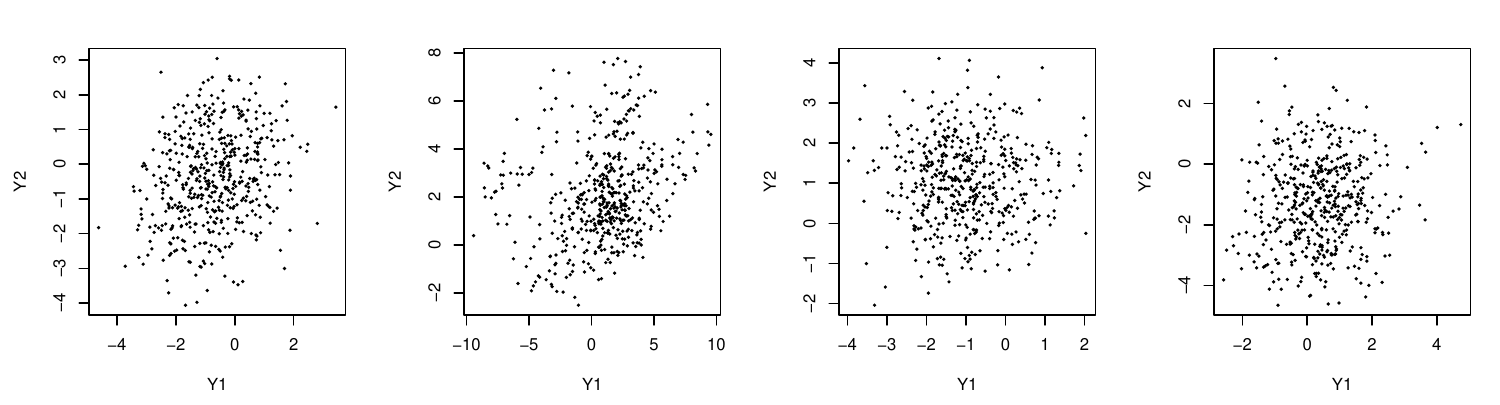}
    \caption{First two coordinates of the data matrices corresponding to figure~\ref{fig:dimension_selection}, showing much less structure than the principal components.}
    \label{fig:indistinguishability}
\end{figure}

\begin{figure}%[tbhp]
    \centering
    \includegraphics[width=\textwidth]{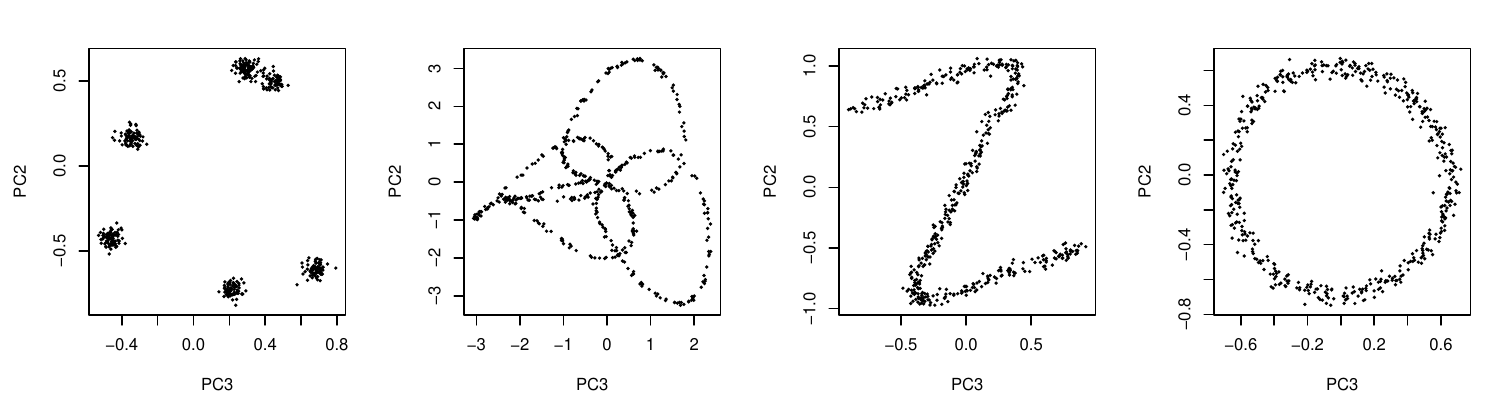}
    \caption{Third and second principal components of the data matrices corresponding to figure~\ref{fig:dimension_selection} (ordered like this to make the resemblance to $\mathcal{Z}$ more obvious).}
    \label{fig:PCA_other_dimensions}
\end{figure}

\begin{figure}%[tbhp]
    \centering
    \includegraphics[width=.5\textwidth]{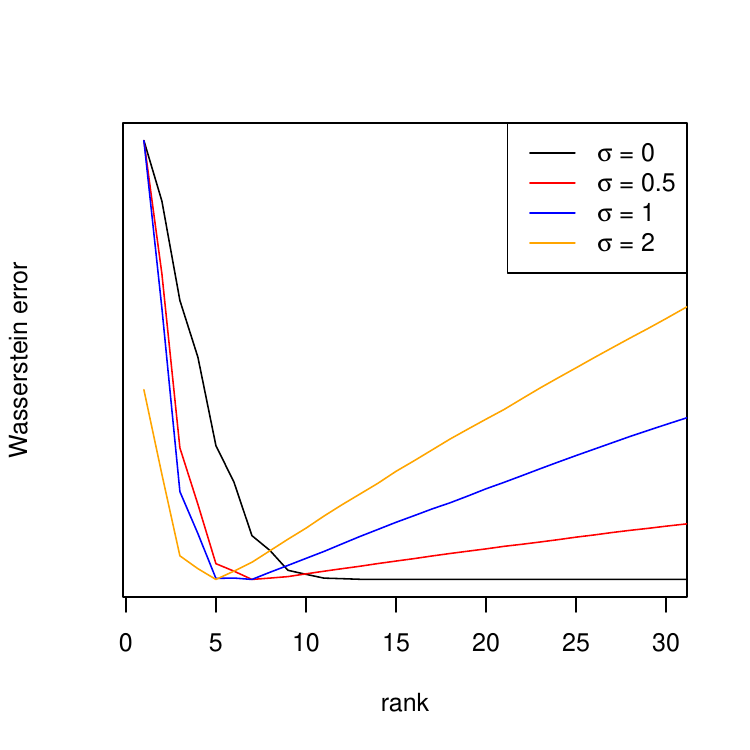}
    \caption{log-Wasserstein error for the fourth configuration in figure~\ref{fig:dimension_selection}, for different error variances. As the variance increases, the optimal dimension (point achieving lowest error) decreases. The curves are shifted and rescaled so that their maxima and minima agree.}
    \label{fig:Wasserstein_varying_sigma}
\end{figure}

\end{document}